\documentclass[a4paper,11pt,twoside]{article}

\usepackage{xcolor}
\usepackage{fontenc}
\usepackage{amsthm,amsmath,amsfonts,amssymb}
\RequirePackage[colorlinks,citecolor=purple,urlcolor=black]{hyperref}%
\RequirePackage{epsfig, graphicx, color}
\RequirePackage{latexsym}

\usepackage{a4wide}
\usepackage{epsf,epsfig,subfigure}
\usepackage{natbib}
\usepackage[utf8]{inputenc}
\usepackage{booktabs}
\usepackage{multirow}
\usepackage{graphics,graphicx}
\usepackage{enumerate}
\usepackage{enumitem}
\usepackage[ruled]{algorithm2e}

\usepackage{subfigure}
\usepackage{bbm}

\usepackage{accents}

\usepackage{color}
\usepackage{dsfont}

\usepackage{chapterbib}

\newtheorem{theorem}{Theorem}

\newtheorem{lemma}[theorem]{Lemma}

\newtheorem{proposition}[theorem]{Proposition}

\theoremstyle{definition}

\newtheorem*{example*gplm}{Example 1}
\newtheorem*{example*gpliv}{Example 2}
\newtheorem*{example*huber}{Example 3}
\newtheorem*{example*survival}{Example 4}
\newtheorem{assumption}{Assumption}

\newtheorem{remark}[theorem]{Remark}

\newcommand{\R}{\mathbb{R}}
\newcommand{\PP}{\mathbb{P}}
\newcommand{\E}{\mathbb{E}}

\newcommand{\sgn}{\mathrm{sgn}}
\newcommand{\iid}{\stackrel{\mathrm{i.i.d.}}{\sim}}

\newcommand{\supp}{\textrm{supp}\,}

\newcommand{\ind}{\mathbbm{1}}
\newcommand{\pr}{\mathbb{P}}

\newcommand{\floor}[1]{\left\lfloor #1 \right\rfloor}

\newcommand{\Cov}{\mathrm{Cov}}
\newcommand{\Corr}{\mathrm{Corr}}
\newcommand{\Var}{\mathrm{Var}}
\newcommand{\tr}{\mathrm{tr}}
\newcommand{\argmin}{\mathrm{argmin}}

\newcommand{\diag}{\textrm{diag}}

\newcommand{\cX}{\mathcal{X}}

\newcommand{\cY}{\mathcal{Y}}

\newcommand{\eval}{\text{eval}}
\newcommand{\Isplit}{\mathcal{I}_{\text{split}}}
\newcommand{\Ieval}{\mathcal{I}_{\text{eval}}}

\newcommand{\cP}{\mathcal{P}}

\newcommand{\cI}{\mathcal{I}}
\newcommand{\supP}{\sup_{P\in\cP}}
\newcommand{\given}{\,|\,}
\newcommand{\biggiven}{\,\big|\,}
\newcommand{\Biggiven}{\,\Big|\,}
\newcommand{\bigggiven}{\,\bigg|\,}

\newcommand{\diam}{\text{diam}\,}

\theoremstyle{plain}

\newcommand{\vertiii}[1]{{\left\vert\kern-0.25ex\left\vert\kern-0.25ex\left\vert #1 
    \right\vert\kern-0.25ex\right\vert\kern-0.25ex\right\vert}}
    
\newcommand\independent{\protect\mathpalette{\protect\independenT}{\perp}}
    \def\independenT#1#2{\mathrel{\rlap{$#1#2$}\mkern2mu{#1#2}}}

\newcommand{\RN}[1]{
  \textup{\uppercase\expandafter{\romannumeral#1}}
}

\makeatletter
\newcommand{\mylabel}[2]{#2\def\@currentlabel{#2}\label{#1}}
\makeatother

\title{Clustered random forests with correlated data for optimal estimation and inference under potential covariate shift}

\date{\today}
\usepackage{authblk}%
\author[1]{Elliot H.\ Young}
\affil{University of Cambridge}
\author[2]{Peter B\"{u}hlmann}
\affil{ETH Z\"{u}rich}

\begin{document}

\maketitle

\begin{abstract}
We develop \emph{Clustered Random Forests}, a random forests algorithm for clustered data, arising from independent groups that exhibit within-cluster dependence. The leaf-wise predictions for each decision tree making up clustered random forests takes the form of a weighted least squares estimator, which leverage correlations between observations for improved prediction accuracy and tighter confidence intervals when performing inference. 
We show that approximately linear time algorithms exist for fitting classes of clustered random forests, matching the computational complexity of standard random forests. 
Further, we observe that the optimality of a clustered random forest, with regards to how optimal weights are chosen within this framework i.e.~those that minimise mean squared prediction error, vary under covariate distribution shift. In light of this, we advocate weight estimation to be determined by a user-chosen covariate distribution, or test dataset of covariates, with respect to which optimal prediction or inference is desired. 
This highlights a key distinction between correlated and independent data with regards to optimality of nonparametric conditional mean estimation under covariate shift. 
We demonstrate our theoretical findings numerically in a number of simulated and real-world settings.
\end{abstract}

\section{Introduction}\label{sec:intro}
Clustered structures are omnipresent in real-world datasets, arising from repeated measurements, longitudinal and geographical studies, or through data augmentation. 
A key challenge in such settings is accounting for potential correlations between intra-cluster components of the response. In (partially) parametric models this is often achieved through an (iteratively re-) weighted least squares estimator for the parameter of interest. Although such approaches perform well when the assumed parametric mean model is correctly specified, model misspecification can introduce substantial bias in the parametric estimates, leading to inconsistent estimation. 

The random forest~\citep{randomforest} is one example of a nonparametric regression method that circumvents the need to impose such a (partially) parametric structure on the data, instead assuming smoothness assumptions on the conditional mean. 
Random forests are constructed by bagging decision trees, each of which can be cast as a two-step procedure: first obtaining a data-driven partition of the covariate support, onto which predictions take the form of an ordinary least squares estimator corresponding to the model of piecewise constant mean with respect to the (high-dimensional) data-driven partition. \citet{meinshausen, biau, wager} among others show that pointwise consistency and asymptotic normality of this estimator follows under Lipschitz smoothness of the mean function, and when forests are constructed using independent, identically distributed data. 
Extending the random forest methodology and theory to optimally account for correlations between observations is however non-trivial. Methodologically one natural  approach, that we adopt, adapts the decision tree algorithm described above so that instead of performing an ordinary least squares regression with respect to a piecewise constant mean `working model' given the data-driven leaves constructed, a weighted least squares regression is performed, with weights that account for correlations between observations. 
A number of existing methods have been introduced that aim to model correlations between data in decision trees and random forests~\citep{REEMTrees, MERF, clustered-rf-review}. 
Such strategies posit a specific, rigid parametric structure on the conditional covariance, such as those implied by a random effects model, and use a likelihood-based loss function to estimate the covariance model. 
As we will see, these approaches can be computationally intractable, and need not even guarantee variance reduction over standard random forests.

We propose a different approach, still estimating predictions via a weighted least squares based decision tree mechanism, but estimating the optimal parameters for estimation in a covariance model free fashion. 
For such a procedure to be computationally competitive with random forests it may be desirable to restrict the class of weights to a restricted, potentially parametric class; we will see that we can achieve approximately linear time clustered random forest fitting for some of the parametric classes we study. The natural question that arises is therefore how to select data-driven optimal weights for such a clustered random forest. 
Here we refer to optimality with respect to minimal mean squared prediction error under a covariate shifted distribution. 
Such covariate distribution shift problems apply naturally to optimal pointwise prediction and inference, in the settings of for example personalized decision making and precision medicine, as well as in supervised learning problems, for example in optimal treatment strategies between geographical areas with different population characteristics, or in natural language processing problems where models may be deployed in settings different to that of the training dataset. 
We observe that with correlated data the optimal clustered random forest weights are a function of the covariate shifted distribution with respect to which the performance benchmark is defined, and that not accounting for these weights in a covariate shift adaptive fashion, as would be the case for e.g.~likelihood based or cross-validation based methods, can cause arbitrarily poor performance under covariate shifted testing. 

\begin{figure}[ht] 
    \centering
    \includegraphics[width=0.85\linewidth]{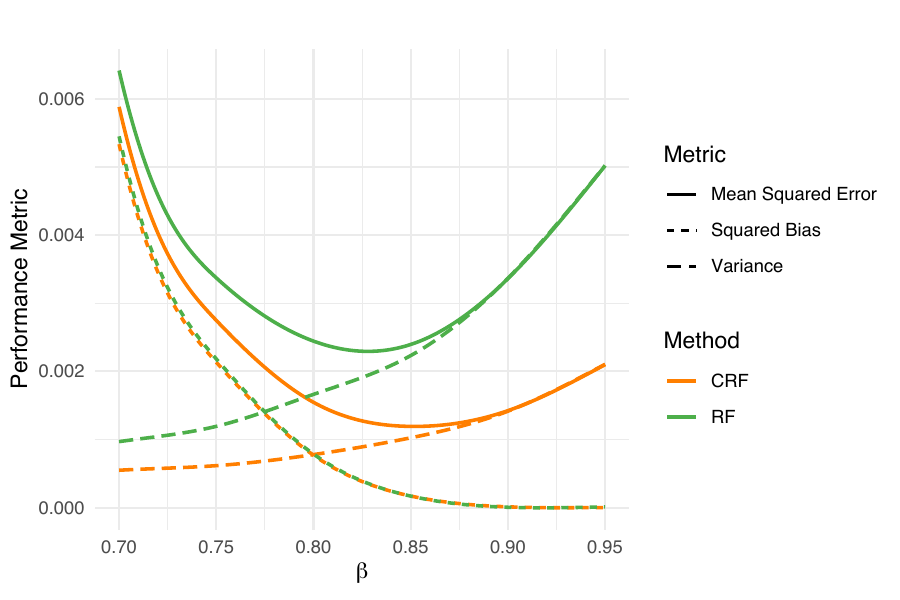}
    \caption{
    Bias--variance tradeoff for random forests, subsampling $s_I=I^\beta$ clusters per tree. 
    Clustered random forests provide a reduction in variance, and by consequence reduction in mean squared error. See Appendix~\ref{appsec:introsim} for further details.}
    \label{fig:sim1}
\end{figure}

\subsection{Our contributions}

The main contributions of our clustered random forest proposal are as follows:

\smallskip
\noindent{\bf Clustered random forest methodology and theory:} We introduce a practical extension of the popular random forest algorithm to clustered data, that in particular provides more accurate predictions by accounting for correlations between observations. We develop asymptotic theory for a flexible class of these clustered random forests. 

\smallskip
\noindent{\bf Minimax optimal rates: }We introduce a class of decision tree splitting rules within which there exists a splitting rule that is minimax rate optimal for pointwise estimation of Lipschitz mean functions. Such results also apply to standard, unclustered honest random forests.

\smallskip
\noindent{\bf Inference with clustered random forests: } A class of these clustered random forests can be shown to be variance-dominating and asymptotically Gaussian, allowing for inference on a broad class of functionals of the conditional mean. Clustered random forests necessarily enjoy (potentially arbitrarily) tighter confidence intervals compared to both standard random forests and pre-existing mixed effects models based random forests.

\smallskip
\noindent{\bf Computational feasibility: } Despite gains that hierarchical or longitudinal random forests may exhibit over standard, unclustered random forests, existing methods have cubic-order computational complexity in the data used to build a decision tree. 
Clustered random forests on the other and are computable in approximately linear time, matching the favorable computational complexity, scalability and parallelizability of standard random forests.

\smallskip
\noindent{\bf Covariate distribution shift with correlated data:} When modeling correlations for improved accuracy in clustered random forests, we see that optimal accuracy in terms of mean squared prediction error (MSPE) can intricately depend on the covariate distribution to which the MSPE is defined.~Consequently we see that likelihood-based and cross-validation based methods can be arbitrarily poor under covariate shift. 
We introduce an important connection between covariate distribution shift problems and correlated data settings. Our results accommodate the covariate shift to be a point mass, thus encompassing optimal pointwise MSE results and tight confidence interval constructions.

\medskip
The remainder of the paper is organised as follows. We begin by briefly reviewing some related literature in Section~\ref{sec:lit-review}, and introduce relevant notation in Section~\ref{sec:notation}. In Section~\ref{sec:methodology} we introduce a general scheme for clustered random forest construction for data driven weights modeling intra-cluster correlations. 
Section~\ref{sec:theory} presents our theoretical results: we derive asymptotic rates for classes of clustered random forests; demonstrate the suboptimalities that likelihood-based and cross-validation weight estimation approaches can succumb to under covariate shifts; and then show asymptotic normality for classes of variance-dominating random forests, from which confidence intervals can be constructed for the pointwise conditional mean. 
Section~\ref{sec:linear-time} concludes by demonstrating the existence of approximately linear time algorithms for constructing these clustered random forests. 
A variety of numerical experiments on simulated and real-world data are explored in Section~\ref{sec:numericals}, demonstrating the effectiveness of our clustered random forest approach. 
The supplementary material contains the proofs of all results presented in the main text, and further details of the numerical experiments. Clustered random forests are implemented in the \texttt{R} package~\texttt{corrRF}\footnote{\url{https://cran.r-project.org/web/packages/corrRF/index.html}}.
Below we review some of the related literature not necessarily covered elsewhere, and collect together some notation used throughout the paper.

\subsection{Some related literature}\label{sec:lit-review}

There is a rich history in acknowledging and modeling correlations between observations for more accurate estimation of parameters of interest, particularly in the setting of parametric multilevel models. Mixed effects models~\citep{membook} have proved popular for hierarchical clustered data in addition to ARMA models for longitudinal data~\citep{box, Brockwell}. In each case a specific parametric structure is imposed on the conditional covariance of the response. Historically these methods were introduced in fully parametric models, where the conditional expectation of the response is also assumed to follow a specific parametric form, although they have also been incorporated into semiparametric models that ask the conditional mean to take a semi or non parametric form. \citet{emmenegger} for example considers a partially linear mixed effects model. Within the context of decision trees (and by extension random forests) random effects expectation maximisation (REEM) trees~\citep{REEMTrees} and mixed effects random forests (MERF)~\citep{MERF} have each been introduced, both of which assume that each data cluster $(Y_i,X_i)\in\R^{n_i}\times\R^{n_i\times d}$, for some $n_i,d\in\mathbb{N}$, follows the semiparametric model
\begin{equation}\label{eq:mixed-effects}
    Y_i\given X_i \sim N_{n_i} \big(
        \mu(X_i)
        \,,\;
        X_i^{\text{(sub)}}\mathcal{V}X_i^{\text{(sub)}\top} + \sigma^2\text{I}_{n_i}
    \big),
\end{equation}
where $\mu$ is a smooth mean function that acts row-wise on the clustered data $X_i$, that satisfies with a slight abuse of notation $(\mu(X_i))_j=\mu(X_{ij})$, and $X_i^{\text{(sub)}}$ is a set of random effects that takes the form of a subset of the covariates that make up $X_i$, and $(\mathcal{V},\sigma)$ is a Euclidean parameter to be estimated. Loosely speaking, both REEM and MERF strategies grow trees by partitioning the covariate support and then fitting a parametric mixed effects model with fixed effects given by a piecewise constant mean working model on the partition, in addition to the random effects covariance structure as in~\eqref{eq:mixed-effects}. Such a model is then high-dimensional in the conditional mean, but relies on a rigid low dimensional structure for the conditional covariance; in Section~\ref{sec:covariate-shift} we will see that this covariance model reliance can be potentially costly with regards to accuracy when misspecified. 

In recent work~\citet{young} consider estimating a linear component in a partially linear model, where they propose an estimation strategy within a constrained class of estimators that is optimal regardless of the form of the true covariance structure. Our work follows a similar robust `covariance model free' ethos, but additionally considers the further relaxed setting of a fully nonparametric conditional mean function. 

Our clustered random forest algorithm naturally relates to a wide literature on random forests~\citep{randomforest}, which have proved popular due to their success when deployed on modern datasets. 
Random forests results have shown~\citep{meinshausen, biau, cevid} and asymptotic normality of variance-dominating random forests~\citep{wager-walther, wager, naf}. Our theoretical results draw on some of these ideas.

One of the contributions of this work highlights a connection between covariate distribution shift problems and correlated data. The setting of optimal estimation in covariate shift problems for independent, uncorrelated data has been studied by~\citet{shimodaira, covshift-book, ma, pathak, kpotufe, koh}. Much less attention has been paid however to correlated data settings, where we see differences in behaviour under covariate shift to the independence settings.

\subsection{Notation}\label{sec:notation}
Given $n\in\mathbb{N}$ we use the shorthand \([n]:=\left\{1,\ldots,n\right\}\). For a symmetric matrix $M\in\R^{m\times m}$ we write \(\Lambda_{\max}(M)\) and \(\Lambda_{\min}(M)\) for its maximum and minimum eigenvalues respectively, and write $\|M\|_1:=\sup_{v\in\R^m:v\neq0}\|v\|_1^{-1}\|Mv\|_1=\max_{j'\in[m]}\sum_{j=1}^m|M_{jj'}|$. 
We write $\mathrm{e}_m\in\R^M$ to be the $m$th unit vector. 
Let \(\Phi\) denote the cumulative distribution function of a standard Gaussian distribution. For our uniform convergence results, it will be helpful to write, for a law $P$ governing the distribution of a random vector \(U\in\R^d\), $\E_P U$ for its expectation and $\pr_P(U \in B)=: \E_P \ind_B (U)$ for any measurable $B \subseteq \R^d$. Further, given a family of probability distributions $\cP$, and for a sequence of families of real-valued random variables \(\left(A_{P,I}\right)_{P\in\cP,I\in\mathbb{N}}\), we write \(A_{P,I}=o_{\mathcal{P}}(1)\) if \(\lim_{I\to\infty}\sup_{P\in\mathcal{P}_I}\PP_{P}\left(\left|A_{P,I}\right|>\epsilon\right)=0\) for all \(\epsilon>0\) and \(A_{P,I}=O_{\mathcal{P}}(1)\) if for any \(\epsilon>0\) there exist \(M_{\epsilon}, I_{\epsilon}>0\) such that \(\sup_{I\geq I_{\epsilon}}\sup_{P\in\mathcal{P}_I}\PP_{P}\left(\left|A_{P,I}\right|>M_{\epsilon}\right)<\epsilon\). Further, we write $A_{P,I}=o_{\mathcal{P}}(f(I))$ and $A_{P,I}=O_{\mathcal{P}}(f(I))$ for a given function $f:(0, \infty) \to (0, \infty)$ if $f(I)^{-1}A_{P,I}=o_{\mathcal{P}}(1)$ and $f(I)^{-1}A_{P,I}=O_{\mathcal{P}}(1)$ respectively.

\section{Clustered random forests}\label{sec:methodology}
In this section we outline the details of our clustered random forest algorithm. We introduce the following notation. Suppose we have access to a dataset consisting of $I\in\mathbb{N}$ clusters $(Y_i,X_i)\in\R^{n_i}\times\cX^{n_i}$, with $i$th cluster of size $n_i\in\mathbb{N}$. We assume that the conditional mean acts componentwise, i.e.
\begin{equation}\label{eq:mean-model}
\E_P[Y_{ij}\given X_i] = \mu(X_{ij}),
\end{equation}
for each $i\in[I]$ and $j\in[n_i]$, and where the mean function $\mu$ is the target of estimation. If for example $x\in\cX=\R^d$ for some $d\geq2$ and $\mu(x)=x_1\beta_P+f_P(x_2,\ldots,x_d)$ for some function $f_P:\R^{d-1}\to\R$ and $\beta_P\in\R$ then~\eqref{eq:mean-model} takes the form of the grouped partially linear model as studied in~\citet{emmenegger,young}. 

We first introduce our clustered decision trees in Algorithm~\ref{alg:cdt}. Our clustered random forest procedure, outlined in Algorithm~\ref{alg:crf}, is then presented, which takes the form of an aggregation of a number of clustered decision trees generated using a bootstrapped subsample of our clustered data. We outline a few key details of these two algorithms.

\begin{algorithm}[H]
\KwIn{Grouped dataset $(Y_i,X_i)_{i\in\cI_{\text{split}}\cup\cI_{\text{eval}}\cup\cI_{\text{corr}}}$, with groups $(Y_i,X_i)\in\R^{n_i}\times\cX^{n_i}$ indexed by the (pre-partitioned) set $\cI_{\text{split}}\cup\cI_{\text{eval}}\cup\cI_{\text{corr}}$; cluster-wise weight matrix structure $(W_i(\rho))_{\rho\in\Gamma}$; optional test covariate distribution or dataset; hyperparameters for tree splitting.}

    {\bf Split Construction:} Build a decision tree using  the data indexed by $\cI_{\text{split}}$ (with splitting criterion satisfying Assumption~\ref{ass:tree} e.g.~the CART criterion). 
    Let $M$ be the number of nodes for this tree, with leaves denoted by $L_1,\ldots,L_M$. For each $x\in\cX$ let $J(x)$ be the unique element of $\{m\in[M]:x\in L_m\}$.

    {\bf Weight Estimation:} Calculate weight parameter $\hat{\rho}$:

    \For{$i\in\cI_{\mathrm{corr}}$} {

    Calculate $\tilde{\chi}_i\in\{0,1\}^{n_i\times M}$ as $(\tilde{\chi}_i)_{j,m} = \ind_{(X_{ij}\in L_m)}$.

    Calculate $\tilde{\varepsilon}_{ij} := Y_{ij} - \mathrm{e}_{J(X_{ij})}^\top (\sum_{i\in\cI_{\text{corr}}}\tilde{\chi}_i^\top \tilde{\chi}_i)^{-1}(\sum_{i\in\cI_{\text{corr}}}\tilde{\chi}_i^\top Y_i)$.

    }

    Calculate weight parameter $\hat{\rho}$ using data $(\tilde{\varepsilon}_i)_{i\in\cI_{\text{corr}}}$ (alongside potentially a user-chosen test covariate distribution/ dataset). 
    For a $Q$-covariate shift optimised clustered random forest, take e.g.~$\hat\rho$ the minimiser of the loss function~\eqref{eq:rho-testMSPE}. 

    {\bf Evaluation on nodes: }Construct tree predictions node-wise:
    
    \For{$i\in\cI_{\mathrm{eval}}$} {
        Calculate $\chi_i\in\{0,1\}^{n_i\times M}$ as $(\chi_i)_{j,m}=\ind_{(X_{ij}\in L_m)}$.
    }
    
    $\hat{T}(x) := \mathrm{e}_{J(x)}^\top \big(\sum_{i\in\cI_{\text{eval}}}\chi_i^\top W_i(\hat{\rho})\chi_i\big)^{-1}\big(\sum_{i\in\cI_{\text{eval}}}\chi_i^\top W_i(\hat{\rho})Y_i\big)$.

\KwOut{Clustered decision tree estimate $\hat{T}:\cX\to\R$.}
\caption{Clustered Decision Tree (CDT)}
\label{alg:cdt}
\end{algorithm}

\begin{algorithm}[H]
\KwIn{Grouped dataset $(Y_i,X_i)_{i\in\cI}$, with groups $(Y_i,X_i)\in\R^{n_i}\times\cX^{n_i}$ indexed by the set $\cI:=[I]$; subsampling quantities $(s_I,s_I^{\mathrm{corr}})$; clustered decision tree inputs (see Algorithm~\ref{alg:cdt}); number of trees $B\in\mathbb{N}$; number of little bags for variance estimation $R\in\mathbb{N}$.}

\For{$r \in [R]$} {
Select a random subset $\cI_r\subset\cI$ of the index set of size $\floor{I/2}$.

\For{$b\in[B]$} {
    Select disjoint random subsets $\cI_{r,b,\text{split}},\cI_{r,b,\text{eval}},\cI_{r,b,\text{corr}}\subset\cI_r$ of sizes $s_I,s_I$ and $s_I^{\mathrm{corr}}$ respectively. 

    Construct a clustered decision tree $\hat{\mu}_{r,b}:\cX\to\R$ using Algorithm~\ref{alg:cdt} for the data partition $(\cI_{r,b,\text{split}},\cI_{r,b,\text{eval}},\cI_{r,b,\text{corr}})$.
}

$\hat{\mu}_r(\cdot) := \frac{1}{B}\sum_{b=1}^B \hat{\mu}_{r,b}(\cdot)$.

}

    $\hat{\mu}(\cdot) := \frac{1}{R}\sum_{r=1}^R \hat{\mu}_r(\cdot)$.

    $\hat{V}(\cdot) := \frac{1}{R}\sum_{r=1}^R \big(\hat{\mu}_r(\cdot)-
    \hat{\mu}(\cdot)\big)^2$.

    $\hat{C}_{\tilde{\alpha}}(\cdot):=\big[\hat{\mu}(\cdot)-\Phi^{-1}\big(1-\frac{\tilde{\alpha}}{2}\big)\hat{V}^{\frac{1}{2}}(\cdot),\;\hat{\mu}(\cdot)+\Phi^{-1}\big(1-\frac{\tilde{\alpha}}{2}\big)\hat{V}^{\frac{1}{2}}(\cdot)\big]$.

\KwOut{Clustered random forest $\hat{\mu}:\cX\to\R$, estimator of its pointwise variance $\hat{V}:\cX\to\R$, and pointwise confidence intervals $\hat{C}_{\alpha}$ for $\mu$.}
\caption{Clustered Random Forests (CRF)}
\label{alg:crf}
\end{algorithm}

\newpage
\noindent{\bf Weight classes to model intra-cluster correlation:} 

An important distinguishing feature of the clustered random forest algorithm is the added flexibility resulting from our weighting scheme, that in turn models the intra-cluster correlations between observations for more accurate estimation of the mean function $\mu$.

Algorithm~\ref{alg:cdt} involves a user-chosen parametric class of cluster-wise weight matrices $\{W_i(\rho):\rho\in\Gamma\}$ (where by default we take $\rho=0$ to correspond to $W_i(0)=I_{n_i}$). The restriction to specifically a parametric class allows for computationally simple weight estimation, in particular in light of our desire for minimal additional computational burden over standard random forests. Our numerical experiments (see Section~\ref{sec:numericals} to follow) explore two concrete examples of possible weight classes, which are also implemented in the \texttt{corrRF} package: the equicorrelated (exchangeable); and $\text{AR}(1)$ weight structures, which are both popular weighting methods in parametric clustered models (such as mixed effects models and generalised estimating equations) and take the forms
\begin{equation*}
    W_i(\rho)=\big(\big(\ind_{\{j=k\}}+\rho\cdot\ind_{\{j\neq k\}}\big)_{j,k\in[n_i]}\big)^{-1}
    ,
    \quad\text{and}\quad
    W_i(\rho)=\big(\big(\rho^{|j-k|}\big)_{j,k\in[n_i]}\big)^{-1},
\end{equation*}
respectively. These two structures are shown to allow for approximately linear time (in the number of observations used to build a tree) fitting of clustered decision trees (see Section~\ref{sec:linear-time} to follow). Whilst our theoretical results are shown to hold for classes of weight structures beyond these two examples, we treat these two examples as popular running examples throughout. 

\medskip
\noindent{\bf Estimation of weights:} 

Given a user-specified weight matrix class (parametrised by some Euclidean parameter $\rho\in\Gamma$) we will see that the estimator $\hat{\rho}$ that defines these weights may be estimated via any method that need only converge at logarithmic rates to a deterministic weight parameter $\rho^*\in\Gamma$, with optimality under potentially covariate shifted MSPE controlled by the covariate shift integrated variance of an individual $\rho^*$-weighted tree (see Theorem~\ref{thm:minimax}). 

There are a number of natural strategies, not necessarily optimal, one may choose to estimate weights; we discuss some of these methods and highlight their drawbacks. One natural approach would be to fit a random effects model on the residual error terms $(\tilde{\varepsilon}_{ij})_{i\in\cI_{\text{corr}},j\in[n_i]}$ (e.g.~for the residuals defined in Algorithm~\ref{alg:cdt}). This follows the ethos of~\citep{REEMTrees, MERF}, which builds on standard parametric theory in multilevel modeling~\citep{membook, fahrmeir, glmmbook}. An alternative strategy could be to estimate weights that minimise an empirical estimator of the mean squared training error of a clustered decision tree weighted by the parameter $\rho$. Adopting the notation of Algorithm~\ref{alg:cdt}, a simple empirical estimator of the integrated variance over the training covariate distribution is given by

\begin{multline}\label{eq:TRAIN}
    \underset{\rho\in\Gamma}{\argmin}\,\tr\bigg\{\Big(\sum_{i\in\cI_{\text{corr}}}\chi_i^\top \chi_i\Big)\Big(\sum_{i\in\cI_{\text{corr}}}\chi_i^\top W_i(\rho)\chi_i\Big)^{-1}
    \\
    \Big(\sum_{i\in\cI_{\text{corr}}}\chi_i^\top W_i(\rho)\tilde{\varepsilon}_i\tilde{\varepsilon}_i^\top W_i(\rho)\chi_i\Big)
    \Big(\sum_{i\in\cI_{\text{corr}}}\chi_i^\top W_i(\rho)\chi_i\Big)^{-1}\bigg\}.
\end{multline}
We will see later in Section~\ref{sec:covariate-shift} that (see Theorem~\ref{thm:dist-shift-2}) both of these two strategies can select arbitrarily suboptimal MSPE under covariate shift. Because clustered random forest performance is controlled by the individual tree-level covariate shift integrated variance (Theorem~\ref{thm:minimax}) a natural optimal approach would be to directly target this quantity. Given a covariate shift distribution $Q$, consider the empirical estimator of this quantity (again adopting the notation of Algorithm~\ref{alg:cdt}) via
\begin{multline}\label{eq:rho-testMSPE}
    \underset{\rho\in\Gamma}{\argmin}\, \tr\bigg\{ \diag\Big(\big(Q(L_m)\big)_{m\in[M]}\Big)
    \Big(\sum_{i\in\cI_{\text{corr}}}\chi_i^\top W_i(\rho)\chi_i\Big)^{-1}
    \\
    \Big(\sum_{i\in\cI_{\text{corr}}}\chi_i^\top W_i(\rho)\tilde{\varepsilon}_i\tilde{\varepsilon}_i^\top
    W_i(\rho)\chi_i\Big)\Big(\sum_{i\in\cI_{\text{corr}}}\chi_i^\top W_i(\rho)\chi_i\Big)^{-1} \bigg\},
\end{multline}
where $Q(L_m)=\int_{L_m}dQ$ may be replaced by an estimator for this quantity e.g.~the empirical frequency using a new testing dataset consisting only of covariates drawn from $Q$. 
Theory for this empirical estimator~\eqref{eq:rho-testMSPE} is derived in Section~\ref{sec:theory}. 
Alternatively $Q$ can be directly taken as the empirical distribution of the covariates a test set, and so it may contain only a few unlabeled test points. 

\medskip
\noindent{\bf Honesty property for decision trees:}

The clustered decision trees of Algorithm~\ref{alg:cdt} satisfy a principle termed as honesty, following the naming scheme of~\citet{wager}, in that for each individual decision tree a single cluster is used in determining at most one of the decision tree splits, working correlation (weight structure) estimation, and predictions (evaluations); in practice this is achieved through a tree-wise sample splitting of the subsampled data for each tree. As we see in Section~\ref{sec:theory} this allows for guaranteed theoretical results such as asymptotic normality of the pointwise mean function estimator and functionals of the mean. In generality however, an analogous clustered random forest could easily be introduced that does not satisfy honesty (obtained by omitting the tree-wise sample split step) and that therefore more closely mimics the original random forests of~\citet{randomforest}; see Appendix~\ref{appsec:dishonest-rf} for more details, which also allows for instance the option to treat the weight parameters as additional hyperparameters to be tuned via cross validation.

\medskip
\noindent{\bf Decision tree splitting rule:}

The clustered decision tree Algorithm~\ref{alg:cdt} also notably allows for an arbitrary user-chosen splitting rule. Whilst it is likely that different splitting mechanisms will differ in behaviour in practice, the asymptotic results for clustered random forests we present only require the splitting rule used to satisfy certain minimal properties (see Assumption~\ref{ass:tree} to follow). For these reasons, the precise choice of splitting rule here we leave to further work. As a default in our numerical simulations we will use the CART criterion of~\citet{randomforest}, in part due to its popularity and computationally attractive properties in practice.

\medskip
\noindent{\bf Subsampling to generate random forests:}

Algorithm~\ref{alg:crf} constructs clustered random forests, which aggregates a number of clustered decision trees generated with subsamples of the clustered data. 
In practice, we aggregate $R\times B$ clustered decision trees, where the subsampling process is structured into two sample splits; for each $r\in[R]$ a subsample of the data (of size approximately half the full data) is taken, and then $B$ decision trees are constructed from a further subsampling, in which a $s_I$ clusters are used in the evaluation step for each tree. The rate at which $s_I$ diverges will dictate the convergence rate of the final random forests (see Section~\ref{sec:theory}). The former part of this subsampling mechanism, which does not appear in the random forests of~\citet{randomforest}, is introduced to construct a so-called `bootstrap of little bags' estimator of the variance~\citep{sexton, GRF, naf}, allowing for confidence interval construction of the mean function for variance dominating random forests (see Theorem~\ref{thm:normality}). Note however if this estimator for the variance is not desired, for example if one merely wishes to obtain a more accurate estimator for the mean function (see Theorem~\ref{thm:minimax}), then Algorithm~\ref{alg:crf} can simply be taken with $R=1$ and $\cI_1=\cI$.

\section{Asymptotic theory}\label{sec:theory}

Here we present asymptotic theory of our clustered random forests algorithm, in a regime where the number of clusters $I$ diverges. For simplicity of exposition we will present our results in the case where the group size $n_i\equiv n_{\mathrm{c}}$ is constant across clusters, though our results extend to differing group sizes. 
First, we present an asymptotic upper bound on the general $Q$-covariate shifted MSPE, and thus by corollary an upper bound on the pointwise MSE (that for certain splitting rules coincide with the minimax rate). We then study optimality with respect to $Q$-covariate shifted MSPE to motivate the $Q$-adjusted weight estimation method of Algorithm~\ref{alg:crf}. Finally, we prove uniform asymptotic Gaussianity results for variance dominating clustered random forests, which in turn can be used to show asymptotic coverage of the pointwise $\alpha$-level confidence intervals $\hat{C}_\alpha(x)$ presented in Algorithm~\ref{alg:crf}. 
We conclude by discussing the linear time computational comparability of clustered random forests to standard random forests. 
Proofs of all the results presented in this section are given in Appendix~\ref{appsec:theory}.

Our results build on results for random forests for unclustered data~\citep{wager, naf}, which themselves build on consistency results of~\citet{biau, meinshausen, wager-walther, wager, cevid}. In particular, our extensions to clustered data involves a number of additional non-trivial technicalities. To indicate how these arise observe that the predictor at a point $x\in\cX$ of a standard decision tree is a function of only the local data in the leaf $L(x)$ of that tree. In contrast a clustered decision tree's prediction at $x\in\cX$ is also a function of (potentially) all the data, due to the inclusion of information from not only the local datapoints contained in the leaf $L(x)$ but also the data correlated with observations in $L(x)$. A number of common techniques used to analyse random forests in the aforementioned references consequently don't follow in correlated data settings, and must be adapted to accommodate these extra intricacies; see Appendix~\ref{appsec:theory} for further details. 

Throughout this section, our theoretical results are presented for `Monte-Carlo limit' random forests~\citep{wager, naf} constructed from trees generated via all possible subsamplings of the data, and take the form 
\begin{equation}\label{eq:MC}
    \hat{\mu}_I^{\text{MC}}(x) := 
    \binom{I}{\bar{s}_I}^{-1}\sum_{i_1<\cdots<i_{\bar{s}_I}}\E_{\xi\sim\Xi}\bigl[T(x,\xi;Z_{i_1},\ldots,Z_{i_{\bar{s}_I}})\,\big|\, Z_{i_1},\ldots,Z_{i_{\bar{s}_I}}\bigr],
\end{equation}
where $\bar{s}:=3s_I$, and where $T(x,\xi;Z_{i_1},\ldots,Z_{i_{\bar{s}}})$ denotes the prediction at the point $x\in\cX$ of a decision tree $T$ constructed from the data $\{Z_{i_1},\ldots,Z_{i_{\bar{s}}}\}$, with $Z_i:=(Y_i,X_i)$ and $\xi\sim\Xi$ denoting a random variable that encapsulates the randomness in the decision tree splitting procedure. 
In practice our clustered random forests (Algorithm~\ref{alg:crf}) approximate their `Monte-Carlo limit' for a sufficiently large $B\times R$, and in our empirical experiments of Section~\ref{sec:sim2} we take $(B,R)=(500,1)$, which suffices as a sufficient approximation to the Monte-Carlo limit (see Appendix~\ref{appsec:sim2-B}). 
The following assumptions are made on the tree splitting mechanism and random forest construction of Algorithm~\ref{alg:crf}.

\begin{assumption}[Forest properties]\label{ass:tree}
Suppose each tree of Algorithm~\ref{alg:crf} satisfies the following properties:
\begin{itemize}
    \item[(\mylabel{honesty}{A1.1})] (\emph{Honesty}) Each group-wise datapoint $(Y_i,X_i)_{i\in\cI}$ used to build a single tree is used to determine at most one of the splits of the tree, the `random effects' $\hat{\rho}$, the `fixed effects' evaluation (see Algorithm~\ref{alg:crf} for a splitting technique that enforces this property).
    \item[(\mylabel{symmetry}{A1.2})] (\emph{Symmetry}) The distribution of the randomised tree splits are independent of the ordering of the training samples, for orderings preserving the group structure i.e.~interchanging the order of clusters and interchanging intra-cluster observations. 
    \item[(\mylabel{regularity}{A1.3})] (\emph{Regularity}) Each tree split leaves a fraction $\alpha\in(0,0.5]$ of the available training data in each child node. Further, trees are grown until every leaf contains between $k_I$ and $2k_I-1$ observations for some $k_I\in\mathbb{N}$. \label{ass:k}
    \item[(\mylabel{feature}{A1.4})] (\emph{Feature splits}) Each path from the root node to a terminal node (leaf), and for each covariate feature $j\in[d]$, a proportion $d^{-1}\pi$ of the tree splits along the path are along the $j$th feature space almost surely, for some fixed $\pi\in(0,1]$.
    \item[(\mylabel{beta}{A1.5})] (\emph{Subsampling sizes}) The subsampling sizes $s_I,s_I^{\mathrm{corr}}\to\infty$ as $I\to\infty$. 
\end{itemize}
\end{assumption}

These restrictions to the splitting procedure are commonplace in pre-existing results on distributional limits of random forests~\citep{wager,GRF,naf}. 
Our assumption~\ref{feature} differs to~\citet{wager} who assume each tree split is on the $j$th feature with probability at least $d^{-1}\pi$, although notably the two assumptions are equivalent in the $d=1$ case. 
This however allows us to strengthen our theoretical results. Whilst the aforementioned requires require $\alpha\in(0,0.2]$ in Assumption~\ref{regularity} we may accommodate $\alpha\in(0,0.5]$. Subsequently we obtain faster rates that approach minimax rates for certain splitting rules (see Remark~\ref{remark:minimax} to follow), and we obtain variance dominating, asymptotically normal random forests for smaller subsampling regimes, and thus allowing us to form confidence intervals for the pointwise conditional mean function with widths that converge at faster rates. 
Assumption~\ref{ass:tree} does restrict the splitting procedures that we provide our theoretical guarantees over, but the above assumptions can be enforced for popular splitting criteria such as CART, provided the candidate splits being assessed are restricted so as not to include sufficiently extreme splits at each iteration.

Our theoretical results are also shown to hold for a broad class of weight functions. 
It will help to define for symmetric matrices $M\in\R^{n\times n}$ the diagonally dominant function
    \[
    \zeta(M) := \min_{j\in[n]}\biggl(M_{jj}-\sum_{j'\in[n]\setminus\{j\}}|M_{jj'}|\biggr).
    \]

\begin{assumption}[Weight estimation]\label{ass:weights}
{\color{white}.}
\begin{enumerate}
    \item[(\mylabel{weights}{A2.1})] (\emph{Admissible weight classes}) 
    The weight class $(W_i(\rho))_{\rho\in\Gamma}$ defined over the compact set $\Gamma$ in Algorithm~\ref{alg:cdt} is symmetric and satisfies
    \begin{equation*}
        \sup_{\rho\in\Gamma}\|W_i(\rho)\|_1 \leq C_W
        ,\quad
        \inf_{\rho\in\Gamma}\zeta(W_i(\rho))\geq c_W
        ,\quad\text{and}\quad
        \sup_{\rho\in\Gamma}\Lambda_{\max}\bigg(\bigg\{\frac{\partial W_i(\rho)}{\partial\rho}\bigg\}^2\bigg)\leq C_W',
    \end{equation*}
    for finite constants $c_W,C_W,C_W'>0$ and for every $i\in[I]$.
     \item[(\mylabel{weights-rates}{A2.2})]  (\emph{Weight estimation}) Suppose that $\hat\rho=\rho^*+o_{\mathcal{P}}\bigl((\log I)^{-d}\bigr)$ of some deterministic weight limit $\rho^*$.
\end{enumerate} 
\end{assumption}
Assumption~\ref{weights} is satisfied for two common diagonally dominant weight structures used in multilevel modeling: the inverse equicorrelated structure with finite cluster size; and the inverse $\text{AR}(1)$ structure with potentially diverging cluster size (see Appendix~\ref{appsec:equi-ar1}), both of which are implemented in~\texttt{corrRF}. The second condition of~\ref{weights} implies a lower bound on the minimum eigenvalue $\Lambda_{\min}(W_i(\rho))\geq c_W$ via the Gershgorin circle theorem. 

We also make the following assumptions on the training data generating mechanism.

\begin{assumption}[Data generating procedure]\label{ass:data}
    Let $\mathcal{P}$ be a class of distributions on i.i.d.~clusters $(Y_i,X_i)$ with the properties that for each $P\in\mathcal{P}$ the following hold:
    \begin{itemize}
    
    \item[(\mylabel{bounded-density}{A3.1})] (\emph{Covariate Density}) $\cX=[0,1]^d$ with i.i.d.~cluster-level covariates $(X_i)_{i\in[I]}$, with each entry $X_{ij}$ admitting a marginal density $f$ with respect to Lebesgue measure, satisfying $\nu^{-1}\leq f(x)\leq \nu$ for some $\nu\geq 1$,  and either:~(i) covariates $(X_{ij})_{j\in[n_i]}$ are i.i.d and we allow the cluster size $n_i$ to diverge with rates $\log n_i = O(\log I)$; or
    (ii) $X_{ij}=X_{ij'}$ almost surely for all $j,j'\in[n_i]$, $n_i$ is finite, and $k/n_i\in\mathbb{N}$.
     \item[(\mylabel{lip}{A3.2})] (\emph{Lipschitz continuity}) The functions $x \mapsto \E_P[Y_{ij}\given X_{ij}=x] =: \mu(x)$ and $(x,x') \mapsto \Cov_P(Y_{ij},Y_{ik}\given X_{ij}=x, X_{ik}=x')$ 
    are bounded and Lipschitz continuous with Lipschitz constants $L_\mu$ and $L_v$ respectively.
    \item[(\mylabel{bounded-sigma}{A3.3})] (\emph{Bounded covariance}) The covariance matrix $\Cov_P(Y_i\given X_i)$
    has eigenvalues in the interval $[\lambda^{-1},\lambda]$ almost surely for some $\lambda\geq1$.
    \item[(\mylabel{4+delta}{A3.4})] (\emph{Bounded moments}) There exists some $\delta,\tau>0$ such that for any $a\in\R^{n_i}$ and $x\in[0,1]^d$ we have $\E_P\big[|a^T(Y_{i}-\mu(x))|^{2+\delta}\given X_{i}=x\big]\leq\tau\,\E_P[\{a^T(Y_{i}-\mu(x))\}^2\given X_{i}=x]^{1+\frac{\delta}{2}}$. 
    \end{itemize}
\end{assumption}

The condition~\ref{bounded-density} allows for an extreme case when the joint distribution of the cluster-level covariate $X_i$ is degenerate, which occurs for example when one has repeated measures with the same cluster-level covariates. 
Equipped with Assumption~\ref{ass:data}, we may prove the following bound on a general $Q$-covariate shifted MSPE
\begin{equation*}
    \mathrm{MSPE}_Q\bigl(\hat\mu_I^{\mathrm{MC}}\bigr)
    := 
    \int_\cX \bigl(\hat\mu_I^{\mathrm{MC}}(x)-\mu(x)\bigr)^2\,dQ(x).
\end{equation*}

\begin{theorem}\label{thm:minimax}
    Let $\hat\mu_I^{\mathrm{MC}}$ be a clustered random forest estimator~\eqref{eq:MC} satisfying~\ref{ass:tree}~and~\ref{ass:weights}, trained on $I$ clusters drawn from a law $P$ satisfying~\ref{ass:data}, 
    with $k_I^{-1}I^{-1}s_I\to0$ and $k_I^{-1}n_{\mathrm{c}}s_I\to\infty$ as $I\to\infty$. 
    Also fix an arbitrarily small constant $\epsilon\in(0,1)$ and covariate distribution $Q$, there exists some $I_0\in\mathbb{N}$ that does not depend on $(P,Q)$ such that, for all $I\geq I_0$,
   \begin{equation*}
        \E_P\bigl[\mathrm{MSPE}_Q\bigl(\hat\mu_I^{\mathrm{MC}}\bigr)\bigr]
        \leq
        C_{\mathrm{bias}}^2 \biggl(\frac{2k_I}{n_{\mathrm{c}}s_I}\biggr)^{2(1-\epsilon)\frac{\log((1-\alpha)^{-1})}{\log(\alpha^{-1})}\frac{\pi}{d}} +\, \mathcal{V}_Q(\rho^*) \biggl(\frac{s_I}{k_I I}\biggr),
    \end{equation*}
    where $C_{\mathrm{bias}}:=\frac{6C_WL_\mu d^{1/2}\nu}{c_W}$, and 
    \begin{equation*}
        \mathcal{V}_Q(\rho^*) := k_I\int_\cX\Var_P\bigl(T_{\rho^*}(x)\bigr)\,dQ(x),
    \end{equation*}
    is the $Q$-integrated variance of an individual clustered decision tree $T_{\rho^*}(x)$ with fixed correlation weighting parameter $\rho^*$. 
\end{theorem}

The two terms in the above bound correspond to the integrated squared bias and variance respectively, and dictate the bias-variance tradeoff with respect to the parameters $(s_I,k_I)$. 
Taking $Q$ as the marginal law of the training covariate distribution, the above is a bound on the mean squared training error. When $Q=\delta_x$ is a point mass at some $x\in\cX$ we obtain pointwise MSE convergence rates
\begin{equation*}
    \E_P\Big[\big(\hat{\mu}_I^{\mathrm{MC}}(x)-\mu(x)\big)^2\Big]
    \leq
    C_{\mathrm{bias}}^2 \biggl(\frac{2k_I}{n_{\mathrm{c}}s_I}\biggr)^{2(1-\epsilon)\frac{\log((1-\alpha)^{-1})}{\log(\alpha^{-1})}\frac{\pi}{d}} +\, \mathcal{V}_{\delta_x}(\rho^*) \biggl(\frac{s_I}{k_I I}\biggr),
\end{equation*}
where $\mathcal{V}_{\delta_x}(\rho^*)=k_I\Var_P\bigl(T_{\rho^*}(x)\bigr)$
Moreover, the quantity~$\mathcal{V}_Q(\rho)$ converges to some finite non-zero deterministic limit as $I$ diverges, that in particular does not depend on $(s_I,k_I)$; the pointwise variance of a tree with order $k_I$ observations per leaf scales as $k_I^{-1}$. Because $I_0$ in Theorem~\ref{thm:minimax} does not depend on $P\in\cP$ this result is a uniform result in $P\in\cP$.

Theorem~\ref{thm:minimax} also holds for standard random forests (with cluster-wise subsampling) applied to clustered data by taking $\hat\rho=\rho^*=0$ in the above. Thus, if a clustered random forest with weights $\rho^*$ satisfy $\mathcal{V}_Q(\rho^*)\leq V_Q(0)$ then the clustered random forest improves on the standard random forest regardless of the choice of hyperparameters $(s_I,k_I)$. Whilst not necessarily optimally tuned, if a user had access to a standard random forest already fit and hyperparameter tuned, then asymptotic improvements are achieved by adding in the weight structure of clustered random forests without having to re-tune the other hyperparameters in a `greedy' final post-hoc evaluation step.

\subsection{Optimal weights under covariate distribution shift}\label{sec:covariate-shift}

The $Q$-covariate shifted MSPE in Theorem~\ref{thm:dist-shift-2} is controlled monotonically by $\mathcal{V}_Q(\rho^*)$ 
Thus, given a target covariate distribution~$Q$, we are interested in obtaining the minimiser of the function $\rho\mapsto\mathcal{V}_Q(\rho)$. The following theorem shows that this minimiser can depend critically on the covariate shift $Q$. 

\begin{theorem}[Covariate shift]\label{thm:dist-shift-2}
    Consider a clustered random forest with properties~\ref{ass:tree}~and~\ref{ass:weights} with equicorrelated weight structure, trained on $I$ clusters drawn from law satisfying~\ref{ass:data} with $d=1$, and 
    with leaf size~$k_I^{-1}=o(1)$, $k_I=o\bigl(s_I^{2\phi/(2\phi+1)}\bigr)$ where $\phi:=(1-\epsilon)\frac{\log((1-\alpha)^{-1})}{\log(\alpha^{-1})}\pi$. 

    Fix an arbitrarily small $\varrho\in(0,1)$, and arbitrarily large $\kappa\geq1$. 
    Then there exists a distribution $P$ on the training data satisfying Assumption~\ref{ass:data} with
    $
        \Corr_P(Y_{ij},Y_{ij'}\given X_i) = \varrho,
    $
    almost surely for $j\neq j'$, and there exists covariate shifted distributions $Q_1$, $Q_2$ such that for sufficiently large $I$,
    \begin{equation*}
        \frac{\mathcal{V}_{Q_1}(\rho_{Q_2\text{-}\mathrm{opt}})}{\mathcal{V}_{Q_1}(\rho_{Q_1\text{-}\mathrm{opt}})}
        \geq\kappa,
    \end{equation*}
    where
    \begin{equation*}
        \rho_{Q_j\text{-}\mathrm{opt}} := \underset{\rho\in[0,1)}{\argmin}
        \,\mathcal{V}_{Q_j}(\rho),
        \quad
        j\in\{1,2\}.
    \end{equation*}
\end{theorem}

Essentially Theorem~\ref{thm:dist-shift-2} shows that for correlated data $\argmin_{\rho\in\Gamma}\mathcal{V}_Q(\rho)$ is a function of $Q$, and thus the optimal weights vary depending on covariate shift; an empirical realisation of this observation is given in Section~\ref{sec:sim2} (see also Figure~\ref{fig:sim2}). Moreover, the optimal clustered random forest under one covariate distribution $Q_1$ can be arbitrarily suboptimal under another covariate distribution $Q_2$. The result of Theorem~\ref{thm:dist-shift-2} holds when restricting $Q_2$ to be the marginal distribution of $X$ under $P$, in which case $\rho_{Q_2\text{-}\mathrm{opt}}$ is the asymptotic limit of weights estimated via cross-validation. Theorem~\ref{thm:dist-shift-2} also holds when restricting $Q_1$ to a point mass e.g.~when performing pointwise inference. 

The consequences listed below follow from the above result, with comparisons between standard random forests, clustered random forests (CRF), REEM (likelihood based) and MERF (likelihood based) random forests given for the same parameter choices $(s_I,k_I)$ and (with the exception of standard random forests that don't have a weight class) 
fit with identical weight classes:
\begin{enumerate}[label=(\roman*)]
    \item $Q$-shift optimised clustered random forests are optimal in the sense of Theorem~\ref{thm:minimax} 
    in the sense that their weights directly target the minimal $\mathcal{V}_Q(\cdot)$ within the user-specified weight class.
    \item Likelihood based random forests (e.g.~REEM and MERF) are suboptimal in the sense that they can select weights that see an arbitrarily larger variance than the optimal weights within the user-specified weight class, i.e.~there exists distributions $P$ and covariate shifts $Q$ for which REEM and MERF random forests perform arbitrarily worse than \emph{both} clustered random forests and even standard random forests.
    \item An approach based on estimating the weights by minimising the MSPE under the training covariate distribution (as is targeted by e.g.~cross validation) can exhibit arbitrarily suboptimal performance on a covariate shifted test set compared to~\emph{both} clustered random forests and even standard random forests.
\end{enumerate}
All the above points are studied empirically in Section~\ref{sec:sim2} (see also Figure~\ref{fig:sim2}). 

In light of Theorem~\ref{thm:dist-shift-2}, our clustered random forest that adapts to $Q$-covariate shift minimises an empirical estimator of $\mathcal{V}_Q(\cdot)$ via~\eqref{eq:rho-testMSPE} as advocated in Algorithm~\ref{alg:cdt} to estimate the optimal parameter $\rho^*_{Q\text{-}\mathrm{opt}}:=\argmin_{\rho\in\Gamma}\mathcal{V}_{Q}(\rho)$. 
The following proposition justifies this empirical estimator~\eqref{eq:rho-testMSPE}.

\begin{proposition}\label{prop:weight-estimation}
    Consider the weight estimator $\hat\rho$ that minimises the empirical $Q$-adjusted loss function~\eqref{eq:rho-testMSPE}, as in Algorithm~\ref{alg:cdt} for the weight class satisfying Assumption~\ref{weights}, alongside the additional identifiability condition that 
    $$\mathcal{V}_Q(\rho)\geq \mathcal{V}_Q(\rho^*_{Q\text{-}\mathrm{opt}})+c\,\|\rho-\rho^*_{Q\text{-}\mathrm{opt}}\|_2^2,$$ for all $\rho\in\Gamma$ with $\|\rho-\rho^*_{Q\text{-}\mathrm{opt}}\|_2\leq\Delta$ and for some $c,\Delta>0$. 
    Also suppose the minimum node size is chosen to satisfy  $k_I^{-1}n_{\mathrm{c}}=o(1)$. Then
    $$\hat\rho=\rho_{Q\text{-}\mathrm{opt}}^* + O_P\Bigl(n_{\mathrm{c}}k_I^{-1/2}+(s_I^{\mathrm{corr}})^{-(\frac{\delta}{2+\delta}\wedge\frac{\phi}{d})}\Bigr), $$
    where $\phi:=(1-\epsilon)\frac{\log((1-\alpha)^{-1})}{\log(\alpha^{-1})}\pi$. 
\end{proposition}

The correlation weighting estimation scheme~\eqref{eq:rho-testMSPE} therefore satisfies the rate requirement~\ref{weights-rates} for Theorem~\ref{thm:dist-shift-2} for example if the subsampling rate $s_I$ is taken to be polynomial in $I$, cluster size $n_{\mathrm{c}}$ is finite and $k_I\sim(\log I)^{2d+1}$. Proposition~\ref{prop:weight-estimation} also accommodates diverging cluster sizes at sufficiently slow polynomial rates by taking the asymptotic regime of $k_I$ to diverge sufficiently fast. 
 
\subsubsection{Distinction between correlated and uncorrelated data when adjusting for covariate shift}

Our result of Theorem~\ref{thm:dist-shift-2} directly contrasts uncorrelated data settings e.g.~when the data are independent. In misspecified parametric mean models with independent observations, it is well known that improvements can be achieved (under covariate shift from $Q_1$ to $Q_2$) by the re-weighting scheme given by $dQ_2/dQ_1$~\citep{shimodaira, covshift-book}, with no improvements under model well specification~\citep{ge}. 
Note that for a random forest however provided the conditional mean is indeed Lipschitz smooth, the conditional mean is not model misspecified, and consequently the reweighting $dQ_2/dQ_1$ applied to a standard random forest effectively leaves the random forest predictor and hence covariate shifted MSPE unchanged. For a single decision tree this can be seen by noticing that the leaves are sufficiently localised, with diameters converging to zero, and so the weights $dQ_2/dQ_1$ are effectively constant on the local region on which we do estimation, giving asymptotically identical estimation to standard decision trees. 

\subsection{Optimality over subsampling rates}

We now turn our attention to the dependence of clustered random forests on the other hyperparameters, most notably those of $(s_I,k_I)$. The results of this section give rise to analogous immediate consequences for hyperparameter selection in standard random forests, which may be of independent interest.

Recall that by Theorem~\ref{thm:minimax} we obtain that for sufficiently large $I$,
\begin{equation*}
    \E_P\biggl[\int_\cX\bigl(\hat{\mu}_I^{\mathrm{MC}}(x)-\mu(x)\bigr)^2dQ(x)\biggr]
    \leq
    C_{\mathrm{bias}}^2\biggl(\frac{2k_I}{n_{\mathrm{c}}s_I}\biggr)^{\phi/d}
    +
    \frac{s_I}{k_I}\cdot\frac{\mathcal{V}_Q(\rho^*)}{I},
\end{equation*}
for $C_{\mathrm{bias}}:=\frac{6C_WL_\mu d^{1/2}\nu}{c_W}$, 
which is minimised when 
\begin{equation}\label{eq:subsampling-rate-opt}
    k_I^{-1}s_I=2\Biggl(\frac{C_{\mathrm{bias}}^2\phi I}{n_{\mathrm{c}}^{2\phi/d }\mathcal{V}_Q(\rho^*)}\Biggr)^{\frac{d}{2\phi +d}},
    \qquad
    \phi  := (1-\epsilon)\pi\,\frac{\log((1-\alpha)^{-1})}{\log(\alpha^{-1})}
\end{equation}
whence
\begin{equation}\label{eq:rate}
    \E_P\biggl[\int_\cX\bigl(\hat{\mu}_I^{\mathrm{MC}}(x)-\mu(x)\bigr)^2dQ(x)\biggr]
    \leq C_{\mathrm{bias}}^{\frac{2d}{2\phi +d}}\biggl(\frac{2\phi+d}{d}\biggr)\biggl(\frac{\mathcal{V}_Q(\rho^*)}{\phi N}\biggr)^{\frac{2\phi }{2\phi +d}},
\end{equation}
where $N=n_{\mathrm{c}}I$ constitutes the number of individual-level observations. 

Note therefore the system of hyperparameters is essentially overparametrised with regards to $Q$-shifted MSPE; we only require the ratio $k_I^{-1}s_I$ to diverge as in~\eqref{eq:subsampling-rate-opt} to obtain the above rate of convergence for $Q$-covariate shifted MSPE, thus giving a `valley' of optimal hyperparameters along $s_I=ck_I$ for an optimal constant $c>0$ (see Figure~\ref{fig:sk}). 
This allows us to accommodate forest splitting hyperparameter setups where~$k_I$ diverges sufficiently slowly without incurring an asymptotic cost to MSPE. 
If for example $s_I=3^{-1}I^\beta$ and $k_I=I^\kappa$ for constants $\beta\in(0,1)$, $\kappa>0$ then to achieve the rate of~\eqref{eq:rate} we require $\beta=\kappa+\phi$, implying the restriction on $k_I$ divergence via $\kappa<1-\phi$. 
In practice will grow particularly deep trees, such that we consider either $k_I$ as fixed and finite, or $k_I$ diverging at a poly-logarithmic rate, and in all our empirical experiments we take $k_I=10$. 

Note in particular~\eqref{eq:rate} allows for minimax rates under certain splitting rules.
\begin{remark}[Approaching minimax rates]\label{remark:minimax}
    Taking Theorem~\ref{thm:minimax} with splitting procedure satisfying~\ref{ass:tree} with $(\pi,\alpha)=(1,0.5)$, and $k_I^{-1}s_I$ as in~\eqref{eq:subsampling-rate-opt}, it follows that for sufficiently large $I$,
    \begin{equation*}
        \E_P\Big[\big(\hat{\mu}_I^{\mathrm{MC}}(x)-\mu(x)\big)^2\Big]
        \leq C_{\mathrm{bias}}^{\frac{2d}{2(1-\epsilon)+d}}\biggl(\frac{2(1-\epsilon)+d}{d}\biggr)\biggl(\frac{\mathcal{V}_Q(\rho^*)}{\phi N}\biggr)^{\frac{2}{2+(1-\epsilon)^{-1}d}}.
    \end{equation*}
    Taking $\epsilon\in(0,1)$ arbitrarily small, $\hat\mu_I^{\mathrm{MC}}$ is minimax optimal with respect to pointwise mean squared error. By analogous arguments this minimaxity also hold for the $Q$-integrated MSPE. Note the rate in the $(s_I,k_I)$ regime~\eqref{eq:subsampling-rate-opt} is achieved by taking $(s_I,k_I)\sim(I^\beta,1)$ where $\beta=\frac{d}{2(1-\epsilon)^{-1}+d}$.
\end{remark}
By taking Algorithm~\ref{alg:crf} with unclustered data and $\hat\rho=\rho=0$ the faster rates of Theorem~\ref{thm:minimax} and minimax rates of Remark~\ref{remark:minimax} apply also to standard random forests for unclustered data. 
We remark that, whilst it is informative and of theoretical interest that there exists a random forest mechanism in our setup for which we obtain minimax rates, one would likely expect in practice other splitting rules to outperform the splitting rules utilized in Remark~\ref{remark:minimax}. 
We also observe the curse of dimensionality in the dependence of $d$ in the above rate, as we make merely Lipschitz smoothness assumptions on the $d$-dimensional conditional mean. Throughout we have focused primarily on the case of low to moderate dimensional $d$. Studying the performance of random forests, and by virtue clustered random forests, in the high dimensional case with e.g.~further sparsity assumptions on $\mu$ would be an interesting direction of further work, and we provide some brief motivating empirical results in Section~\ref{sec:sim3}.

\begin{figure}
    \centering
    \includegraphics[width=0.55\linewidth]{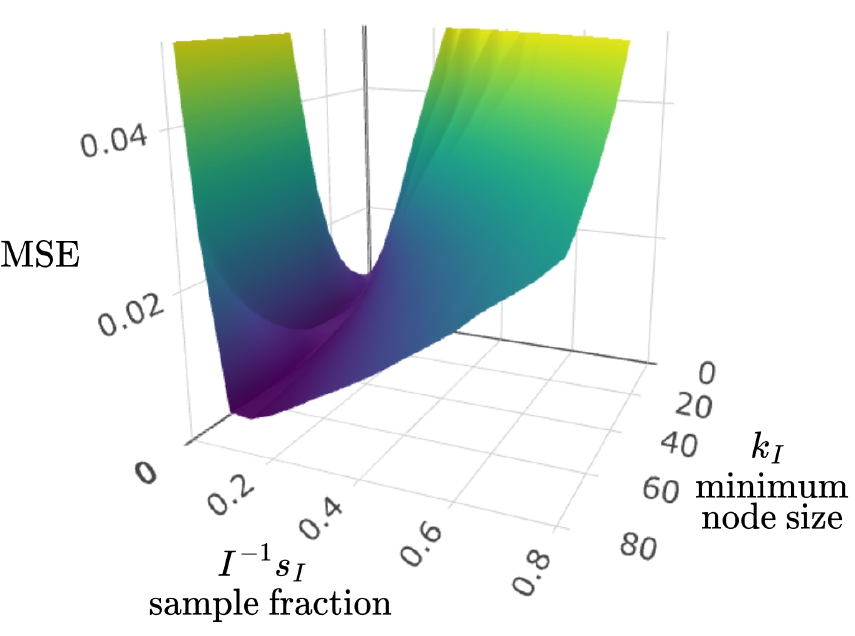}
    \caption{Pointwise MSE as a function of the minimum node size $k_I$ and subsampling fraction $I^{-1}s_I$ for the simulation of Figure~\ref{fig:sim1} with $I=10^4$. A valley runs along $k_I^{-1}s_I\approx30$. See Appendix~\ref{appsec:introsim} for further details.}
    \label{fig:sk}
\end{figure}

The above rate~\eqref{eq:rate} motivates a practical strategy to select a reasonable subsampling rate to minimum node size ratio in practice by using~\eqref{eq:subsampling-rate-opt} as a plug-in estimator given a `bootstrap of little bags' estimator for the quantity~$\mathcal{V}_Q(\rho^*)$. 
As with standard random forests, cross-validation could also be used to select $(s_I,k_I)$, with a reasonable parametrization for this being $s_I=3^{-1}I^\beta$. 

\subsection{Inference}\label{sec:inference}

The subsampling rate~\eqref{eq:subsampling-rate-opt} used to achieve minimax optimality in Theorem~\ref{thm:minimax} corresponds to a random forest that balances the bounds obtained for the squared bias and variance. For larger subsampling rates we obtain variance-dominating random forests, for which we can show uniform (over $\cP$) asymptotic Gaussianity results. For the following result we take an asymptotic regime with subsampling size $s_I=3^{-1}I^\beta$ diverging at a subsampling rate~$\beta$. 

\begin{theorem}[Asymptotic Normality of variance-dominating CRFs]\label{thm:normality}
    Suppose Assumption~\ref{ass:data} holds, and let $\hat{\mu}^{\mathrm{MC}}_I$ denote a clustered random forest estimator~\eqref{eq:MC} with tree splitting procedure satisfying~\ref{ass:tree} with $k = O(\log I)$, and $s_I=3^{-1}I^\beta$ with subsampling rate satisfying $\beta>\beta_{\min}:=1-\Big(1+\frac{d}{2\pi}\frac{\log(\alpha^{-1})}{\log((1-\alpha)^{-1})}\Big)^{-1}$. Suppose also the weighting scheme $(W_i(\rho))_{\rho\in\Gamma}$ satisfies~\ref{ass:weights}.
    Now fix an $x\in\cX$. Then there exists a sequence of positive deterministic functions $(\sigma_{P,I}(x;\rho^*))_{I\in\mathbb{N}}$ with 
    $\sigma_{P,I}^2(x;\rho^*)=O_\cP\big(I^{-(1-\beta)}\big)$ and
    \begin{equation}\label{eq:pointwise-normality}
        \lim_{I\to\infty}\sup_{P\in\mathcal{P}}\sup_{t\in\R}\big|\PP_P\big(\sigma_{P,I}(x;\rho^*)^{-1}\big(\hat{\mu}_I^{\mathrm{MC}}(x)-\mu(x)\big)\leq t\big)-\Phi(t)\big|
        =
        0.
    \end{equation}
\end{theorem}

Theorem~\ref{thm:normality} shows that the variance--dominated clustered random forests are asymptotically unbiased and Gaussian pointwise over $\cX$ and uniformly in probability for the class of distributions $\cP$ as defined by Assumption~\ref{ass:data}. 
Moreover, a similar result holds for estimating the integrated mean squared error over a fixed covariate distribution $Q$ with $\supp Q\subseteq\cX$, in that there exists a sequence of deterministic positive functions $(\sigma^{(Q)}_I(\rho^*))_{I\in\mathbb{N}}$ with $\sigma_I^{(Q)\,2}(\rho^*)=O_\cP\big(I^{-(1-\beta)}\big)$ and 
    \begin{equation*}
        \lim_{I\to\infty}\sup_{P\in\mathcal{P}}\sup_{t\in\R}\bigg|\PP_P\bigg(\sigma^{(Q)}_I(\rho^*)^{-1}\Big(\int_{\cX}\hat{\mu}_I^{\mathrm{MC}}(x)d{Q}(x)-\int_{\cX}\mu(x)d{Q}(x)\Big)\leq t\bigg)-\Phi(t)\bigg|
        =
        0.
    \end{equation*}
An empirical realisation of this result is seen in Appendix~\ref{appsec:sim2}. 
Also note that the subsampling rates we accommodate of $\beta>1-\bigl(1+\frac{d}{2\pi}\frac{\log(\alpha^{-1})}{\log((1-\alpha)^{-1})}\bigr)^{-1}$ allow for smaller subsampling rates than e.g.~\citet{wager,GRF,naf} who require~$\beta>1-\bigl(1+\frac{d}{\pi}\bigl\{\frac{\log(\alpha^{-1})}{\log((1-\alpha)^{-1})}\vee\frac{\log5}{\log(5/4)}\bigr\}\bigr)^{-1}$. As an example, in the case $(d,\pi,\alpha)=(2,1,0.5)$ we may take $\beta>0.5$, whereas the results of the aforementioned references require~$\beta>0.93$. This wider range of permissible subsampling rates allows for us to obtain confidence intervals whose widths converge at faster rates too. As our non-i.i.d.~data setting collapses to i.i.d.~data in special cases, these strengthened subsampling rates can also be applied in this simpler setting.

Regarding testing based on the pointwise mean function, note in particular that, with a sequence of estimators $(\hat{V}_I(x))_{I\in\mathbb{N}}$ of the variances of the sequence of clustered random forest estimators $(\hat{\mu}_I(x))_{I\in\mathbb{N}}$ we can construct asymptotically valid confidence intervals for $\mu(x)$ using~\eqref{eq:pointwise-normality}. To construct such an estimator a number of bootstrap aggregation procedures~\citep{bootstrap} have been proposed for standard honest random forests, and that naturally extend to clustered settings. These include the infinitesimal jackknife~\citep{efron, wager} the bias corrected infinitesimal jackknife~\citep{inf-jack}, and `bootstrap of little bags' estimators~\citep{sexton, GRF, naf}. As a benchmark, in the construction of Algorithm~\ref{alg:crf}, in addition to the numerical results of Section~\ref{sec:numericals}, we consider the latter `bootstrap of little bags' estimator $\hat{V}_I(x)$ of the variance, and obtain nominal coverage of the $\tilde{\alpha}$-level confidence interval $\hat{C}_{\tilde{\alpha}}(x)$ for $\mu(x)$ for a fixed $x\in\cX$ and $P\in\cP$ (as in Algorithm~\ref{alg:crf}) i.e.
\begin{equation*}
    \lim_{I\to\infty}\big|\PP_P\big(\mu(x)\in\hat{C}_{\tilde{\alpha}}(x)\big)-(1-\tilde{\alpha})\big|=0,
\end{equation*}
following from Theorem~\ref{thm:normality} and the consistency of the `bootstrap of little bags' estimator~\citep[Section 4.1]{GRF}. Note that as one can cast pointwise inference as a covariate distribution shift problem, from the covariate distribution of the training data to a point mass at the point of interest, the optimality results of clustered random forests, and the corresponding suboptimalities of pre-existing alternative methods, under covariate shift (Section~\ref{sec:covariate-shift}) directly apply in this context with regards to obtaining tight confidence intervals for functionals of the conditional mean.

\subsection{Approximately linear time computation of clustered random forests}\label{sec:linear-time}

A key practical advantage of standard random forests is their decision trees can be trained in linear time $O(n_{\mathrm{c}}s_I)$ with the number of observations used for splitting and evaluation in the tree. Using the CART criterion for splitting requires $O(n_{\mathrm{c}}s_I)$ operations. For the evaluation step, using the notation of Algorithm~\ref{alg:cdt}, we recover ordinary decision trees by setting $\hat\rho=\rho=0$. In this case, the matrix in the evaluation step simplifies to
\(
\sum_i \chi_i^\top W_i(0)\chi_i \;=\; \sum_i \chi_i^\top \chi_i,
\)
which is diagonal and therefore can be inverted in $O(M)$ operations. Since $M$ scales like $n_{\mathrm{c}}s_I/k_I$, a standard decision tree is fit in linear time $O(n_{\mathrm{c}}s_I)$ with respect to the number of data points. 

For clustered decision trees (Algorithm~\ref{alg:cdt}) the evaluation step is more involved. 
For general $\rho \neq 0$ the matrix
\(
\sum_i \chi_i^\top W_i(\rho)\chi_i \in \mathbb{R}^{M \times M}
\)
is typically dense, and a naive inversion would require $O(M^3)$ operations—prohibitively expensive as the number of leaves grows with sample size. Nonetheless, although this matrix is not sparse, certain choices of working-correlation weights yield structure that enables efficient computation via conjugate gradient descent. In particular, for the equicorrelated and AR$(1)$ weight classes used in our running examples, we show that the fitted values
\(
\bigl(\sum_i \chi_i^\top W_i(\rho)\chi_i\bigr)^{-1}
\bigl(\sum_i \chi_i^\top W_i(\rho)Y_i\bigr)
\)
can still be computed in linear time through a lightweight application of conjugate gradient descent.

\begin{proposition}[Approximate linear time clustered random forest fitting]\label{prop:linear-time}
    Suppose $W_i(\rho)$ takes the equicorrelated or $\mathrm{AR}(1)$ working correlation weight structure. Then:
\begin{enumerate}[label=(\roman*)]
    \item {\bf Evaluation complexity:} There exists a computational strategy to fit a clustered decision tree with fixed $\rho$ that computes fitted values to an $\epsilon$-accuracy in Euclidean norm in $O(n_{\mathrm{c}}s_I\log(\epsilon^{-1}))$ operations.\label{item:comp-1}
    \item 
    {\bf Correlation weight estimation complexity:} For a given covariate distribution $Q$ and tree splits giving rise to leaves $(L_m:m\in[M])$ define $c_Q=|\{m\in[M]:L_m\cap\,\mathrm{supp}\,Q\neq\emptyset\}|$. Then there exists a computational strategy that estimates the optimal correlation parameter via~\eqref{eq:rho-testMSPE} such that Algorithm~\ref{alg:cdt} requires $O(c_Qn_{\mathrm{c}}s_I^{\mathrm{corr}})$ operations.
    \label{item:comp-2}
\end{enumerate}
\end{proposition}

Case~\ref{item:comp-1} motivates a linear time fitting procedure for a fixed correlation parameter $\rho$, e.g.~if $\rho$ were to be treated as a hyperparameter for tuning alongside the standard random forest hyperparameters, which could be estimated in non-covariate shift settings e.g.~by cross-validation. 
Note in the case of~\ref{item:comp-2} when considering a covariate shift to a point mass - e.g.~targeting optimal pointwise MSE or pointwise inference (see also Section~\ref{sec:inference}) - the quantity $c_Q=1$. At the other extreme when $\mathrm{supp}\,Q=\cX$ then $c_Q=M\sim k_I^{-1}n_{\mathrm{c}}s_I^{\mathrm{corr}}$ and thus the computational order of our proposed strategy is of order~$O(k_I^{-1}n_{\mathrm{c}}^2(s_I^{\mathrm{corr}})^2)$. Because we only require slow rates for $\rho$-estimation (of which arbitrarily slow polynomial rates suffices; see Assumption~\ref{weights-rates}), taking an asymptotic regime for $s_I^{\mathrm{corr}}=o\bigl((n_{\mathrm{c}}^{-1}k_Is_I)^{1/2}\bigr)$ allows the correlation weight estimation to add only an asymptotically negligible excess computational cost over standard random forests. 

Together~\ref{item:comp-1} and~\ref{item:comp-2} allow for clustered random forests to be fit in approximately linear time, effectively matching the favourable scalability properties of standard random forests.  
We remark that the pre-existing random forests algorithms for clustered data via REEM and MERF random forests are cubic-time procedures of order $O(M^3)$ (even with oracle knowledge of the correlation) and subsequently accommodating correlations with these procedures come at a high computational cost. 

\section{Numerical experiments}\label{sec:numericals}

We study the performance of our clustered random forests on a number of simulated and real-world data examples. 
We will compare the following random forests: 
\begin{itemize}[itemsep=0.0ex]
    \item[RF] The standard honest random forest as in the package~\texttt{grf}~\citep{wager, GRF}. Sample splitting for each decision tree is performed cluster-wise, but each decision tree is fit as though the data were independent. This approach is therefore essentially that of Algorithm~\ref{alg:crf} with $\rho=0$. This therefore acts as a good `unweighted' benchmark.
    \item[CRF] Our clustered random forest estimator of Algorithm~\ref{alg:crf}, with weight parameters estimated via~\eqref{eq:rho-testMSPE} for covariate distribution $Q$ as specified for each simulation.
    \item[TRAIN] The clustered random forest estimator of Algorithm~\ref{alg:crf}, with weight parameters estimated by minimising the training error loss~\eqref{eq:TRAIN}.
    \item[REEM] The clustered random forest subsampling of Algorithm~\ref{alg:crf} with predictions estimated by fitting a mixed effects model with intercept only random effects  as in the package~\texttt{REEMTrees}~\citep{REEMTrees}. The approach is therefore asymptotically equivalent to the procedure of Algorithm~\ref{alg:crf} with $\rho$ estimated by fitting a random effects model on the residual terms.
\end{itemize}
All forests satisfy honesty~\ref{honesty}, with each tree trained using $s_I=s_I^{\mathrm{corr}}=3^{-1}I^\beta$ observations for the evaluation step, with subsampling rate $\beta=0.9$. Trees are grown in all cases to minimum node size of $k=10$, with tree splits obtained by standard CART. 
Thus the distinction between all four forests are via their weight estimation strategies only, allowing for direct comparison. 
Section~\ref{sec:sim2} studies the setting of optimal prediction under covariate distribution shift. Section~\ref{sec:sim3} studies inference for variance-dominating random forests, and Section~\ref{sec:realdata} studies a real data analysis of personalised predictions of CD4 cell counts for HIV seroconverters, highlighting the variance reduction gains clustered random forests exhibit. 
When the target is predictive performance (Section~\ref{sec:sim2}) we compare all the aforementioned estimators. As REEM trees are not readily amenable for uncertainty quantification, they are not included when the target is inference (Sections~\ref{sec:sim3} and~\ref{sec:realdata}).

\subsection{Pediction under covariate distribution shift}\label{sec:sim2}

Consider training data consisting of $I=10\,000$ i.i.d.~realisations of groups each of size $n_i\equiv4$ and following the distribution
\begin{gather*}
    X \sim N_4\big({\bf 0}, I_4\big),
    \qquad
    Y\given X \sim N_4\big(\mu(X), \Sigma(X)\big),
    \\
    (\mu(x))_j = \mu(x_j) = \tanh(x_j)
    \qquad
    \Sigma_{jk}(x) := \rho_{jk}\sigma(x_j)\sigma(x_k)
    ,
    \\
    \sigma(x_j) := \frac{1}{4}+\frac{1}{1+\text{exp}(4x_j)}, 
    \qquad
    \rho_{jk} = \ind_{(j=k)} + 0.8\cdot\ind_{(j\neq k)}.
\end{gather*}
Weights for the CRF and TRAIN random forests are fit using the equicorrelated weight structure; note as REEM random forests are fit using an intercept only random effects model the weight classes for all estimators coincide. We study the MSPE under two different covariate distributions; the first is precisely that of the training data i.e.~$N(0,1)$, and the second is a covariate distribution shifted test dataset~$\text{Unif}\hspace{0.08em}[1,2]$. The CRF random forest minimises~\eqref{eq:rho-testMSPE} with $Q=N(0,1)$ in the former and $Q=\text{Unif}\hspace{0.08em}[1,2]$ in the latter. Table~\ref{tab:sim2} presents the MSPE the two covariate distributions. In particular, training error based weighted TRAIN and model based weighted REEM forests both are seen to perform rather poorly under the covariate shifted distribution, with accuracy even worse than the benchmark RFs; see Figure~\ref{fig:sim2}. In contrast, our covariate distributed adjusted random forests of CRF outperform all the other methods regardless of the covariate shifting. 

\begin{figure}[h]
    \centering
    \includegraphics[width=0.98\linewidth]{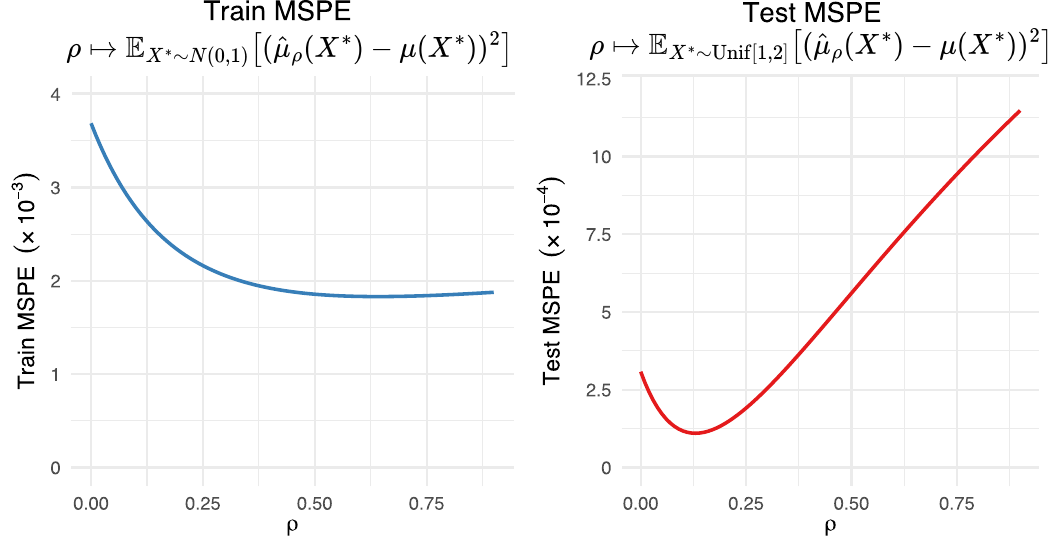}
    \caption{The training (blue) and testing (red) mean squared prediction errors (MSPE) for clustered random forests, weighted by an equicorrelated structure with fixed equicorrelation parameter $\rho\in[0,0.9]$ for the setting of Simulation~\ref{sec:sim2}. 
    The optimal clustered random forest with respect to training MSPE (corresponding to $\rho\approx 0.55$) can be seen to provide worse testing MSPE than any clustered random forests with any $\rho\in[0,0.55)$, including worse testing MSPE over even standard honest random forests (obtained by taking $\rho=0$).
    }
    \label{fig:sim2}
\end{figure}

\begin{table}[ht]
	\begin{center}
		\begin{tabular}{c|cc}
		\toprule
            \multirow{3}{*}{Method}& \text{Mean Squared Train} & \text{Mean Squared Test}
            \\
            &\text{Prediction Error ($\times10^{-3}$)}&\text{Prediction Error ($\times10^{-4}$)}
            \\
            & $\E_{X^*\sim N(0,1)}\big[(\hat{\mu}(X^*)-\mu(X^*))^2\big]$ & $\E_{X^*\sim\text{Unif}\hspace{0.08em}[1,2]}\big[(\hat{\mu}(X^*)-\mu(X^*))^2\big]$
            \\
		\midrule
            RF & 3.47 & 3.16
            \\
            CRF & {\bf 1.71} & {\bf 1.10}
            \\
            TRAIN & {\bf 1.71} & 7.07
            \\
            REEM & 1.72 & 7.06
            \\
		\bottomrule
		\end{tabular}
		\caption{Results of Simulation~\ref{sec:sim2} (1000 simulations).}\label{tab:sim2}
	\end{center}
\end{table}

\subsection{Inference}\label{sec:sim3}

Consider $I=1000$ i.i.d.~clusters each of size $n_i\equiv5$, with each cluster $(Y_i,X_i)\in\R^{n_i}\times\R^{n_i\times d}$ generated via 
\begin{gather*}
    (X_i)_{jk} \iid N(0,1),
    \qquad
    Y_{ij} = \mu(X_{ij}) + \varepsilon_{ij},
    \qquad
    (j\in[n_i], k\in[d])
    \\
    \mu(x) = 4\sin x_1,
    \qquad
    \varepsilon_i \sim \text{AR}(2), 
    \quad
    \phi_{\text{AR}} = (0.6,0.3),
\end{gather*} 
where we take i.i.d.~Gaussian innovations for the autoregressive $\text{AR}(2)$ process with autoregressive parameter $\phi_{\text{AR}}$. 
We consider a range of covariate dimensions $d\in\{1,3,5,10,50\}$. 
Our primary task of interest here is inference on the quantity $\mu({\bf1})$, where ${\bf1}\in\R^d$ is the vector of all ones. 
Confidence intervals are constructed for honest clustered random forests (Algorithm~\ref{alg:crf} with weights estimated via~\eqref{eq:rho-testMSPE} with $Q=\delta_{\bf1}$) with a (misspecified) inverse $\text{AR}(1)$ working weight structure. 
We compare CRFs with RFs, both of which use $B=500$ and $R=100$ in Algorithm~\ref{alg:crf}. 
REEM forests are omitted here as they do not readily allow for uncertainty quantification.

\begin{table}[h]
\centering
    \begin{minipage}{0.49\textwidth}
    \begin{center}
        Length of $95\%$ confidence interval $(\times 10^{-1})$
        \\
        \vspace{0.04cm}
		\begin{tabular}{c|ccccc}
		\toprule
            \multirow{2}{*}{Method} & \multicolumn{5}{c}{$d$}
            \\
            \cline{2-6} & 1 & 3 & 5 & 10 & 50 
            \\
		\midrule
            RF & 2.60 & 3.08 & 3.10 & 3.10 & 2.87
            \\
            CRF & 1.77 & 2.46 & 2.61 & 2.71 & 2.60
            \\
		\bottomrule
		\end{tabular}
    \end{center}
    \end{minipage}
  \begin{minipage}{0.49\textwidth}
    \begin{center}
        Coverage of 95\% confidence intervals
        \\
        \vspace{0.04cm}
		\begin{tabular}{c|ccccc}
            \toprule
            \multirow{2}{*}{Method} & \multicolumn{5}{c}{$d$}
            \\
            \cline{2-6} & 1 & 3 & 5 & 10 & 50
            \\
		\midrule
            RF & 0.948 & 0.945 & 0.950 & 0.935 & 0.955
            \\
            CRF & 0.956 & 0.940 & 0.948 & 0.931 & 0.936
            \\
		\bottomrule
		\end{tabular}
    \end{center}
    \end{minipage}
  \begin{minipage}{0.8\textwidth}
  \vspace{0.2cm}
	\begin{center}
        Mean squared error $\E\big[(\hat{\mu}({\bf1})-\mu({\bf1}))^2\big]$ ($\times 10^{-3}$)
        \\
        \vspace{0.04cm}
		\begin{tabular}{c|ccccc}
		\toprule
            \multirow{2}{*}{Method} & \multicolumn{5}{c}{$d$}
            \\
            \cline{2-6} & 1 & 3 & 5 & 10 & 50 
            \\
		\midrule
            RF & 4.23 & 6.16 & 5.83 & 6.27 & 5.38
            \\
            CRF & 1.93 & 3.99 & 4.18 & 4.86 & 4.32
            \\
		\bottomrule
		\end{tabular}
        \end{center}
    \end{minipage}
    \caption{Results of Simulation~\ref{sec:sim3} (1000 simulations).}\label{tab:sim3}
\end{table}

Table~\ref{tab:sim3} presents the empirical mean squared error of $\hat{\mu}({\bf1})$ and the interval width and coverage of nominal $95\%$ confidence intervals for $\mu({\bf1})$. 
In all cases the confidence intervals for CRFs are seen to maintain equivalent validity in terms of coverage as RFs do, but with smaller confidence interval width. Further, as we would expect, CRFs see a reduced mean squared error over RFs. Whilst these gains are most prevalent in lower dimensions, they are seen to persist even in moderate dimensional settings such as $d=50$. For further details of the results, including QQ plots and a more detailed analysis of coverage properties, see Appendix~\ref{appsec:sim3}. 

\subsection{Real data analysis: CD4 cell count prediction for HIV patients}\label{sec:realdata}

We study the CD4 cell count dataset from a repeated measures study from human immunodeficiency virus (HIV) seroconverters, which can be accessed as the \texttt{aids} dataset in the \texttt{R} package \texttt{jmcm}~\citep{pan}. The dataset contains the CD4 cell count of $I=369$ HIV participants during seroconversion, each of which we treat as an individual cluster. The median number of repeated measures (cluster size) is 6, with minimum and maximum cluster sizes of 1 and 12 respectively. For each patient we have access to the following covariates: time before/after seroconversion, age, smoking status, whether the participant takes drugs, the number of sex partners and depression status. See~\citet{zeger} for further details on this dataset, which has been subsequently studied by~\citet{taavoni, wang, emmenegger}. The aformentioned references all study variable importance of each covariate by imposing (semi)parameteric mixed models on the data. In contrast we study the problem of nonparametric personalised predictions, specifically we consider predicting the CD4 cell count for three `example patients' as per the following demographics:
\begin{table}[ht]
    \centering
    \begin{tabular}{c|cccccc}
        & \begin{tabular}{@{}c@{}c@{}}Years after\\seroconversion\end{tabular} & \begin{tabular}{@{}c@{}c@{}}Smoking\\status\end{tabular} & \begin{tabular}{@{}c@{}c@{}}Drugs use\\status\end{tabular} & Age & \begin{tabular}{@{}c@{}c@{}}Number of\\sexual partners\end{tabular} & \begin{tabular}{@{}c@{}c@{}}Depression\\status\end{tabular}\\
        \hline
        $x_{(1)}^*$ &  1 year & 2 packs & Drug user & 1st quartile & median & 10th centile
        \\
        $x_{(2)}^*$ & 2 years &  4 packs & Drug user & median & 90th centile & 90th centile
        \\
        $x_{(3)}^*$ & 5 years & none 
        & No drugs & 90th centile & 90th centile & 90th centile
    \end{tabular}
    \label{tab:my_label}
\end{table}

\noindent
We compare our CRFs with RFs, constructing confidence intervals for the conditional mean function evaluated at the three patient characteristic covariates above; see Table~\ref{tab:real-data}. As in Section~\ref{sec:sim3} we again omit REEM forests. We see improvements in our CRFs in the variance estimates for each patient characteristic.

\begin{table}[ht]
    \begin{center}
        $\text{CD4 cell count dataset (Section~5.3.2)}$
        \\
		\begin{tabular}{c|cc|cc|cc}
		\toprule
            Method & $\hat{\mu}(x^*_{(1)})$ & $\hat{V}^{1/2}\big(x^*_{(1)}\big)$ & $\hat{\mu}\big(x^*_{(2)}\big)$ & $\hat{V}^{1/2}\big(x^*_{(2)}\big)$ & $\hat{\mu}\big(x^*_{(3)}\big)$ & $\hat{V}^{1/2}\big(x^*_{(3)}\big)$ 
            \\
		\midrule
            RF & 666.0 &
25.3 \;\,(0\%) & 553.9 & 27.1 \;\,(0\%) & 645.2 & 33.5 \;\,(0\%)           \\
            CRF & 655.8 & 20.8 (32\%) & 539.8 & 22.2 (33\%) & 625.1 & 25.9 (40\%)
            \\
		\bottomrule
		\end{tabular}
    \end{center}
    \caption{Results of real data analysis of Section~\ref{sec:realdata}. 
    We present (personalised) pointwise estimates of the expected CD4 cell count for three patients, and their standard error estimates (alongside in brackets the percentage reduction in the estimator of the variance compared to the RF benchmark).}\label{tab:real-data}
\end{table}

\section{Discussion}\label{sec:discussion}

In this work we have introduced a new random forest algorithm for clustered data that utilises intra-cluster correlations for lower variance prediction error. We extend the theory of random forests to obtain asymptotic Gaussianity results for `Monte-Carlo limit' variance-dominating clustered random forests as well as the existence of splitting rules and subsampling mechanisms for which minimax rates are obtained. 

A main contribution of this work is the recognition that when data are correlated the optimal weighting parameters in clustered random forests depend intricately on the covariate distribution with which one wishes to obtain optimal MSPE, and consequently the lack of optimality of standard random forests in these settings (Theorem~\ref{thm:dist-shift-2}). We implement a computationally competitive weight estimation strategy that aims to target the optimal MSPE with respect to a user-specified covariate shift, allowing for optimal covariate shift adaptation in these settings. Given the growing recognition of covariate shift being a potential factor for poor practical performance in a range of machine learning tasks, these results offer a fresh perspective on aspects of the covariate shift problem, specifically through the lens of correlated data. 

This work naturally offers a number of interesting directions for further study. One is to extend the class of `working correlations' that clustered random forests may accommodate, both theoretically with regards to valid inference (i.e.~relaxing the property~\ref{ass:weights}) and practically in terms of computational considerations. 
Another helpful direction would be to extend estimation beyond functionals of the conditional mean to other functionals of the conditional distribution, such as conditional quantiles. 

It would also be of interest to further study the higher dimensional properties of clustered random forests; in particular we see empirical evidence in Section~\ref{sec:sim3} that clustered random forests still maintain good performance and confidence intervals remain valid even in moderate dimensions with a sparse mean function, albeit with a cost to variance. A further theoretical treatment of Theorem~\ref{thm:normality} to high dimensional settings, for example with sparsity, would also be a natural extension. 

\paragraph*{Funding.}~EHY was supported by the EPSRC Doctoral Training Partnership and European Research Council Advanced Grant
101019498.

\paragraph*{Data availability.} The data supporting the findings of this study are available at \url{https://cran.r-project.org/src/contrib/jmcm_0.2.4.tar.gz}, and all code used to produce the numerical results in this article is available at \url{https://github.com/elliot-young/clustered_randomforests_simulations}.

\newpage

\appendix

\begin{center}
    \Large \bf Supplementary material relating to `Clustered random forests with correlated data for optimal estimation and inference under potential covariate shift', by Elliot H.\ Young \\and Peter B\"uhlmann
\end{center}

\section{Additional details on numerical experiment of Section~\ref{sec:intro}}\label{appsec:introsim}

We give further details of the numerical experiment presented in Figure~\ref{fig:sim1} of Section~\ref{sec:intro}. Specifically, we consider $I=10\,000$ i.i.d.~realisations of clusters of size $n_i\equiv2$ and covariate dimension $d=2$ each following
\begin{gather*}
    X_{jk}\iid N(0,1)
    \qquad(j\in[n_i],k\in[d]),
    \\
    Y\given X\iid N_2\bigg(
    \begin{pmatrix}
        \tanh(X_{11})+\tanh(X_{12})
        \\
        \tanh(X_{21})+\tanh(X_{22})
    \end{pmatrix}
    ,\,
    \begin{pmatrix}
        1&0.8\\0.8&1
    \end{pmatrix}
    \bigg),
\end{gather*}
Both random forests are grown as honest forests (see Algorithm~\ref{alg:crf}) with each tree grown to minimum node size 10 and subsampling rate $\beta\in(0.7,0.95)$. 
Performance is given for the pointwise estimator $\hat{\mu}({\bf1})$ where ${\bf1}=(1,1)$. In the case of the clustered random forest, the working correlation parameter dictating the additional weights are estimated via the loss function~\eqref{eq:rho-testMSPE} with $Q=\delta_{\bf1}(x)$ (i.e.~a point mass at the covariate of interest).

\section{Adapted (dishonest) clustered random forests}\label{appsec:dishonest-rf}

In Section~\ref{sec:methodology} we introduced honest clustered random forests, extending honest random forests as in~\citet{wager} to model correlations in clustered data, and allowing for theoretically valid inference as seen in Section~\ref{sec:theory}. Note however we could similarly extend `dishonest' clustered random forests as in~\citet{randomforest} to clustered data. Algorithm~\ref{alg:dcrf} presents this extension, with forests defined in terms of a `working correlation' parameter that acts as an additional hyperparameter and can be tuned via for example cross validation.

\begin{algorithm}[ht]
\KwIn{Grouped dataset, with groups indexed by a set $\cI$ of size $I:=|\cI|$; subsampling fraction $f$; weighting parameter $\rho$.}

\For{$b\in[B]$} {
    Select a random subset $\cI_b\subset\cI$ of size $s_n\sim f\,I$.

    Build a decision tree using  using the data $\Isplit$ (with splitting criterion satisfying Assumption~\ref{ass:tree}). Suppose this outputs $M_b$ nodes, with leaves denotes by $L_1,\ldots,L_{M_b}$. Further, let $L(x)$ be the leaf that contains $x\in\R^d$ and let $J(x)$ be the node index $m$ such that $L(x)=L_m$.

    For each $i\in\cI_{b}$ calculate $(\chi_i)_{j,m}=\ind_{(X_{ij}\in L_m)}$.

    $\hat{\mu}_b(x) := \mathrm{e}_{J(x)}^\top \big(\sum_{i\in\cI_b}\chi_i^\top W_i(\rho)\chi_i\big)^{-1}\big(\sum_{i\in\cI_b}\chi_i^\top W_i(\rho)Y_i\big)$.
}

    $\hat{\mu}(x) := B^{-1}\sum_{b=1}^B \hat{\mu}_b(x)$.

\KwOut{Clustered random forest estimate $\hat{\mu}(x)$.}
\caption{(Dishonest) Clustered Random Forests (dCRF)}
\label{alg:dcrf}
\end{algorithm}

\section{Proofs of results in Section~\ref{sec:theory}}\label{appsec:theory}

In this section we prove the theoretical results of Section~\ref{sec:theory}. First we introduce some relevant notation that will be used throughout this section.

\subsection{Notational Setup}
We notate $Z_i:=(Y_i,X_i)$ as the data making up the $i$th cluster. Our random forest predictor is constructed by aggregating a number of clustered decision trees (as in Algorithm~\ref{alg:crf}), that takes the form $T(x,\xi;Z_{i_1},\ldots,Z_{i_{\bar{s}}})$ where $(Z_{i_1},\ldots,Z_{i_{\bar{s}}})$ is a (cluster-wise) size ${\bar{s}}$ subset of the data set $(Z_i)_{i\in[I]}$ and $\xi\sim\Xi$ is a source of auxiliary randomness in 
the sample splitting for the tree. Note where clear we will omit the explicit notational dependence on $I$ in quantities such as $\bar{s}_I$ and $k_I$ for brevity. In practice an average over $B$ trees (or $B\times R$ trees as in Algorithm~\ref{alg:crf}) is taken. 
The `Monte Carlo limit' of this random forest given by~\citep{wager} 
\begin{align*}
    \hat{\mu}(x) &:= 
    \binom{I}{\bar{s}}^{-1}\sum_{i_1<\cdots<i_{\bar{s}}}\E_{\xi\sim\Xi}[T(x,\xi;Z_{i_1},\ldots,Z_{i_{\bar{s}}})\given Z_{i_1},\ldots,Z_{i_{\bar{s}}}]
    \\
    &=
    \binom{I}{\bar{s}}^{-1}\sum_{i_1<\cdots<i_{\bar{s}}}T(x;Z_{i_1},\ldots,Z_{i_{\bar{s}}}),
\end{align*}
i.e.~where the expectation is over the auxiliary random variable $\xi$. 
Again for brevity we write $\hat\mu$ in place of $\hat\mu_I^{\mathrm{MC}}$ in Section~\ref{sec:theory}. 
For notational convenience we will define $T(x;Z_{i_1},\ldots,Z_{i_{\bar{s}}}) := \E_P[T(x,\xi;Z_{i_1},\ldots,Z_{i_{\bar{s}}})\given Z_{i_1},\ldots,Z_{i_{\bar{s}}}]$, and where clear we omit the dependence on $(Z_{i_1},\ldots,Z_{i_{\bar{s}}})$ i.e.~notate such a tree by $T(x)$. For clarity, $T(x)$ is a clustered decision tree for a weight function $W(\hat{\rho})$ in terms of an estimated $\hat{\rho}$; whereas we also introduce the notation $T_{\rho}(x)$ for an analogous tree with weights taken as $W(\rho)$ in terms of a fixed, deterministic parameter $\rho\in\Gamma$, and $\hat{\mu}_{\rho}(x)$ the analogous `Monte Carlo limit' random forest.

The number of clusters in the dataset is denoted by $I$, and an individual tree's evaluation data index set is given by $\cI_{\text{eval}}$. Also define $n_i$ to be the number of observations in the $i$th cluster, $N:=\sum_{i=1}^In_i$ the total number of observations across all clusters, $s_I:=|\{i:i\in\cI_{\text{eval}}\}|$ to be the number of subsampled clusters in the evaluation dataset, and $s_N:=|\{(i,j):i\in\cI_{\text{eval}}\}|=\sum_{i\in\cI_{\text{eval}}}n_i=n_{\mathrm{c}}s_I$ the total number of observations in the subsampled evaluation dataset. 
We also write $Z_{\mathrm{split}}:=(Z_i:i\in\cI_{\mathrm{split}})$, $Z_{\mathrm{eval}}:=(Z_i:i\in\cI_{\mathrm{eval}})$ and $Z_{\mathrm{corr}}:=(Z_i:i\in\cI_{\mathrm{corr}})$, as well as the analogous $X_{\mathrm{split}}, Y_{\mathrm{split}}, X_{\mathrm{eval}}$, $Y_{\mathrm{eval}}$. 

For sequences $(a_n)_{n\in\mathbb{N}}$ and $(b_n)_{n\in\mathbb{N}}$ we write $a_n\lesssim b_n$ when there exists a constant $c>0$ that does not depend on $n\in\mathbb{N}$, $x\in\cX$ or $P\in\cP$ and $n^*\in\mathbb{N}$, that does not depend on $x\in\cX$ nor $P\in\mathcal{P}$, such that for all $n\geq n^*$ we have $a_n\leq cb_n$, and analogously for $a_n\gtrsim b_n$. Note in particular the lack of dependence on $x\in\cX$ and $P\in\mathcal{P}$ in $c$ here allows us to obtain the uniform convergence results where relevant. Also write $a_n\sim b_n$ when $a_n\lesssim b_n$ and $a_n\gtrsim b_n$.

\subsection{Asymptotic unbiasedness of clustered random forests}\label{appsec:proof-unbiasedness}

In this section we prove the asymptotic unbiasedness of the estimator in Algorithm~\ref{alg:crf}.

\begin{lemma}[Asymptotic unbiasedness of clustered random forests]\label{lem:unbiasedness}
    Suppose Assumption~\ref{ass:data} holds, and let $\hat{\mu}_I$ denote a clustered random forest estimator (Algorithm~\ref{alg:crf}) with properties~\ref{ass:tree} and~\ref{ass:weights}. Fix an arbitrarily small constant $\epsilon\in(0,1)$. Then for sufficiently large $I$,
    \begin{equation*}
        \sup_{P\in\mathcal{P}}
        \sup_{x\in\cX}\big|\E_P[\hat{\mu}_I(x)]-\mu(x)\big| \leq C_{\mathrm{bias}}\,\Bigl(\frac{s_N}{2k_I-1}\Bigr)^{-(1-\epsilon)\frac{\log((1-\alpha)^{-1})}{\log(\alpha^{-1})}\frac{\pi}{d}\beta},
    \end{equation*}
    for the finite positive constant $C_{\mathrm{bias}}:=\frac{6C_WL_\mu d^{1/2}\nu}{c_W}$.
\end{lemma}
The result of Lemma~\ref{lem:unbiasedness} holds for clustered random forests consisting of an arbitrary number of clustered decision trees, and so applies to both a single decision tree as well as $\hat{\mu}_I^{\text{MC}}$ as in~\eqref{eq:MC}. We observe that for our splitting properties~\ref{ass:tree} - that are equivalent in the $d=1$ case - we obtain an extra factor of two in the bias convergence rate over the analogous result of~\citet[Lemma 2]{wager}. 
This faster bias convergence rate, alongside an upper bound on the variance of the random forest estimator, allows us to obtain minimax optimal rates in pointwise mean squared error for Lipschitz functions~\citep{stone} for certain decision tree splitting regimes for (clustered) random forests $\hat{\mu}_I^{\text{MC}}$ in unclustered i.i.d.~settings (see Remark~\ref{remark:minimax}).  

\begin{proof}[Proof of Lemma~\ref{lem:unbiasedness}]
    Let $B\in\mathbb{N}$ be the number of trees that make up the random forest $\hat{\mu}_I$. We begin by considering an individual tree $T(x)$, for which we drop the notational dependence on $b\in[B]$ for brevity.   
     For an individual tree $T(x)$, let $M$ be the number of leaves of the tree, and $(L_m)_{m\in[M]}$ be the rectangular region that makes up each leaf. Also for a given $x\in\cX$ define $L(x)$ to be the leaf to which $x$ belongs, and let $J(x)$ be the index for the leaf $L(x)$ i.e.~$J(x)$ is the unique value of $m$ satisfying $L(x)=L_m$. 

    Then the tree predictor at $x\in\cX$ is given by
    \begin{align*}
        T(x) &= \mathrm{e}_{J(x)}^\top \big(\chi^\top \hat{W}\chi\big)^{-1} \big(\chi^\top \hat{W}\cY\big)
        \\
        &=
        \mathrm{e}_{J(x)}^\top \Big(\sum_{i=1}^I\chi_i^\top W_i(\hat{\rho})\chi_i\Big)^{-1}\Big(\sum_{i=1}^I\chi_i^\top W_i(\hat{\rho})Y_i\Big)
        ,
    \end{align*}
    where
    \begin{equation*}
        \chi := \begin{pmatrix}\chi_1\\\vdots\\\chi_I\end{pmatrix}
        ,\quad
        \mathcal{Y} := \begin{pmatrix}Y_1\\\vdots\\Y_I\end{pmatrix}
        ,\quad
        \hat{W} := \diag\big(W_i(\hat{\rho})\ind_{(i\in\Ieval)}\big)_{i\in[I]}.
    \end{equation*}
    \begin{equation*}
        (\chi_i)_{j,m} = \ind_{(X_{ij}\in L_m)}\ind_{(i\in\Ieval)},
        \quad
        (i\in[I],j\in[n_i],m\in[M]).
    \end{equation*}
    
    Now for any $x\in\cX$, define $\gamma(x) := \E_P[Y^*\given X^*\in L(x), Z_{\text{split}}]$, where the data $(Y^*,X^*)\in\R\times\cX$ is such that $(Y^*,X^*)\independent Z_{\text{split}}$. Note therefore by construction $\gamma(x')=\gamma(x'')$ for all $x',x''\in J(x)$. 
    Now for each $m\in[M]$ define $\gamma_{m} := \gamma(x)$ for any $x\in L_m$. We then have
    \begin{align}
        & \qquad \E_P[T(x)\given Z_{\text{split}},Z_{\text{corr}},X_{\text{eval}}] - \mu(x)
        \notag
        \\
        & = \mathrm{e}_{J(x)}^\top \big(\chi^\top \hat{W} \chi)^{-1} (\chi^\top \hat{W} \{\mu(X_{\text{eval}})-\chi\boldsymbol{\gamma}\} \big)
        +
        \{\gamma(x)-\mu(x)\}
        \notag
        \\
        & = \sum_{i\in\Ieval}\sum_{j=1}^{n_i}\omega_{ij}(x)\theta_{ij}
        + \{\gamma(x)-\mu(x)\},
        \label{eq:bias-decomp}
    \end{align}
    where 
    $\boldsymbol{\gamma}:=(\gamma_1,\gamma_2,\ldots,\gamma_M)$, and
    \begin{align}
        \omega_{ij}(x) &:= e_{ij}^\top \hat{W}\chi(\chi^\top \hat{W}\chi)^{-1}\mathrm{e}_{J(x)},
        \label{eq:omega(x)}
        \\
        \theta_{ij} &:= \mu(X_{ij}) - \gamma_{J(X_{ij})}.
        \notag
    \end{align}
    We work on the event $\Omega_I$ given by Lemma~\ref{lem:diam} which holds with probability at least $1-s_N^{-3}$ (see also Lemma~12 of~\citet{wager-walther}). 
    By Lemma~\ref{lem:diam} on $\Omega_I$ and for sufficiently large $I$,
    \begin{multline}\label{eq:bias2}
        \big|\gamma(x)-\mu(x)\big|
        = 
        \big|\E_P[Y^*\given X^*\in L(x), Z_{\text{split}}] - \E_P[Y^*\given X^*=x]\big|
        \\
        \leq
        L_{\mu}d^{1/2} \,\diam L(x)
        \leq 
        L_{\mu}d^{1/2} \Big(\frac{s_N}{2k-1}\Big)^{-(1-\epsilon)\frac{\log((1-\alpha)^{-1})}{\log(\alpha^{-1})}\cdot\frac{\pi}{d}}
        ,
    \end{multline}
    which follows by Lemma~\ref{lem:diam}.
    
    The former term of~\eqref{eq:bias-decomp} that arises is precisely zero for standard random forests, but typically non-zero for clustered random forests. First, note that
    \begin{align*}
        |\theta_{ij}| &= \big|\mu(X_{ij})-\gamma_{J(X_{ij})}\big|
        \\
        &= 
        |\E_P[Y^*\given X^* = X_{ij}] - \E_P[Y^*\given X^*\in L(X_{ij}), Z_{\text{split}}]|
        \\
        &\leq L_{\mu}d^{1/2}\,\diam L(X_{ij}),
    \end{align*}
    and so, also on the event $\Omega_I$,
    \begin{align}
        \Big|\sum_{i=1}^I
        \sum_{j=1}^{n_i}\omega_{ij}(x)\theta_{ij}\Big|
        &\leq \Big|\sum_{i\in\cI_{\eval}}
        \sum_{j=1}^{n_i}\omega_{ij}(x)\theta_{ij}\Big|
        \notag
        \\
        &\leq \sum_{i\in\Ieval}\sum_{j=1}^{n_i}|\omega_{ij}(x)||\theta_{ij}|
        \notag
        \\
        &\leq
        L_{\mu}d^{1/2} \sum_{i\in\Ieval}\sum_{j=1}^{n_i}|\omega_{ij}(x)| \, \diam L(X_{ij})
        \notag
        \\
        &\leq L_{\mu}d^{1/2} \Big(\frac{s_N}{2k-1}\Big)^{-(1-\epsilon)\frac{\log((1-\alpha)^{-1})}{\log(\alpha^{-1})}\cdot\frac{\pi}{d}} \sum_{i\in\Ieval}\sum_{j=1}^{n_i}|\omega_{ij}(x)|
        \notag
        \\
        &\leq
        L_{\mu}d^{1/2}C_1 \Big(\frac{s_N}{2k-1}\Big)^{-(1-\epsilon)\frac{\log((1-\alpha)^{-1})}{\log(\alpha^{-1})}\cdot\frac{\pi}{d}},
        \label{eq:bias1}
    \end{align}
    where the penultimate inequality follows by Lemma~\ref{lem:diam}, and the final inequality by Lemma~\ref{lem:sum-abs}, where $C_1\in\R_+$ is as given in Lemma~\ref{lem:sum-abs}. 
    Combining~\eqref{eq:bias2} and~\eqref{eq:bias1} with the triangle inequality we therefore have that
    \begin{equation}\label{eq:tree-bias}
        \big| \E_P[T(x)\given Z_{\text{split}},X_{\text{eval}}] - \mu(x) \big|
        \leq
        L_{\mu}d^{1/2}(1+C_1) \Big(\frac{s_N}{2k-1}\Big)^{-(1-\epsilon)\frac{\log((1-\alpha)^{-1})}{\log(\alpha^{-1})}\cdot\frac{\pi}{d}},
    \end{equation}
    Now consider the clustered random forest consisting of $B$ trees i.e.~$\hat{\mu}_I(x)=B^{-1}\sum_{b=1}^B T_b(x)$, and where we denote $Z_{b,\text{split}}$ to be the split data for specifically the $b$th tree, and similarly $X_{b,\text{eval}}$ etc. Then for each tree $T_b$ we have
    \begin{equation}\label{eq:tree-bias-b}
        \big| \E_P[T_b(x)\given Z_{b,\text{split}},X_{b,\text{eval}}] - \mu(x) \big|
        \leq
        L_{\mu}d^{1/2}(1+C_1) \Big(\frac{s_N}{2k-1}\Big)^{-(1-\epsilon)\frac{\log((1-\alpha)^{-1})}{\log(\alpha^{-1})}\cdot\frac{\pi}{d}},
    \end{equation}
    almost surely by the above. Therefore
    \begin{align*}
        \big| \E_P[\hat{\mu}(x)] - \mu(x) \big|
        &\leq
        B^{-1}\sum_{b=1}^B \big|\E_P[T_b(x)]-\mu(x)\big|
        \\
        &=
        B^{-1}\sum_{b=1}^B \big|\E_P\big[\E_P[T_b(x)\given Z_{b,\text{split}},X_{b,\text{eval}}]-\mu(x)\big]\big|
        \\
        &\leq
        B^{-1}\sum_{b=1}^B \E_P\Big[\big|\E_P[T_b(x)\given Z_{b,\text{split}},X_{b,\text{eval}}]-\mu(x)\big|\Big]
        \\
        &\leq
        L_{\mu}d^{1/2}(1+C_1) \Big(\frac{s_N}{2k-1}\Big)^{-(1-\epsilon)\frac{\log((1-\alpha)^{-1})}{\log(\alpha^{-1})}\cdot\frac{\pi}{d}},
    \end{align*}
    where the penultimate inequality follows by Jensen's inequality and the final inequality by applying~\eqref{eq:tree-bias-b}. The universal constants $L_{\mu}d^{1/2},C_1,k,\alpha,\pi$ carry no $x\in\cX$ or $P\in\cP$ dependence, so
    \begin{equation*}
        \sup_{P\in\cP}\sup_{x\in\cX}\big| \E_P[\hat{\mu}(x)] - \mu(x) \big|
        \leq
        L_{\mu}d^{1/2}(1+C_1)\Big(\frac{s_N}{2k-1}\Big)^{-(1-\epsilon)\frac{\log((1-\alpha)^{-1})}{\log(\alpha^{-1})}\cdot\frac{\pi}{d}}
        ,
    \end{equation*}
    completing the proof.
\end{proof}

\begin{lemma}\label{lem:diam}
    Fix an abitrarily small $\epsilon\in(0,1)$. Let $T$ be a tree satisfying Assumption~\ref{ass:tree}, built such that $s$ observations are used in the evaluation step. Let $L(x)$ be the leaf of $T$ that contains the point $x\in\cX$. Then with probability at least $1-s^{-3}$ and for large enough $s$,
    \begin{equation*}
        \sup_{x\in\cX}\mathrm{diam}\, L(x)
        \leq
        \Big(\frac{s}{2k-1}\Big)^{-(1-\epsilon)\frac{\log((1-\alpha)^{-1})}{\log(\alpha^{-1})}\cdot\frac{\pi}{d}}.
    \end{equation*}
\end{lemma}

\begin{proof}
    Fix an $x\in\cX=[0,1]^d$. Let $c(x)$ be the number of splits leading to leaf $L(x)$, and $c_j(x)$ be the number of these splits along the $j\text{th}$ feature. 
    Also let $L_j(x)$ be the interval in the $j$th covariate direction of the rectangular leaf $L(x)$, so that $L(x)=\Pi_{j=1}^d L_j(x)$. 
    By Assumption~\ref{regularity} we have that $s\alpha^{\frac{d}{\pi} c(x)}\leq 2k-1$, and thus $c(x)\geq\big(\log(\alpha^{-1})\big)^{-1}\log\big(\frac{s}{2k-1}\big)$ almost surely. Then by Assumption~\ref{feature} it follows $c_j(x)\geq\frac{\pi}{d}c(x)\geq\frac{\pi}{d}\big(\log(\alpha^{-1})\big)^{-1}\log\big(\frac{s}{2k-1}\big)$ almost surely. Further, by  Lemma~12 of~\citet{wager-walther} (see also Lemma~2 of~\citet{wager}), for large enough $s$ and with probability at least $1-s^{-3}$ we have $\sup_{x\in\cX}\frac{\nu^{-1}\diam L_j(x)}{(1-\alpha)^{(1-\epsilon)c_j(x)}}\leq 1$. Therefore, with probability at least $1-s^{-3}$ and for large enough $s$,
    \begin{equation*}
        \diam L_j(x)
        \leq
       \nu(1-\alpha)^{(1-\epsilon)\frac{\log(\frac{s}{2k-1})}{\log(\alpha^{-1})}\cdot\frac{\pi}{d}}
        =
        \nu \Big(\frac{s}{2k-1}\Big)^{-(1-\epsilon)\frac{\log((1-\alpha)^{-1})}{\log(\alpha^{-1})}\cdot\frac{\pi}{d}}.
    \end{equation*}
    As this upper bound is independent of $x\in\cX$, $P\in\cP$ and $j\in[d]$, the result follows;
    \begin{equation*}
        \sup_{P\in\cP}\sup_{x\in\cX}\diam L(x) = \sup_{P\in\cP}\sup_{x\in\cX}\max_{j\in[d]}\diam L_j(x) \leq \nu\Big(\frac{s}{2k-1}\Big)^{-(1-\epsilon)\frac{\log((1-\alpha)^{-1})}{\log(\alpha^{-1})}\cdot\frac{\pi}{d}}.
    \end{equation*}
\end{proof}

\begin{lemma}\label{lem:sum-abs}
    Let the weights $\big(\omega_{ij}(x)\big)_{i\in[I],j\in[n_i]}$ be as defined in~\eqref{eq:omega(x)} in the proof of Lemma~\ref{lem:unbiasedness} for a clustered decision tree satisfying property~\ref{ass:weights}. Then there exists a universal constant $C_1>0$ such that
    \begin{equation*}
        \sup_{P\in\cP}\sup_{x\in\cX}\sum_{i=1}^I\sum_{j=1}^{n_i} |\omega_{ij}(x)| \leq C_1 .
    \end{equation*}
\end{lemma}

\begin{proof}
    Fix some $P\in\cP$ and $x\in\cX$. We first show that ${\bf 1}^\top \omega(x) = 1$. This holds as 
    \begin{align*}
        \sum_{i=1}^I\sum_{j=1}^{n_i}\omega_{ij}(x) &= {\bf 1}_N^\top \omega(x)
        \\
        &= ({\bf 1}_N^\top \hat{W}\chi)(\chi^\top \hat{W}\chi)^{-1}\mathrm{e}_{J(x)}
        \\
        &= {\bf 1}_M^\top (\chi^\top \hat{W}\chi)(\chi^\top \hat{W}\chi)^{-1}\mathrm{e}_{J(x)}
        \\
        &= {\bf 1}_M^\top \mathrm{e}_{J(x)} = 1,
    \end{align*} 
    where in the penultimate line we use that ${\bf 1}_N = \chi {\bf 1}_M$, and where $M$ is the number of leaves that make up the tree defining $\omega_{ij}$. Now define $(d_{ij}(x))_{i\in[I],j\in[n_i]}$ by $d_{ij}(x) = \frac{1-\sgn(\omega_{ij}(x))}{2}\ind_{(\omega_{ij}(x)\neq 0)}$. Note therefore that $d_{ij}(x)\in\{0,1\}$ with $d_{ij}(x)=0$ if $i\notin\Ieval$. 
    Then
    \begin{align*}
        \sum_{i=1}^I\sum_{j=1}^{n_i} |\omega_{ij}(x)|
        &= \sum_{i\in\Ieval}\sum_{j=1}^{n_i} \omega_{ij}(x) \sgn(\omega_{ij}(x)) \ind_{(\omega_{ij}(x)\neq 0)}
        \\
        &= \sum_{i\in\Ieval}\sum_{j=1}^{n_i} \omega_{ij}(x)  +  \sum_{i\in\Ieval}\sum_{j=1}^{n_i} \omega_{ij}(x)\big(\sgn(\omega_{ij}(x))-1\big) \ind_{(\omega_{ij}(x)\neq 0)}    
        \\
        &= 1 - 2\sum_{i\in\Ieval}\sum_{j=1}^{n_i} \omega_{ij}(x) d_{ij}(x)
        \\
        &= 1 + 2\,\bigg|\sum_{i\in\Ieval}\sum_{j=1}^{n_i} \omega_{ij}(x) d_{ij}(x)\bigg|.
    \end{align*}
    Now define $\bar{\chi}_i \in \{0,1\}^{n_i \times M}$ for each $i\in[I]$ and $\bar{\chi}$ by
    \begin{equation*}
        \big(\bar{\chi}_i\big)_{j,m} = \ind_{(X_{ij}\in L_m)}\ind_{(\omega_{ij}(x) < 0)}\ind_{(i\in\Ieval)} 
        \qquad\text{and}\qquad
        \bar{\chi} := \begin{pmatrix}\bar{\chi}_1\\\vdots\\\bar{\chi}_I\end{pmatrix}.
    \end{equation*}
    Then by construction we have that $\text{vec}(d(x)) = \bar{\chi}{\bf{1}}_M$. Therefore
    \begin{align*}
        \bigg|\sum_{i\in\Ieval}\sum_{j=1}^{n_i} \omega_{ij}(x) d_{ij}(x)\bigg|
        &=
        \big| \mathrm{e}_{J(x)}^\top (\chi^\top \hat{W}\chi)^{-1} (\chi^\top \hat{W}\bar{\chi}) {\bf 1}_M \big|
        \\
        &\leq
        \underbrace{\|{\bf 1}_M\|_{\infty}}_{=1} \big\|(\bar{\chi}^\top \hat{W}\chi)(\chi^\top \hat{W}\chi)^{-1}\mathrm{e}_{J(x)}\big\|_1
        \\
        &\leq
        \big\|(\bar{\chi}^\top \hat{W}\chi)(\chi^\top \hat{W}\chi)^{-1}\big\|_1 \underbrace{\|\mathrm{e}_{J(x)}\|_1}_{=1}
        \\
        &\leq
        \big\|\bar{\chi}^\top \hat{W}\chi\big\|_1 \big\|(\chi^\top \hat{W}\chi)^{-1}\big\|_1.
    \end{align*}
Now
\begin{equation*}
    \big\|\bar{\chi}^\top \hat{W}\chi\big\|_1
    \leq
    \underbrace{\|\bar{\chi}\|_1}_{\in\{0,1\}} \|\hat{W}\|_1 \underbrace{\|\chi\|_{\infty}}_{\leq 2k-1}
    \leq
    (2k-1) \|\hat{W}\|_1,
\end{equation*}
following by H\"{o}lder's inequality and Assumption~\ref{regularity}, and where for a vector $v\in\R^n$ we write the norm $\|v\|_p:=(\sum_i|v_i|^p)^{1/p}$ and for a matrix $M\in\R^{n\times n}$ write the operator norm $\|M\|_p:=\sup_{v\in\R^n:\|v\|_p=1}\|Mv\|_p$. Now denote the entries of the block diagonal matrix $\hat{W}$ by $\hat{W}_{(i,j),(i',j')}$, so that
\begin{equation}\label{eq:hatW-notation}
    \hat{W}_{(i,j),(i',j')} =
    (W_i(\hat{\rho}))_{jj'}\cdot\ind_{(i=i')}
\end{equation}
Then we also have
\begin{align}
    \big\|(\chi^\top \hat{W}\chi)^{-1}\big\|_1
    &=
    \big\|(\chi^\top \hat{W}\chi)^{-1}\big\|_{\infty}
    \label{eq:dd1}
    \\
    &\leq
    \Big(\min_{m\in[M]}\Big(\big|\mathrm{e}_m^\top \chi^\top \hat{W}\chi \mathrm{e}_m\big|-\Big|\sum_{m'\neq m}e_{m'}^\top \chi^\top \hat{W}\chi \mathrm{e}_m\Big|\Big)\Big)^{-1}
    \label{eq:dd2}
    \\
    &\leq
    \Big(\min_{m\in[M]}\Big(\Big|\sum_{(i,j)}\ind_{(X_{ij}\in L_m)}\hat{W}_{(i,j),(i,j)}\Big| - \Big|\sum_{(i,j):X_{ij}\in L_m}\sum_{(i',j')\neq(i,j)} \hat{W}_{(i,j),(i',j')}\Big|\Big)\Big)^{-1}
    \notag
    \\
    &=
    \Big(\min_{m\in[M]}\Big(\sum_{(i,j):X_{ij}\in L_m} \Big\{\hat{W}_{(i,j),(i,j)}
    - \sum_{(i',j')\neq(i,j)}\big|\hat{W}_{(i,j),(i',j')}\big|\Big\}\Big)\Big)^{-1}
    \notag
    \\
    &\leq
    \frac{1}{k\,\tilde{\nu}(\hat{W})} \leq k^{-1}c_W^{-1},
    \label{eq:dd3}
\end{align}
where each step in the above is justified as follows:~\eqref{eq:dd1} follows by symmetry of $(\chi^\top\hat{W}\chi)^{-1}$; \eqref{eq:dd2} follows by the Ahlberg--Nilson--Varah bound for diagonally dominant matrices, which can be applied as
\begin{equation*}
    \tilde{\nu}(\hat{W}) := \min_{i\in\cI_{\text{eval}}, j\in[n_i]}\Big((W_i(\hat{\rho}))_{jj'}-\sum_{j'\neq j}\big|(W_i(\hat{\rho}))_{jj'}\big|\Big) \geq c_W,
\end{equation*}
by Assumption~\ref{ass:weights} (see also Lemma~\ref{lem:diag-dom}); and \eqref{eq:dd3} follows by Assumption~\ref{regularity}. 
Therefore
    \begin{multline*}
        \sum_{i=1}^I\sum_{j=1}^{n_i}|\omega_{ij}(x)|
        =
        1 + 2\|\omega(x)^\top \bar{\chi}{\bf 1}_M\|
        \leq 
        1 + 2\|\bar{\chi}^\top \hat{W}\chi\|_1 \|(\chi^\top \hat{W}\chi)^{-1}\|_1
        \\
        \leq
        1 + \frac{2(2k-1)}{k}\cdot\frac{\|\hat{W}\|_1}{\zeta\hat{W})}
        \leq
        1+\frac{2(2k-1)}{k}\cdot\frac{C_W}{c_W}
        \leq \frac{5C_W}{c_W}
        =: C_1,
    \end{multline*}
    completing the proof. Note in particular that $C_1$ by construction is a universal constant i.e.~carries no $x\in\cX$ nor $P\in\cP$ dependence.

\end{proof}

\begin{remark}\label{remark:weight-class-l2}
    Note that Assumption~\ref{ass:weights} on the weight structure is satisfied for inverse correlation structures that take the form of equicorrelated structures with finite group size, and inverse $\text{AR}(1)$ processes of potentially diverging group sizes (see Lemma~\ref{lem:diag-dom}). 
    Note however the assumption~\ref{ass:weights} may be relaxed in a number of cases. In particular in the case $d=1$ and $k\sim1$ 
    the condition on weights $W$ can be weakened to the $\ell_2$ condition number of the weight matrix $\kappa_2(W):=\frac{\Lambda_{\max}(W)}{\Lambda_{\min}(W)}$ being bounded. This relaxation follows by observing
    \begin{align*}
    &\quad\sum_{i=1}^I \sum_{j=1}^{n_i}
    |\omega_{ij}(x)| \, \diam L(X_{ij})
    \\
    &\leq
    \bigg(\sum_{i}\sum_{j=1}^{n_i}\omega_{ij}(x)^2\bigg)^{\frac{1}{2}}\bigg(\sum_{i}\sum_{j=1}^{n_i}\diam L(X_{ij})^2\bigg)^{\frac{1}{2}}
    \\
    &\leq \|\omega\|_2^2\bigg(\sum_{i}\sum_{j=1}^{n_i}\diam L(X_{ij})^2\bigg)^{\frac{1}{2}}
    \\
    &\leq \frac{\Lambda_{\max}(W)}{k\Lambda_{\min}(W)}\bigg(\sum_{i}\sum_{j=1}^{n_i}\diam L(X_{ij})^2\bigg)^{\frac{1}{2}}
    \\
    &\leq \frac{\Lambda_{\max}(W)}{k\Lambda_{\min}(W)}\left(\frac{s_N}{2k-1}\right)^{-\frac{1}{2}(1-\epsilon)\frac{\log((1-\alpha)^{-1})}{\log(\alpha^{-1})}\cdot\frac{\pi}{d}}\bigg(\sum_{i}\sum_{j=1}^{n_i}\diam L(X_{ij})\bigg)^{\frac{1}{2}}
    \\
    &= \frac{\Lambda_{\max}(W)}{k\Lambda_{\min}(W)}\left(\frac{s_N}{2k-1}\right)^{-\frac{1}{2}(1-\epsilon)\frac{\log((1-\alpha)^{-1})}{\log(\alpha^{-1})}\cdot\frac{\pi}{d}}\bigg(\sum_{m=1}^M\sum_{(i,j):X_{ij}\in L_m}\diam L_m\bigg)^{\frac{1}{2}}
    \\
    &\leq \frac{\sqrt{2k-1}\,\Lambda_{\max}(W)}{k\,\Lambda_{\min}(W)}\left(\frac{s_N}{2k-1}\right)^{-\frac{1}{2}(1-\epsilon)\frac{\log((1-\alpha)^{-1})}{\log(\alpha^{-1})}\cdot\frac{\pi}{d}}\bigg(\sum_{m=1}^M\diam L_m\bigg)^{\frac{1}{2}}
    \\
    &= \frac{\sqrt{2k-1}\,\Lambda_{\max}(W)}{k\,\Lambda_{\min}(W)}\left(\frac{s_N}{2k-1}\right)^{-\frac{1}{2}(1-\epsilon)\frac{\log((1-\alpha)^{-1})}{\log(\alpha^{-1})}\cdot\frac{\pi}{d}},
\end{align*}
which makes use of the inequality
\begin{multline}\label{eq:omega2}
        \|\omega(x)\|_2^2
        = \mathrm{e}_{J(x)}^\top (\chi^\top W\chi)^{-1} (\chi^\top W^2\chi) (\chi^\top W\chi)^{-1} \mathrm{e}_{J(x)}
        \leq \Lambda_{\max}(W) \mathrm{e}_{J(x)}^\top (\chi^\top W\chi)^{-1} \mathrm{e}_{J(x)}
        \\
        \leq \frac{\Lambda_{\max}(W)}{\Lambda_{\min}(W)\sigma_{\min}^2(\chi^\top )}
        \leq \frac{\Lambda_{\max}(W)}{\Lambda_{\min}(W)} \cdot\frac{1}{k}.
\end{multline}
Note however that whilst these weight structures still provide control over the bias, this is at a slower rate than in Lemma~\ref{lem:unbiasedness} by a factor of two and so the result of Theorem~\ref{thm:minimax} only would be shown here at a cost of two in the bias convergence rate, and the results of Theorem~\ref{thm:normality} follow for this wider weight class for $\beta>1-\Big(1+\frac{d}{2\pi}\frac{\log(\alpha^{-1})}{\log((1-\alpha)^{-1})}\Big)^{-1}$.
\end{remark}

\subsection{Interpreting Assumption~\ref{ass:weights} in the context of popular working correlation models}\label{appsec:equi-ar1}

In the following we show that the general weight structure property of Assumption~\ref{ass:weights} holds for equicorrelated and $\text{AR}(1)$ working correlation structures.

\begin{lemma}\label{lem:diag-dom}
    Suppose the clustered random forest weight structure $(W_i(\rho))_{\rho\in\Gamma}$ take either of the forms:
    \begin{enumerate}[label=(\roman*)]
        \item Equicorrelated structure, with finite group size 
        i.e.~$W_i(\rho) = \big(I_{n_i}+\rho{\bf1}_{n_i}{\bf1}_{n_i}^\top\big)^{-1}$ with $\Gamma\subset[0,\infty)$ bounded. 
        \item Autoregressive $\mathrm{AR}(1)$ structure, with potentially diverging group size 
        i.e. $W_i(\rho)=\big(\rho^{|j-k|}\big)^{-1}_{(j,k)\in[n_i]\times[n_i]}$ with $\Gamma\subset(-1+t,1-t)$ for some $t>0$.
    \end{enumerate}
    Then in either case there exists universal constants $c_W,C_W,C_W'>0$ such that
    \begin{equation*}
        \sup_{\rho\in\Gamma}\|W_i(\rho)\|_1 \leq C_W
        ,\qquad
        \inf_{\rho\in\Gamma}\zeta(W_i(\rho))\geq c_W
        ,\quad\;\text{and}\quad\;
        \sup_{\rho\in\Gamma}\Lambda_{\max}\bigg(\frac{\partial W_i(\rho)}{\partial\rho}\bigg)\leq C_W'
    \end{equation*}
    Recall the function $\nu$ is as defined in Assumption~\ref{ass:weights}.
\end{lemma}
\begin{proof}
    As $W=\diag(W_i(\rho))$ is block diagonal, with blocks $W_{(i,\cdot),(i,\cdot)}=W_i(\rho)$ for each $i\in[I]$ (recall the notation as in~\eqref{eq:hatW-notation}), it follows that
    \begin{equation*}
        \frac{\|W\|_1}{\zeta(W)} \leq \frac{\max_{i\in[I]}\|W_{(i,\cdot),(i,\cdot)}\|_1}{\min_{i\in[I]}\nu\big(W_{(i,\cdot),(i,\cdot)}\big)}
        =
        \frac{\max_{i\in[I]}\|W_i(\rho)\|_1}{\min_{i\in[I]}\nu\big(W_i(\rho)\big)}
        .
    \end{equation*} 
    We consider each case in turn.
    
    \medskip
    \noindent{\bf Case (i): Equicorrelated, finite group size: }
    For the equicorrelated structure, whose correlation matrix takes the form $\big(\ind_{\{j=k\}}+\ind_{\{j\neq k\}}\frac{\rho}{1+\rho}\big)_{(j,k)\in[n_i]^2}$, the resulting weights (proportional to the inverse of the correlation matrix) takes the form
    \begin{equation*}
        W_{(i,\cdot),(i,\cdot)} = \left(\ind_{\{j=k\}} - \frac{\rho}{1 + \rho n_i}\right)_{(j,k)\in[n_i]^2}.
    \end{equation*}
    Thus
    \begin{align*}
        \|W_{(i,\cdot),(i,\cdot)}\|_1 
        &= \frac{|1+(n_i-1)\rho|+(n_i-1)|\rho|}{|1+n_i\rho|}
        = \frac{1+2(n_i-1)\rho\ind_{\{\rho>0\}}}{|1+n_i\rho|},
        \\
        \nu\big(W_{(i,\cdot),(i,\cdot)}\big) &= \frac{|1+(n_i-1)\rho|-(n_i-1)|\rho|}{|1+n_i\rho|}
        =
        \frac{1+2(n_i-1)\rho\ind_{\{\rho<0\}}}{|1+n_i\rho|},
    \end{align*}
    noting that for $W$ to be positive definite it is necessary that $1+n_i\rho>0$, and so $1+(n_i-1)\rho>0$. 
    Then
        \begin{multline*}
        \frac{\|W\|_1}{\zeta(W)}
        \leq
        \frac{\max_{i\in[I]}\|W_{(i,\cdot),(i,\cdot)}\|_1}{\min_{i\in[I]}\nu\big(W_{(i,\cdot),(i,\cdot)}\big)}
        \leq
        \frac{\max_{i\in[I]}|1+n_i\rho|}{\min_{i\in[I]}|1+n_i\rho|}\cdot\frac{\max_{i\in[I]}\big(1+2(n_i-1)\rho\ind_{\{\rho>0\}}\big)}{\min_{i\in[I]}\big(1+2(n_i-1)\rho\ind_{\{\rho<0\}}\big)},
    \end{multline*}
    which is finite whenever $\rho>0$ and  $n_i$ is finite. 
    Further,
    \begin{equation*}
        \bigg\{\frac{\partial W_i(\rho)}{\partial\rho}\bigg\}^2 = \frac{n_i}{(1+\rho n_i)^4}{\bf1}_{n_i}{\bf1}_{n_i}^\top,
    \end{equation*}
    and so
    \begin{equation*}
        \Lambda_{\max}\bigg(\bigg\{\frac{\partial W}{\partial\rho}\bigg\}^2\bigg)
        \leq
        \max_{i\in[I]}\Lambda_{\max}\bigg(\bigg\{\frac{\partial W_i(\rho)}{\partial\rho}\bigg\}^2\bigg)
        =
        \frac{n_i^2}{(1+\rho n_i)^4},
    \end{equation*}
    which is finite.
    
    \medskip
    \noindent{\bf Case (ii): Autoregressive $\text{AR}(1)$, potentially diverging group size: }
    For the autocorrelated structure whose correlation matrix takes the form $(\rho^{|j-k|})_{(j,k)\in[n_i]^2}$, the resulting weights (proportional to the inverse of the correlation matrix) takes the form
    \begin{equation*}
    (W_{(i,j),(i,k)})_{(j,k)\in[n_i]^2}
        =\begin{cases}
        \left( \ind_{\{j=k\}} + \rho^2\ind_{\{1<j=k<n_i\}} - \rho \ind_{\{|j-k|=1\}} \right)_{(j,k)\in[n_i]^2} & \text{if } n_i > 1
        \\
        (1-\rho^2) &\text{if } n_i = 1
        \end{cases},
    \end{equation*}
    and so
    \begin{equation*}
        \|W_{(i,\cdot),(i,\cdot)}\|_1 
        =
        \begin{cases}
            \big(1+|\rho|\big)^2 &\text{if }n_i>2
            \\
            1+|\rho| &\text{if }n_i=2
            \\
            1-\rho^2 &\text{if }n_i=1
        \end{cases}
    ,\qquad
        \nu\big(W_{(i,\cdot),(i,\cdot)}\big)
        =
        \begin{cases}
            \big(1-|\rho|\big)^2 &\text{if }n_i>2
            \\
            1-|\rho| &\text{if }n_i=2
            \\
            1-\rho^2 &\text{if }n_i=1
        \end{cases},
    \end{equation*}
    and thus
    \begin{equation*}
        \frac{\|W\|_1}{\zeta(W)}
        \leq
        \frac{\max_{i\in[I]}\|W_{(i,\cdot),(i,\cdot)}\|_1}{\min_{i\in[I]}\nu\big(W_{(i,\cdot),(i,\cdot)}\big)}
        \leq
        \frac{\max\big((1+|\rho|)^2,\,1+|\rho|,\,1-\rho^2\big)}{\min\big((1-|\rho|)^2,\,1-|\rho|,\,1-\rho^2\big)}
        \leq
        \bigg(\frac{1+|\rho|}{1-|\rho|}\bigg)^2,
    \end{equation*}
    in all cases. Note this result holds irrespective of cluster sizes $n_i$, and thus clusters can be of diverging size. Further
    \begin{equation*}
        \bigg(\frac{\partial W_i(\rho)}{\partial\rho}\bigg)_{(j,k)\in[n_i]^2} = \begin{cases}
            \big(2\rho\ind_{\{1<j=k<n_i\}}-\ind_{\{|j-k|=1\}}\big)_{(j,k)\in[n_i]^2} &\quad\text{if $n_i>1$}
            \\
            (-2\rho) &\quad\text{if $n_i=1$}
        \end{cases},
    \end{equation*}
    and so, by the Gershgorin circle theorem
    \begin{equation*}
        \bigg|\Lambda_{\max}\bigg(\frac{\partial W_i(\rho)}{\partial\rho}\bigg)\bigg| \vee \bigg|\Lambda_{\min}\bigg(\frac{\partial W_i(\rho)}{\partial\rho}\bigg)\bigg| \leq 1+2|\rho| \leq 3,
    \end{equation*}
    and thus
    \begin{equation*}
        \Lambda_{\max}\bigg(\bigg\{\frac{\partial W}{\partial\rho}\bigg\}^2\bigg)
        \leq
        \max_{i\in[I]}\Lambda_{\max}\bigg(\bigg\{\frac{\partial W_i(\rho)}{\partial\rho}\bigg\}^2\bigg)
        \leq 9,
    \end{equation*}
    which again is finite.
\end{proof}

\subsection{Proof of Theorem~\ref{thm:minimax}}\label{appsec:proof-thm-minimax}

Theorem~\ref{thm:minimax} incorporates the bias bound of Lemma~\ref{lem:unbiasedness} with an existing variance bound, which follows from results by~\citet{efron} (also see~\citet{dynkin, vitale, cevid}). 
For simplicity of exposition from hereon, and without loss of generality, we consider the simplified setting where, in place of Assumption~\ref{honesty}, all $s_I$ subsampled data points in each tree are used solely for the prediction / evaluation step (using instead auxiliary data for splits and weight estimation); the relevant additional terms are absorbed into the $\epsilon>0$ term in Theorem~\ref{thm:minimax} (see also e.g.~\citet[Corollary 6]{wager}). We also denote the Monte-Carlo limit random forest $\hat{\mu}_I^{\mathrm{MC}}$ by $\hat{\mu}$ for simplicity.

\begin{proof}[Proof of Theorem~\ref{thm:minimax}]
    Recall the bias--variance decomposition
    \begin{equation*}
        \E_P\Big[\big(\hat{\mu}(x)-\mu(x)\big)^2\Big]
        =
        \big\{\E_P[\hat{\mu}(x)]-\mu(x)\big\}^2
        +
        \Var_P\big(\hat{\mu}(x)\big).
    \end{equation*}
    Let $\Omega_I$ be the event of Lemma~\ref{lem:unbiasedness} for which on $\Omega_I$ we have for sufficiently large $I$,
    \begin{equation*}
        \supP\big\{\E_P[\hat{\mu}(x)]-\mu(x)\big\}^2 \leq C_{\mathrm{bias}}^2\biggl(\frac{n_{\mathrm{c}}s_I}{k}\biggr)I^{-2(1-\epsilon)\frac{\log((1-\alpha)^{-1})}{\log(\alpha^{-1})}\frac{\pi}{d}\beta},
    \end{equation*}
    and $\supP\PP_P(\Omega_I^c)\leq s_N^{-3}$ for sufficiently large $I$. It thus remains to prove 
    the relevant upper bound on the variance quantity, which follows by standard results from~\citet{ANOVA}; we include an overview here for completeness. Fix a $P\in\cP$. Then each individual decision tree prediction $T(x)$ making up $\hat{\mu}(x)$ admits the Efron--Stein ANOVA decomposition~\citep{efron}
\begin{equation*}
    T(x) = \E_P[T(x)] + \sum_{k=1}^{s_I}
    \sum_{i_1<\cdots<i_k}T^{(k)}(Z_{i_1},\ldots,Z_{i_k}),
\end{equation*}
for some symmetric functions $(T^{(k)})_{k=1,\ldots,s}$ such that $\big(T^{(k)}(Z_{i_1},\ldots,Z_{i_k})\big)_{k\in\{1,2,\ldots,s\}}$ are all mean-zero, uncorrelated functions. 
Note in particular that
\begin{equation}\label{eq:T^{(1)}}
    T^{(1)}(Z_i) = \E_P[T(x)\given Z_i]-\E_P[T(x)].
\end{equation}
Similarly the Efron--Stein ANOVA decomposition for the aggregated random forest predictor $\hat{\mu}(x)$ is
\begin{equation*}
    \hat{\mu}(x) = \E_P[T(x)] + \binom{I}{s_I}^{-1} \sum_{k=1}^{s_I}
        \binom{I-k}{s_I-k} \sum_{i_1<\cdots<i_k}
        T^{(k)}(Z_{i_1},\ldots,Z_{i_k}).
\end{equation*}
Note in particular, such a decomposition holds even if the tree $T$ is not symmetric and $(Z_1,\ldots,Z_I)$ are independent but not identically distributed (see~\citet[Comment 4]{ANOVA}). Then
\begin{align*}
    \Var_P\big(\hat{\mu}(x)\big)
    &=
    \binom{I}{s_I}^{-2}\sum_{k=1}^{s_I}\binom{I-k}{s_I}^2\Var_P\bigg(\sum_{i_1<\cdots<i_k}T^{(k)}(Z_{i_1},\ldots,Z_{i_k})\bigg)
    \\
    &= \frac{s_I^2}{I^2}\Var_P\bigg(\sum_{i=1}^I T^{(1)}(Z_i)\bigg)
    +
    \binom{I}{s_I}^{-2}\sum_{k=2}^{s_I}\binom{I-k}{s_I-k}^2\Var_P\bigg(\sum_{i_1<\cdots<i_k}T^{(k)}(Z_{i_1},\ldots,Z_{i_k})\bigg)
    \\
    &\leq
    \frac{s_I^2}{I^2}\Var_P\bigg(\sum_{i=1}^I T^{(1)}(Z_i)\bigg)
    +
    \frac{s_I^2}{I^2}\sum_{k=2}^{s_I}\Var_P\bigg(\sum_{i_1<\cdots<i_k}T^{(k)}(Z_{i_1},\ldots,Z_{i_k})\bigg)
    \\
    &\leq
    \frac{s_I^2}{I^2}\Var_P\bigg(\sum_{i=1}^I T^{(1)}(Z_i)\bigg)
    +
    \frac{s_I^2}{I^2}\sum_{k=1}^{s_I}\Var_P\bigg(\sum_{i_1<\cdots<i_k}T^{(k)}(Z_{i_1},\ldots,Z_{i_k})\bigg)
    \\
    &\leq
    \frac{s_I^2}{I^2}\Var_P\bigg(\sum_{i=1}^I T^{(1)}(Z_i)\bigg)
    +
    \frac{s_I^2}{I^2}\Var_P\bigg(\sum_{k=1}^{s_I}\sum_{i_1<\cdots<i_k}T^{(k)}(Z_{i_1},\ldots,Z_{i_k})\bigg)
    \\
    &=
    \frac{s_I^2}{I^2}\sum_{i=1}^I\Var_P\big(\E_P[T(x)\given Z_i]\big)
    +
    \frac{s_I^2}{I^2}\Var_P\big(T(x)\big)
    .
\end{align*}
For identically distributed data, and by symmetry of the decision trees, we have that
\begin{equation*}
    \frac{s_I^2}{I^2}\sum_{i=1}^I\Var_P\big(\E_P[T(x)\given Z_i]\big)
    \leq
    \frac{s_I}{I}\Var_P\big(T(x)\big),
\end{equation*}
see for example~\citet{dynkin, vitale}~\citet[Lemma 9]{cevid}. 
Therefore
\begin{equation*}
    \Var_P\big(\hat{\mu}(x)\big) \leq \bigg(\frac{s_I}{I}+\frac{s_I^2}{I^2}\bigg)\Var_P\big(T(x)\big)
    \leq
    \frac{(1+\epsilon')s_I}{I}\Var_P\big(T(x)\big),
\end{equation*}
for some arbitrarily small $\epsilon'>0$ and for sufficiently large $I$. 

We now claim that $\sup_{P\in\cP}\sup_{x\in\cX}\Bigl|\frac{\Var_P(T_{\hat\rho}(x))}{\Var_P(T_{\rho^*}(x))}-1\Bigr|=o(1)$. For every $\rho\in\Gamma$, and writing $X=(X_{i})_{i\in\cI_{\mathrm{eval}}}$,
\begin{align*}
    &\quad\;
    \sup_{x\in\cX}\bigl|\Var_P\bigl(T_{\rho}(x)\bigr)-\E_P\bigl[\Var_P\bigl(T_{\rho}(x)\biggiven X\bigr)\bigr]\bigr|
    \\
    &=
    \sup_{x\in\cX}\Var_P\bigl(\E_P[T_{\rho}(x)\given X]\bigr)
    \\
    &\leq
    \sup_{x\in\cX}\E_P\bigl(\{\E_P[T_{\rho}(x)-\mu(x)\given X]\}^2\bigr)
    \\
    &\leq
    L_\mu^2d(1+C_1)^2\Bigl(\frac{s_N}{2k-1}\Bigr)^{-2(1-\epsilon)\frac{\log((1-\alpha)^{-1})}{\log(\alpha^{-1})}\cdot\frac{\pi}{d}}.
\end{align*}
Moreover
\begin{align*}
    &
    \quad\;
    \biggl|\frac{\partial}{\partial\rho}\E_P\bigl[\Var_P\bigl(T_{\rho}(x)\biggiven X\bigr)\bigr]\biggr|
    \\
    &=
    \biggl|\frac{\partial}{\partial\rho}\E_P\bigl[
    \mathrm{e}_{J(x)}^\top(\chi^\top W\chi)^{-1}(\chi^\top W\Sigma W\chi)(\chi^\top W\chi)^{-1}\mathrm{e}_{J(x)}
    \bigr]\biggr|
    \\
    &=
    \bigl|\E_P\bigl[
    \mathrm{e}_{J(x)}^\top A^{-1}BA^{-1}CA^{-1}\mathrm{e}_{J(x)}
    +
    \mathrm{e}_{J(x)}^\top A^{-1}CA^{-1}BA^{-1}\mathrm{e}_{J(x)}
    \\
    &\qquad
    +
    \mathrm{e}_{J(x)}^\top A^{-1}DA^{-1}\mathrm{e}_{J(x)}
    \bigr]\bigr|
    \\
    &\leq
    \E_P\bigl[
    \mathrm{e}_{J(x)}^\top A^{-1}(B+C)A^{-1}(B+C)A^{-1}\mathrm{e}_{J(x)}
    +
    \mathrm{e}_{J(x)}^\top A^{-1}BA^{-1}BA^{-1}\mathrm{e}_{J(x)}
    \\
    &\qquad
    +
    \mathrm{e}_{J(x)}^\top A^{-1}CA^{-1}CA^{-1}\mathrm{e}_{J(x)}
    +
    \mathrm{e}_{J(x)}^\top A^{-1}DA^{-1}\mathrm{e}_{J(x)}
    \bigr]
\end{align*}
where
\begin{equation*}
    A := \chi^\top W\chi,
    \quad
    B := \chi^\top\frac{\partial W}{\partial\rho}\chi,
    \quad
    C := \chi^\top W\Sigma W \chi,
    \quad
    D := \chi^\top \biggl(\frac{\partial W}{\partial\rho}\Sigma W + W\Sigma\frac{\partial W}{\partial\rho}\biggr) \chi.
\end{equation*}
Because
\begin{gather*}
    \Lambda_{\min}(A)\geq c_Wk,
    \quad
    \Lambda_{\max}(B) \leq C_W'(2k-1),
    \\
    \Lambda_{\max}(C) \leq \lambda C_W^2(2k-1),
    \quad
    \lambda_{\max}(D) \leq 2\lambda C_WC_W'(2k-1),
\end{gather*}
almost surely, and $$\E_P\bigl[\Var_P\bigl(T_{\rho}(x)\bigr)\biggiven X\bigr]\geq 
\E_P\bigl[\Lambda_{\min}(W^{1/2}\Sigma W^{1/2})\mathrm{e}_{J(x)}^\top A^{-1} \mathrm{e}_{J(x)}\biggiven X\bigr]
\geq \lambda^{-1}c_WC_W^{-1}
(2k-1)^{-1},$$ it follows that for any $\bar{\rho},\tilde{\rho}$, 
\begin{multline*}
    \biggl(\E_P\bigl[\Var_P\bigl(T^{(\bar{\rho})}(x)\biggiven X\bigr)\bigr]\biggr)^{-1}
    \biggl(\biggl|\frac{\partial}{\partial\rho}\E_P\bigl[\Var_P\bigl(T_{\rho}(x)\biggiven X\bigr)\bigr]\biggr|_{\rho=\tilde{\rho}}\biggr)
    \\
    \leq \frac{\lambda C_W(2k-1)}{c_W}\biggl(c_W^{-3}\frac{(2k-1)^2}{k^3}\bigr[(\lambda C_W^2+C_W')^2+(C_W')^2+\lambda^2C_W^4\bigr]+\frac{(2k-1)}{k^2}2c_W^{-2}\lambda C_WC_W'\biggr)
    \lesssim 1.
\end{multline*}
Thus by Lemmas~\ref{lem:VarT(x)}~and~\ref{lem:unif-rho},
\begin{equation*}
    \sup_{P\in\cP}\sup_{x\in\cX}\biggl|\frac{\Var_P(T_{\hat\rho}(x))}{\Var_P(T_{\rho^*}(x))}-1\biggr| = o(1).
\end{equation*}

Combining the above variance bound with Lemma~\ref{lem:unbiasedness} gives the required result;
\begin{equation*}
    \sup_{P\in\cP}\E_P\Big[\big(\hat{\mu}_I^{\mathrm{MC}}(x)-\mu(x)\big)^2\Big]
    \leq
    C_{\mathrm{bias}}^2 \biggl(\frac{s_N}{2k_I-1}\biggr)^{-2(1-\epsilon)\frac{\log((1-\alpha)^{-1})}{\log(\alpha^{-1})}\frac{\pi}{d}} +\, I^{-1}s_I\,\Var_P\bigl(T_{\rho^*}(x)\bigr), 
\end{equation*}
where the term involving $\epsilon'$ and the term on the event $\Omega_I$ is incorporated into the $\epsilon>0$ term for sufficiently large $I$. 
Finally for $I\geq I_0$ (note that $I_0$ carries no $x$ dependence),
\begin{multline*}
    \int_\cX\E_P\Bigl[\bigl(\hat{\mu}_I^{\mathrm{MC}}(x)-\mu(x)\bigr)^2\Bigr]dQ(x)
    \\
    \leq
    Q(\cX) C_{\mathrm{bias}}^2 \biggl(\frac{s_N}{2k_I-1}\biggr)^{-2(1-\epsilon)\frac{\log((1-\alpha)^{-1})}{\log(\alpha^{-1})}\frac{\pi}{d}}
    + Q(\cX) I^{-1}s_I \sup_{x\in\cX}\Var_P\bigl(T_{\rho^*}(x)\bigr)
    <\infty,
\end{multline*}
where $Q(\cX)=1$ and $k_I\sup_{x\in\cX}\Var_P(T_{\rho^*}(x))<\infty$ by Lemma~\ref{lem:VarT(x)}. Thus applying Fubini's Theorem gives the required result.

\end{proof}

\subsection{Proof of Theorem~\ref{thm:dist-shift-2}}

We adopt the same setup and notation of Section~\ref{appsec:proof-thm-minimax}. Contained within this section, we will use $s$ as a shorthand for $s_I$. 

\begin{proof}[Proof of Theorem~\ref{thm:dist-shift-2}]

Recall that $T_{\rho}(\cdot)$ is a the clustered random forest with fixed working correlation weighting $\rho\in\Gamma$. 
First note that
\begin{align*}
    &\quad\;
    \sup_{x\in\cX}\bigl|\Var_P\bigl(T_{\rho}(x)\bigr)-\E_P\bigl[\Var_P\bigl(T_{\rho}(x)\biggiven X_{\mathrm{eval}},Z_{\mathrm{split}}\bigr)\bigr]\bigr|
    \\
    &=
    \sup_{x\in\cX}\Var_P\bigl(\E_P[T_{\rho}(x)\given X_{\mathrm{eval}},Z_{\mathrm{split}}]\bigr)
    \\
    &\leq
    \sup_{x\in\cX}\E_P\bigl(\{\E_P[T_{\rho}(x)-\mu(x)\given X_{\mathrm{eval}},Z_{\mathrm{split}}]\}^2\bigr)
    \\
    &\leq
    L_\mu^2d(1+C_1)^2\Bigl(\frac{s_N}{2k-1}\Bigr)^{-2(1-\epsilon)\frac{\log((1-\alpha)^{-1})}{\log(\alpha^{-1})}\cdot\frac{\pi}{d}} 
    \\
    &\leq k^{-1}\epsilon_1,
\end{align*}
for all $I\geq I_1$, for sufficiently large $I_1\in\mathbb{N}$, where the penultimate inequality follows by Lemma~\ref{lem:unbiasedness}. We therefore consider from hereon a fixed arbitrary set of splits, and treat expectations throughout as conditional on the splits. Moreover, for any $x\in\cX$,
\begin{equation*}
    \E_P\bigl[\Var_P(T_{\rho}(x)\given X_{\mathrm{eval}},Z_{\mathrm{split}})\bigr]
    =
    \E_P\bigl[\mathrm{e}_{J(x)}^\top (\chi^\top W\chi)^{-1}(\chi^\top W\Sigma W\chi)(\chi^\top W\chi)^{-1}\mathrm{e}_{J(x)}\bigr],
\end{equation*}
thus
\begin{align*}
    \int_\cX\Var_P\bigl(T_{\rho}(x)\bigr)dQ(x)
    &
    \int_\cX\E_P\bigl[\mathrm{e}_{J(x)}^\top(\chi^\top W\chi)^{-1}(\chi^\top W\Sigma W\chi)(\chi^\top W\chi)^{-1}\mathrm{e}_{J(x)}\bigr]dQ(x) + o\big(k^{-1}\big)
    \\
    &=
    (1+o(1))\E_P\bigl[\tr\bigl\{\diag(Q(L_m))^\top(\chi^\top W\chi)^{-1}(\chi^\top W\Sigma W\chi)(\chi^\top W\chi)^{-1}\bigr\}\bigr],
\end{align*}
with the final equality following because, as in Lemma~\ref{lem:VarT(x)},
\begin{equation*}
    \int_{\cX}\E_P\bigl[\mathrm{e}_{J(x)}^\top(\chi^\top W\chi)^{-1}(\chi^\top W\Sigma W\chi)(\chi^\top W\chi)^{-1}\mathrm{e}_{J(x)}\bigr]dQ(x)
    \geq
    \frac{\Lambda_{\min}(W^{1/2}\Sigma W^{1/2})}{k\Lambda_{\max}(W)}\gtrsim \frac{1}{k}.
\end{equation*}

Now fix $\varrho\in(0,1)$ and $\kappa\geq1$. Also fix some arbitrary $0\leq a_1<a_2<b_1<b_2\leq1$, then take $Q_1=\text{Unif}[a_1,a_2]$, $Q_2=\text{Unif}[b_1,b_2]$, and define $P$ via the data generating mechanism $X_{ij}\iid\text{Unif}[0,1]$ and 
$$Y_i\given X_i \iid N\biggl(
\begin{pmatrix}
    \mu(X_{i1})\\\mu(X_{i2})
\end{pmatrix},\;
\begin{pmatrix}
    \sigma^2(X_{i1}) & \varrho\sigma(X_{i1})\sigma(X_{i2})
    \\
    \varrho\sigma(X_{i1})\sigma(X_{i2}) & \sigma^2(X_{i2})
\end{pmatrix}
\biggr),$$
where $\sigma(x):=\ind_{[0,a_2]}(x)+\frac{(\eta-1)x+(b_1-\eta a_2)}{b_1-a_2}\ind_{(a_2,b_1)}(x)+\eta\ind_{[b_1,1]}(x)$ for some sufficiently large $\eta>0$; specifically fix some $\eta\geq 2\varrho^{-1}\kappa^{1/2}$. st will help to introduce the notation
\begin{equation*}
    \Sigma:=\diag(\Sigma_i)_{i\in[s]},
    \qquad
    \Sigma_i := \begin{pmatrix}
        \sigma_{i1}^2 & \varrho\sigma_{i1}\sigma_{i2}
        \\
        \varrho\sigma_{i1}\sigma_{i2} & \sigma_{i2}^2
    \end{pmatrix},
\end{equation*}
with shorthand $\sigma_{ij}:=\sigma(X_{ij})$. Weights take the form of the block diagonal matrix 
\begin{equation*}
    W:=\diag(W_i)_{i\in[s]},
    \qquad
    W_i=\begin{pmatrix}
    1&-\rho\\-\rho&1
\end{pmatrix},
\end{equation*}
proportional to the inverse equicorrelated weight matrix in terms of some correlation parameter $\rho\in[0,1)$.

Let $M$ be the number of leaves, and denote the leaves $L_1,\ldots,L_M$ in $T$, with $L(x)$ being the leaf to which a covariate $x\in[0,1]$ belongs. For each $m\in[M]$, let $p_m:=P(L_m)=\Pr_P(X_{ij}\in L_m)$ and $q_m:=Q(L_m)=\Pr_{X'\sim Q}(X'\in L_m)$, and further write $\boldsymbol{p}:=(p_1,\ldots,p_M)$ and $\boldsymbol{q}:=(q_1,\ldots,q_M)$. Also take $\chi_i\in\{0,1\}^{2\times M}$ as in Algorithm~\ref{alg:cdt} and $\chi:=(\chi_1,\ldots,\chi_s)$. Then
\begin{multline*}
    k\int_\cX\E_P\bigl[\mathrm{e}_{J(x)}^\top(\chi^\top W\chi)^{-1}(\chi^\top W\Sigma W\chi)(\chi^\top W\chi)^{-1}\mathrm{e}_{J(x)}\bigr]dQ(x)
    \\
    =
    \int_\cX\E_P\bigl[\mathrm{e}_{J(x)}^\top T_s^{-1}S_sT_s^{-1}\mathrm{e}_{J(x)}\bigr]dQ(x)
    =\sum_{m=1}^Mq_m\E_P\bigl[\mathrm{e}_m^\top T_s^{-1}S_sT_s^{-1}\mathrm{e}_m\bigr],
\end{multline*}
where
\begin{equation*}
    T_s := k^{-1}\chi^\top W\chi = k^{-1}\sum_{i=1}^s\chi_i^\top W_i\chi_i,
    \qquad
    S_s := k^{-1}\chi^\top W \Sigma W \chi = k^{-1}\sum_{i=1}^s\chi_i^\top W_i\Sigma_iW_i\chi_i,
\end{equation*}
along with $T:=\E_P(T_s)$ and $S:=\E_P(S_s)$. 
By Chebyshev's inequality,
\begin{equation*}
    T_s=T+O_P\bigl(k^{-1/2}\bigr)=T+o_P(1),
    \qquad
    S_s=S+O_P\bigl(k^{-1/2}\bigr) = S+o_P(1).
\end{equation*}
Additionally
\begin{align*}
    \sum_{m=1}^Mq_m\E_P\bigl[\mathrm{e}_m^\top T_s^{-1}S_sT_s^{-1}\mathrm{e}_m\bigr] 
    &\leq
    k\sum_{m=1}^Mq_m\E_P\bigl[\mathrm{e}_m^\top (\chi^\top W \chi)^{-1}(\chi^\top W\Sigma W \chi) (\chi^\top W \chi)^{-1}\mathrm{e}_m\bigr]
    \\
    &\leq
    \frac{\E_P[\Lambda_{\max}(W^{1/2}\Sigma W^{1/2})]}{\Lambda_{\min}(W)}
    \leq
    \frac{1+\rho}{1-\rho}\eta^2(1+\varrho), 
\end{align*}
and so applying dominated convergence gives
\begin{equation*}
    k\int_\cX\Var_P\bigl(T_{\rho}(x)\bigr)dQ(x)
    =
    \mathcal{R}_Q(\rho) + o(1),
\end{equation*}
where
\begin{equation*}
    \mathcal{R}_Q(\rho)
    :=
    k\sum_{m=1}^Mq_m\mathrm{e}_m^\top T^{-1}ST^{-1}\mathrm{e}_m.
\end{equation*}
Note in the above the dependence on $Q$ is seen through $q_1,\ldots,q_m$. 
Moreover,
\begin{equation*}
    \mathcal{R}_Q(\rho)
    \geq
    k\Lambda_{\min}(S) \{\Lambda_{\max}(T)\}^{-2}\geq 
    \frac{(1-\rho)(1-\varrho)}{1+\rho}\cdot\frac{k^2}{(2k-1)^2}
    \gtrsim 1,
\end{equation*}
as $\Lambda_{\max}(T)\leq \E_P(\Lambda_{\max}(T_s))\leq (1+\rho)(2k-1)$ by Jensen's inequality, and $\Lambda_{\min}(S)\geq k(1-\rho^2)(1-\varrho)$. Therefore 
\begin{equation*}
    k\int_\cX\Var_P\bigl(T_{\rho}(x)\bigr)dQ(x)
    =(1+o(1))\mathcal{R}_Q(\rho).
\end{equation*}

We proceed to evaluate $\mathcal{R}_Q(\rho)$. For each $m\in[M]$,
\begin{multline}\label{eq:chie}
    \chi_i\mathrm{e}_m = \boldsymbol{1}_2\ind_{L_m}(X_{i1})\ind_{L_m}(X_{i2})
    +e_1\ind_{L_m}(X_{i1})\ind_{L_m^c}(X_{i2})
    \\
    +e_2\ind_{L_m^c}(X_{i1})\ind_{L_m}(X_{i2})
    +\boldsymbol{0}_2\ind_{L_m^c}(X_{i1})\ind_{L_m^c}(X_{i2}),
\end{multline}
and so
\begin{equation*}
    \mathrm{e}_m^\top \chi_i^\top W_i\chi_i \mathrm{e}_m
    =
    2(1-\rho)\ind_{L_m}(X_{i1})\ind_{L_m}(X_{i2})
    + \ind_{L_m}(X_{i1})\ind_{L_m^c}(X_{i2})
    + \ind_{L_m^c}(X_{i1})\ind_{L_m}(X_{i2}),
\end{equation*}
and further
\begin{equation}\label{eq:T-same-m}
    s^{-1}kT_{mm} = 2(1-\rho)p_m^2+2p_m(1-p_m)
    =
    2p_m-2\rho p_m^2.
\end{equation}
Moreover, for $m,\tilde{m}\in[M]$ with $m\neq\tilde{m}$,
\begin{equation*}
    \mathrm{e}_m^\top \chi_i^\top W_i\chi_i e_{\tilde{m}}
    =
    -\rho\ind_{L_m}(X_{i1})\ind_{L_{\tilde{m}}}(X_{i2})
    -\rho\ind_{L_{\tilde{m}}}(X_{i1})\ind_{L_m}(X_{i2})
    ,
\end{equation*}
thus
\begin{equation}\label{eq:T-diff-m}
    s^{-1}kT_{m\tilde{m}} = -2\rho p_m^2.    
\end{equation}
Combining~\eqref{eq:T-same-m}~and~\eqref{eq:T-diff-m}, it follows that
\begin{equation*}
    T = 2sk^{-1}\bigl(\diag(\boldsymbol{p})-\rho\,\boldsymbol{p}\boldsymbol{p}^\top\bigr).
\end{equation*}
By the Sherman--Morrison formula
\begin{equation*}
    T^{-1} = 
    \frac{k}{2s}\Bigl(\diag(\boldsymbol{p})^{-1}+\frac{\rho}{1-\rho}\,\boldsymbol{1}\boldsymbol{1}^\top\Bigr)
    ,
\end{equation*}
and so for $m\in[M]$,
\begin{equation*}
    T^{-1}\mathrm{e}_m = \frac{k}{2s}\Bigl(p_m^{-1}\mathrm{e}_m+\frac{\rho}{1-\rho}\boldsymbol{1}\Bigr),
\end{equation*}
thus
\begin{align}
    4s^2k^{-2} \mathrm{e}_m^\top T^{-1}ST^{-1}\mathrm{e}_m
    &=
    \Bigl(p_m^{-1}\mathrm{e}_m+\frac{\rho}{1-\rho}\boldsymbol{1}\Bigr)^\top S\Bigl(p_m^{-1}\mathrm{e}_m+\frac{\rho}{1-\rho}\boldsymbol{1}\Bigr)
    \notag
    \\
    &=
    p_m^{-2}\mathrm{e}_m^\top S\mathrm{e}_m + \frac{2\rho}{1-\rho}p_m^{-1}\mathrm{e}_m^\top S\boldsymbol{1} + \frac{\rho^2}{(1-\rho)^2}\boldsymbol{1}^\top S\boldsymbol{1}.
    \label{eq:TST-decomp}
\end{align}

We now evaluate each term in the decomposition~\eqref{eq:TST-decomp} separately. 
The first term in~\eqref{eq:TST-decomp} decomposes as
\begin{align}
    \boldsymbol{1}^\top S\boldsymbol{1}
    &=
    \E_P[\boldsymbol{1}^\top\chi^\top W\Sigma W\chi\boldsymbol{1}]
    =
   \E_P[ \boldsymbol{1}^\top W\Sigma W\boldsymbol{1} ]
    =
    (1-\rho)^2\E_P[\boldsymbol{1}^\top\Sigma\boldsymbol{1}]
    \notag\\
    &=
    (1-\rho)^2\sum_{i=1}^s\E_P\bigl[\sigma^2(X_{i1})+\sigma^2(X_{i2})+2\varrho\sigma(X_{i1})\sigma(X_{i2})\bigr]
    \notag\\
    &=
    2s (1-\rho)^2 \Bigl(\E_P[\sigma^2(X)]+\varrho\bigl(\E_P[\sigma(X)]\bigr)^2\Bigr)
    \notag\\
    &=
    \tfrac{s}{2} (1-\rho)^2 \Bigl(2(1+\eta^2)+\varrho(1+\eta)^2\Bigr).
    \label{eq:TST-decomp-1}
\end{align}
For the second term of~\eqref{eq:TST-decomp},
\begin{align*}
    \mathrm{e}_m^\top S\boldsymbol{1}
    =
    (1-\rho)\E_P[\mathrm{e}_m^\top\chi^\top W\Sigma\boldsymbol{1}]
    =
    (1-\rho)\sum_{i=1}^s\E_P[\mathrm{e}_m^\top\chi_i^\top W_i\Sigma_i\boldsymbol{1}]
    ,
\end{align*}
and further, utilizing~\eqref{eq:chie},
\begin{align*}
    \mathrm{e}_m^\top\chi_i^\top W_i\Sigma_i\boldsymbol{1}
    &=
    \mathrm{e}_m^\top\chi_i^\top \begin{pmatrix}
        \{\sigma_{i1}^2+\varrho\sigma_{i1}\sigma_{i2}\}-\rho\{\sigma_{i2}^2+\varrho\sigma_{i1}\sigma_{i2}\}
        \\
        \{\sigma_{i2}^2+\varrho\sigma_{i1}\sigma_{i2}\}-\rho\{\sigma_{i1}^2+\varrho\sigma_{i1}\sigma_{i2}\}
    \end{pmatrix}
    \\
    &=
    \begin{cases}
        (1-\rho)\big\{(\sigma_{i1}^2+\sigma_{i2}^2)+2\varrho\sigma_{i1}\sigma_{i2}\big\}
        &\quad\text{if }(X_{i1},X_{i2})\in L_m\times L_m
        \\
        \sigma_{i1}^2-\rho\sigma_{i2}^2+(1-\rho)\varrho\sigma_{i1}\sigma_{i2}
        &\quad\text{if }(X_{i1},X_{i2})\in L_m\times L_m^c
        \\
        \sigma_{i2}^2-\rho\sigma_{i1}^2+(1-\rho)\varrho\sigma_{i1}\sigma_{i2}
        &\quad\text{if }(X_{i1},X_{i2})\in L_m^c\times L_m
        \\
        0
        &\quad\text{if }(X_{i1},X_{i2})\in L_m^c\times L_m^c
    \end{cases}
    \\
    &=
    \bigl(\sigma_{i1}^2-\rho\sigma_{i2}^2+(1-\rho)\varrho\sigma_{i1}\sigma_{i2}\bigr)\ind_{L_m}(X_{i1})
    \\
    &\quad
    +\bigl(\sigma_{i2}^2-\rho\sigma_{i1}^2+(1-\rho)\varrho\sigma_{i1}\sigma_{i2}\bigr)\ind_{L_m}(X_{i2})
    ,
\end{align*}
and so
\begin{align*}
    \E_P[\mathrm{e}_m^\top\chi_i^\top W_i\Sigma_i\boldsymbol{1}]
    &=
    2\E_P\bigl[\bigl\{\sigma^2(X_1)-\rho\sigma^2(X_2)+(1-\rho)\varrho\sigma(X_1)\sigma(X_2)\bigr\}\ind_{L_m}(X_1)\bigr]
    \\
    &=
    2\Bigl(\E_P[\sigma^2(X_1)\ind_{L_m}(X_1)]-\rho\E_P[\sigma^2(X_2)]p_m
    \\
    &\qquad
    +(1-\rho)\varrho\E_P[\sigma(X_1)\ind_{L_m}(X_1)]\E_P[\sigma(X_2)]\Bigr)
    \\
    &=
    2p_m\Bigl(\varsigma_{\mathrm{sq}.m}-\rho\bigl\{\tfrac{1+\eta^2}{2}\bigr\}+(1-\rho)\varrho\varsigma_m\bigl\{\tfrac{1+\eta}{2}\bigr\}\Bigr)
    \\
    &=
    p_m\Bigl(2\varsigma_{\mathrm{sq}.m}-\rho(1+\eta^2)+(1-\rho)\varrho(1+\eta)\varsigma_m\bigr\}\Bigr),
\end{align*}
thus
\begin{equation}\label{eq:TST-decomp-2}
    k\,\mathrm{e}_m^\top S\boldsymbol{1}
    =
    s(1-\rho)p_m\Bigl(2\varsigma_{\mathrm{sq}.m}-\rho(1+\eta^2)+(1-\rho)\varrho(1+\eta)\varsigma_m\bigr\}\Bigr).
\end{equation}
For the third term in the decomposition~\eqref{eq:TST-decomp}, first note that
\begin{align*}
    \mathrm{e}_m^\top\chi_i^\top W_i\Sigma_iW_i \mathrm{e}_m 
    &=
    \begin{cases}
        (1-\rho)^2\bigl\{(\sigma_{i1}^2+\sigma_{i2}^2)+2\varrho\sigma_{i1}\sigma_{i2}\bigr\}
        &\quad\text{if }(X_{i1},X_{i2})\in L_m\times L_m
        \\
        \sigma_{i1}^2+\rho^2\sigma_{i2}^2-2\rho\varrho\sigma_{i1}\sigma_{i2}
        &\quad\text{if }(X_{i1},X_{i2})\in L_m\times L_m^c
        \\
        \sigma_{i2}^2+\rho^2\sigma_{i1}^2-2\rho\varrho\sigma_{i1}\sigma_{i2}
        &\quad\text{if }(X_{i1},X_{i2})\in L_m^c\times L_m
        \\
        0
        &\quad\text{if }(X_{i1},X_{i2})\in L_m^c\times L_m^c
    \end{cases}
    \\
    &=
    \bigl(\sigma_{i1}^2+\rho^2\sigma_{i2}^2-2\rho\varrho\sigma_{i1}\sigma_{i2}\bigr)\ind_{L_m}(X_{i1})
    \\
    &\qquad
    +\bigl(\sigma_{i2}^2+\rho^2\sigma_{i1}^2-2\rho\varrho\sigma_{i1}\sigma_{i2}\bigr)\ind_{L_m}(X_{i2})
    \\
    &\qquad
    +2\bigl\{(1+\rho^2)\varrho\sigma_{i1}\sigma_{i2}-\rho(\sigma_{i1}^2+\sigma_{i2}^2)\bigr\}\ind_{L_m}(X_{i1})\ind_{L_m}(X_{i2})
    ,
\end{align*}
and so
\begin{align}
    k\,\mathrm{e}_m^\top S\mathrm{e}_m = kS_{mm}
    &=
    s\,\E_P[\mathrm{e}_m^\top\chi_i^\top W_i\Sigma_iW_i \mathrm{e}_m]
    \notag\\
    &=
    2s\Bigl\{\E_P\big[\bigl\{\sigma^2(X_1)+\rho^2\sigma^2(X_2)-2\rho\varrho\sigma(X_1)\sigma(X_2)\bigr\}\ind_{L_m}(X_1)\big]
    \notag\\
    &\qquad
    +
    (1+\rho^2)\varrho\,\E_P[\sigma(X_1)\sigma(X_2)\ind_{L_m}(X_1)\ind_{L_m}(X_2)]
    \notag\\
    &\qquad
    -\rho\,\E_P[(\sigma^2(X_1)+\sigma^2(X_2))\ind_{L_m}(X_1)\ind_{L_m}(X_2)]
    \Bigr\}
    \notag\\
    &=
    2p_ms\Bigl(
    \varsigma_{\mathrm{sq}.m}+\rho^2\bigl\{\tfrac{1+\eta^2}{2}\bigr\}-2\rho\varrho\varsigma_m\bigl\{\tfrac{1+\eta}{2}\bigr\}
    \Bigr)
    +2p_m^2s\Bigl((1+\rho^2)\varrho\varsigma_m^2-2\rho\varsigma_{\mathrm{sq}.m}\Bigr).
    \label{eq:TST-decomp-3}
\end{align}
Combining~\eqref{eq:TST-decomp-1},~\eqref{eq:TST-decomp-2}, and~\eqref{eq:TST-decomp-3} in the decomposition~\eqref{eq:TST-decomp} gives
\begin{align*}
    4sk^{-1} \mathrm{e}_m^\top T^{-1}ST^{-1}\mathrm{e}_m
    &=
    p_m^{-1}\Bigl(2\varsigma_{\mathrm{sq}.m}+\rho^2(1+\eta^2)-2\rho\varrho\varsigma_m(1+\eta)\Bigr)
    +
    \Bigl(2(1+\rho^2)\varrho\varsigma_m^2
    \\&\qquad
    -\rho^2(1+\eta^2)+2\rho(1-\rho)\varrho(1+\eta)\varsigma_m
    +\rho^2\varrho(1+\eta)^2
    \Bigr).
\end{align*}
Then for all $\rho,\varrho\in[0,1)$,
\begin{align*}
    &\quad\;
    \Bigl|2(1+\rho^2)\varrho\varsigma_m^2
    -\rho^2(1+\eta^2)+2\rho(1-\rho)\varrho(1+\eta)\varsigma_m
    +\rho^2\varrho(1+\eta)^2
    \Bigr|
    \\
    &=
    4\varsigma_m^2 + 1+\eta^2 + \frac{1}{2}(1+\eta)|\varsigma_m| + (1+\eta)^2
    \\
    &\leq
    \frac{1}{2}(13\eta^2+5\eta+4)
    \\
    &\leq 11\eta^2,
\end{align*}
and so
\begin{equation*}
    \biggl|
    \mathrm{e}_m^\top T^{-1}ST^{-1}\mathrm{e}_m
    -
    (4p_msk^{-1})^{-1}\Bigl(2\varsigma_{\mathrm{sq}.m}+\rho^2(1+\eta^2)-2\rho\varrho\varsigma_m(1+\eta)\Bigr)
    \biggr|
    \leq 
    \tfrac{11\eta^2k}{4s},
\end{equation*}
and so
\begin{equation*}
    \mathrm{e}_m^\top T^{-1}ST^{-1}\mathrm{e}_m
    =
    \frac{k}{4p_ms}\Bigl(2\varsigma_{\mathrm{sq}.m}+\rho^2(1+\eta^2)-2\rho\varrho\varsigma_m(1+\eta)\Bigr) + O\bigl(s^{-1}k\bigr),
\end{equation*}
and thus
\begin{equation}\label{eq:decomp1}
    \sum_{m=1}^Mq_m\mathrm{e}_m^\top T^{-1}ST^{-1}\mathrm{e}_m 
    =
    \frac{k}{4s}\sum_{m=1}^M\frac{q_m}{p_m}\Bigl((1+\eta^2)\rho^2-2\varrho\varsigma_m(1+\eta)\rho+2\varsigma_{\mathrm{sq}.m}\Bigr) + O\bigl(s^{-1}k\bigr),
\end{equation}
noting that $\sum_{m=1}^Mq_m=1$. 

We now consider an arbitrary $Q=\mathrm{Unif}[c_1,c_2]$ for some $0\leq c_1<c_2\leq 1$ (as is the case for both $Q_1$ and $Q_2$). 
Partition the leaf index set into $[M]=\mathcal{L}_{\mathrm{in}}\cup\mathcal{L}_{\mathrm{out}}\cup\mathcal{L}_{\mathrm{both}}$, where $\mathcal{L}_{\mathrm{in}}:=\{m\in[M]:L_m\subseteq\supp Q\}$ indexes the leaves that are entirely sub-settable in $\supp Q$, $\mathcal{L}_{\mathrm{out}}:=\{m\in[M]:L_m\cap\supp Q=\emptyset\}$ indexes the leaves that are disjoint to $\supp Q$, and finally $\mathcal{L}_{\mathrm{both}}:=[M]\backslash(\mathcal{L}_{\mathrm{in}}\cup\mathcal{L}_{\mathrm{out}})$. 
Note that because $Q=[c_1,c_2]$ is an interval $|\mathcal{L}_{\mathrm{both}}|\leq 2$. Then
\begin{align*}
    &\quad\;
    \frac{k}{4s}\biggl|\sum_{m\in\mathcal{L}_{\mathrm{both}}}\frac{q_m}{p_m}\Bigl((1+\eta^2)\rho^2 - 2\varrho\varsigma_m(1+\eta)\rho + 2\varsigma_{\mathrm{sq}.m}\Bigr)\biggr|
    \\
    &\leq
    \frac{k}{4s}\sum_{m\in\mathcal{L}_{\mathrm{both}}}\frac{\diam(L_m\cap [c_1,c_2])}{(c_2-c_1)\diam(L_m)}\Bigl((1+\eta^2)+2\eta(1+\eta)+2\eta^2\Bigr)
    \\
    &\leq
    \frac{4\eta^2k}{(c_2-c_1)s}
    =O\bigl(s^{-1}k\bigr).
\end{align*}
Because also $q_m=0$ for all $m\in\mathcal{L}_{\mathrm{out}}$, combining~\eqref{eq:decomp1} with the above,
\begin{equation*}
    \mathcal{R}_Q(\rho) = \frac{k}{4s}\sum_{m\in\mathcal{L}_{\mathrm{in}}}\frac{q_m}{p_m}\Bigl((1+\eta^2)\rho^2 - 2\varrho\varsigma_m(1+\eta)\rho + 2\varsigma_{\mathrm{sq}.m}\Bigr)
    +O\bigl(s^{-1}k\bigr).
\end{equation*}
Moreover, if $\varsigma_m=\varsigma_Q$ and $\varsigma_{\mathrm{sq}.m}=\varsigma_{\mathrm{sq},Q}$ is constant for all $m\in\mathcal{L}_{\mathrm{in}}$ then
\begin{equation*}
    \mathcal{R}_Q(\rho) = \frac{\theta_Qk}{4s}\Bigl((1+\eta^2)\rho^2-2\varrho\varsigma_Q(1+\eta)\rho+2\varsigma_{\mathrm{sq},Q}\Bigr)+ O\bigl(s^{-1}k\bigr),
\end{equation*}
where
\begin{multline*}
    \theta_Q:=\sum_{m\in\mathcal{L}_{\mathrm{in}}}\frac{q_m}{p_m}=\frac{1}{c_2-c_1}|\mathcal{L}_{\mathrm{in}}|\geq \biggl(\biggl(\frac{s}{2k-1}\biggr)^{(1-\epsilon)\frac{\log((1-\alpha)^{-1})}{\log(\alpha^{-1})}\cdot\frac{\pi}{d}}-\frac{2}{c_2-c_1}\biggr)_+
    \\
    \gtrsim \bigl((\log s)^{-1}s\bigr)^{(1-\epsilon)\frac{\log((1-\alpha)^{-1})}{\log(\alpha^{-1})}\cdot\frac{\pi}{d}},
\end{multline*}
by Lemma~\ref{lem:diam}, where $u_+:=\max(u,0)$. Thus
\begin{equation*}
    \mathcal{R}_Q(\rho) = (1+o(1))\frac{\theta_Qk}{4s}\Bigl((1+\eta^2)\rho^2-2\varrho\varsigma_Q(1+\eta)\rho+2\varsigma_{\mathrm{sq},Q}\Bigr).
\end{equation*}
Now note that $\varsigma_{Q_1}=\varsigma_{\mathrm{sq},Q_1}=1$, $\varsigma_{Q_2}=\eta$, and $\varsigma_{\mathrm{sq},Q_2}=\eta^2$, and so
\begin{align*}
    \mathcal{R}_{Q_1}(\rho) &=
    (1+o(1))\frac{\theta_{Q_1}k}{4s}\Bigl((1+\eta^2)\rho^2-2\varrho(1+\eta)\rho+2\Bigr),
    \\
    \mathcal{R}_{Q_2}(\rho) &=
    (1+o(1))\frac{\theta_{Q_2}k}{4s}\Bigl((1+\eta^2)\rho^2-2\varrho\eta(1+\eta)\rho+2\eta^2\Bigr),
\end{align*}
each of which are minimised by
\begin{equation*}
    \rho^*_{Q_1\mathrm{-optimal}}
    := (1+o(1))\frac{1+\eta}{1+\eta^2}\varrho
    \quad\text{and}\quad
    \rho^*_{Q_2\mathrm{-optimal}}
    := (1+o(1))\frac{\eta(1+\eta)}{1+\eta^2}\varrho,
\end{equation*}
respectively. Therefore
\begin{align*}
    \frac{\int_\cX\Var_P\bigl(T^{(Q_2;\mathrm{optimal})}(x)\bigr)dQ_1(x)}{\int_\cX\Var_P\bigl(T^{(Q_1;\mathrm{optimal})}(x)\bigr)dQ_1(x)}
    &=(1+o(1))
    \frac{\mathcal{R}_{Q_1}\bigl(\rho^*_{Q_2\mathrm{-optimal}}\bigr)}{\mathcal{R}_{Q_1}\bigl(\rho^*_{Q_1\mathrm{-optimal}}\bigr)}
    \\
    &=
    (1+o(1))\frac{2(1+\eta^2)+\eta(\eta-2)(1+\eta)^2\varrho^2}{2(1+\eta^2)-(1+\eta)^2\varrho^2}
    \\
    &\geq
    (1+o(1))\biggl(1+\frac{\eta(\eta-2)(1+\eta)^2}{2(1+\eta^2)}\varrho^2\biggr)
    \\
    &\geq
    (1+o(1))\biggl(1+\frac{1}{4}\eta^2\varrho^2\biggr)
    \\
    &\geq
    (1+o(1))(1+\kappa)
    ,
\end{align*}
where the penultimate inequality follows from the inequality $\frac{u(u-2)(1+u)^2}{2(1+u^2)}\geq \frac{1}{4}u^2$ for all $u\geq 3$, and where we note $\eta\geq3$ by construction.

\begin{remark}

In the above construction we took $Q_1$ and $Q_2$ to have supports as two disjoint intervals. We may however form an analogous construction with $Q_1$ a point mass at some $x\in\cX$ and/or $Q_2$ the training covariate distribution $P_X$. 

   \smallskip\noindent
    {\bf Applicability to  pointwise estimation: }$Q_1=\delta_x$.
    
    Taking the same construction as above but with $[a_1,a_2]=[x-\epsilon,x+\epsilon]$ for some fixed sufficiently small $\epsilon$; specifically take any $x<b_1<b_2\leq 1$ and  $0<\epsilon<x\wedge(b_2-x)$. The result then follows by identical arguments to the above. 
    
    \smallskip\noindent
    {\bf Suboptimality of optimising training error: }$Q_2=P_X$.

    If $Q_2=P_X$ then because $\varsigma_m$ and $\varsigma_{\mathrm{sq}.m}$ are linear in $(1,\eta)$ and $(1,\eta^2)$ respectively and for all $m\in[M]$, it follows that
    \begin{equation*}
        \mathcal{R}_{P_X}(\rho)=(1+o(1))\frac{\theta_{P_X}k}{4s}\Bigl((1+\eta^2)\rho^2-2\varrho(1+\alpha_2\eta)(1+\eta)\rho+2(1+\alpha_3\eta^2)\Bigr),
    \end{equation*}
    for some constants $\alpha_2,\alpha_3\geq0$. Moreover, as the density of $P_X$ is bounded away from zero (Assumption~\ref{ass:data}) it follows that $\alpha_2,\alpha_3>0$. Then $\rho^*_{P_X\mathrm{optimal}}=(1+o(1))\frac{(1+\alpha_2\eta)(1+\eta)}{1+\eta^2}\varrho\geq (1+o(1))\alpha_2\varrho$,
    for a sufficiently large $\eta$ (that does not depend on $s$). The analogous result then follows.
\end{remark}

\end{proof}

\subsection{Proof of Proposition~\ref{prop:weight-estimation}}

We adopt the same setup and notation of Section~\ref{appsec:proof-thm-minimax}. Contained within this section, we will use $s$ as a shorthand for $s_I^{\mathrm{corr}}$. 

\begin{proof}[Proof of Proposition~\ref{prop:weight-estimation}]
    Fix an arbitrary $x\in\cX$ and define
    \begin{align*}
        \hat{\mathrm{R}}_x(\rho) &:= k\,\mathrm{e}_{J(x)}^\top\bigl(\tilde{\chi}^\top W(\rho)\tilde{\chi}\bigr)^{-1}\biggl(\sum_i\tilde{\chi}_i^\top W_i(\rho)\tilde{\varepsilon}_i\tilde{\varepsilon}_i^\top W_i(\rho)\tilde{\chi}_i\biggr)\bigl(\tilde{\chi}^\top W(\rho)\tilde{\chi}\bigr)^{-1}\mathrm{e}_{J(x)},
        \\
        \mathrm{R}_x(\rho)&:=k\,\Var_P\bigl(T_{\rho}(x)\bigr).
    \end{align*}
    Without loss of generality consider fixed splits (see e.g.~the proof of Theorem~\ref{thm:dist-shift-2}). 
    Define $\tilde{\omega}$ analogous to $\omega$ in the proof of Lemma~\ref{lem:unbiasedness} (with $\tilde{\chi}$ in place of $\chi$ and $W(\rho)$ in place of $\hat{W}=W(\hat\rho)$). For $i\in\cI_{\mathrm{corr}}$, $j\in[n_i]$, 
    \[\tilde{\varepsilon}_{ij}-\varepsilon_{ij}=\mu(X_{ij})-
    \hat{\mu}_{\mathrm{p.c.}}(X_{ij}) =: -\hat{e}(X_{ij}),
    \]
    where
    \[
    \hat{\mu}_{\mathrm{p.c.}}(x) = \hat\mu_{\mathrm{pilot.corr}}(x) := \mathrm{e}_{J(x)}^\top \Big(\sum_{i\in\cI_{\mathrm{corr}}}\tilde{\chi}_i^\top W_i(\rho)\tilde{\chi}_i\Big)^{-1} \Big(\sum_{i\in\cI_{\mathrm{corr}}}\tilde{\chi}_i^\top W_i(\rho)Y_i\Big)
    \]
    Thus
    \begin{align*}
        &\quad\;
        \bigg|k\,\mathrm{e}_{J(x)}^\top\bigl(\tilde{\chi}^\top W(\rho)\tilde{\chi}\bigr)^{-1}\biggl(\sum_{i}\tilde{\chi}_i^\top W_i(\rho)\bigl(\tilde{\varepsilon}_i\tilde{\varepsilon}_i^\top-\varepsilon_i\varepsilon_i^\top\bigr) W_i(\rho)\tilde{\chi}_i\biggr)\bigl(\tilde{\chi}^\top W(\rho)\tilde{\chi}\bigr)^{-1}\mathrm{e}_{J(x)}\bigg|
        \\
        &=
        k\bigg|
        \sum_{i}\bigg\{\sum_{j=1}^{n_i}\tilde{\omega}_{ij}(x)\hat{e}(X_{ij})\bigg\}^2
        -
        2\sum_{i}\bigg\{\sum_{j=1}^{n_i}\tilde{\omega}_{ij}(x)\hat{e}(X_{ij})\bigg\}\bigg\{\sum_{j=1}^{n_i}\tilde{\omega}_{ij}(x)\varepsilon_{ij}\bigg\}
        \bigg|
        \\
        &\leq
        k\bigg\{\sum_{i}\bigg(\sum_{j=1}^{n_i}|\tilde{\omega}_{ij}(x)|\bigg)^2\bigg\}\sup_{x\in\cX}\hat{e}^2(x)
        +
        2k\bigg\{\sum_{i}\sum_{j=1}^{n_i}\sum_{j'=1}^{n_i}|\tilde{\omega}_{ij}(x)|\cdot|\tilde{\omega}_{ij'}(x)|\cdot|\varepsilon_{ij}|\bigg\}\sup_{x\in\cX}|\hat{e}(x)|
        \\
        &\leq
        n_{\mathrm{c}} \sup_{x\in\cX}\hat{e}^2(x) + 2k\bigg\{\sum_{i}\sum_{j=1}^{n_i}\sum_{j'=1}^{n_i}|\tilde{\omega}_{ij}(x)|\cdot|\tilde{\omega}_{ij'}(x)|\cdot|\varepsilon_{ij}|\bigg\}\sup_{x\in\cX}|\hat{e}(x)|
        \\
        &=O_{P}\big(n_{\mathrm{c}}k^{-1/2}\big),
    \end{align*}
    uniformly in $\rho\in\Gamma$, and where we use~$\sup_{x\in\cX}\hat{e}^2(x)=o_P\big(k^{-1}+s^{-2\phi/d}\big)$ (Theorem~\ref{thm:minimax} for a fixed deterministic $\hat\rho$ e.g.~$\hat\rho=\rho^*=0$) alongside Markov's inequality and
    \begin{align*}
        \E_P\bigg\{\sum_{i}\sum_{j=1}^{n_i}\sum_{j'=1}^{n_i}|\tilde{\omega}_{ij}(x)\tilde{\omega}_{ij'}(x)\varepsilon_{ij}|\bigg\}
        &\leq
        \E_P\bigg\{\sum_{i}\sum_{j=1}^{n_i}\sum_{j'=1}^{n_i}|\tilde{\omega}_{ij}(x)\tilde{\omega}_{ij'}(x)|\E_P(|\varepsilon_{ij}|\,|\,Z_{\mathrm{splits}})\bigg\}
        \\
        &\leq
        \lambda^{1/2}\E_P\bigg\{\sum_{i}\bigg(\sum_{j=1}^{n_i}|\tilde{\omega}_{ij}(x)|\bigg)^2\bigg\}
        \\
        &\leq 
        \lambda^{1/2}n_{\mathrm{c}}\,\E_P\big(\|\tilde{\omega}(x)\|_2^2\big)
        \lesssim \frac{n_{\mathrm{c}}}{k},
    \end{align*}
    where we make use of~\eqref{eq:omega2}. 
    Therefore
    \begin{equation*}
        \hat{\mathrm{R}}_x(\rho) = 
        k\,\mathrm{e}_{J(x)}^\top\bigl(\tilde{\chi}^\top W(\rho)\tilde{\chi}\bigr)^{-1}\biggl(\sum_i\tilde{\chi}_i^\top W_i(\rho)\varepsilon_i\varepsilon_i^\top W_i(\rho)\tilde{\chi}_i\biggr)\bigl(\tilde{\chi}^\top W(\rho)\tilde{\chi}\bigr)^{-1}\mathrm{e}_{J(x)} + o_P(1),
    \end{equation*}
    uniformly in $\rho\in\Gamma$. 
    Moreover, by the von--Bahr--Esseen inequality,
    \begin{align*}
        \E_P\biggl\{\bigg|k\sum_{i}\big(\tilde{\omega}_i(x)^\top\varepsilon_i\big)^2\bigg|^{(2+\delta)/2}\biggr\}
        &\leq
        2k^{(2+\delta)/2}\sum_{i}\E_P\big[|\tilde{\omega}_i(x)^\top\varepsilon_i|^{2+\delta}\big]
        \\
        &\leq
        2\tau k^{(2+\delta)/2}\sum_{i}\big(\E_P\big[(\tilde{\omega}_i(x)^\top\varepsilon_i)^2\big]\big)^{(2+\delta)/2}
        \\
        &\leq 
        2\tau\lambda^{(2+\delta)/2}k^{(2+\delta)/2}\sum_{i}\big(\E_P\big[\|\tilde{\omega}_i(x)\|_2^2\big]\big)^{(2+\delta)/2}
        \\
        &=
        2\tau\lambda^{(2+\delta)/2}k^{(2+\delta)/2}s(k^{-1}s^{-1})^{(2+\delta)/2}
        \lesssim s^{-\delta/2},
    \end{align*}
    and so by Markov's inequality
    \begin{multline*}
        \sup_{\rho\in\Gamma}\sup_{x\in\cX}\Bigl|\hat{\mathrm{R}}_x(\rho) - \E_P\Big\{k\,\mathrm{e}_{J(x)}^\top\bigl(\tilde{\chi}^\top W(\rho)\tilde{\chi}\bigr)^{-1}\bigl(\tilde{\chi}^\top W(\rho)\Sigma W(\rho)\tilde{\chi}\bigr)\bigl(\tilde{\chi}^\top W(\rho)\tilde{\chi}\bigr)^{-1}\mathrm{e}_{J(x)}\Big\}\Bigr| 
        \\ 
        = O_P\bigl(n_{\mathrm{c}}k^{-1/2} + s^{-\delta/(2+\delta)}\bigr).
    \end{multline*}
    Also define
    \[
    \hat{\mathrm{R}}_Q(\rho) := k\sum_{m=1}^Mq_m\mathrm{e}_m^\top\bigl(\tilde{\chi}^\top W(\rho)\tilde{\chi}\bigr)^{-1}\biggl(\sum_i\tilde{\chi}_i^\top W_i(\rho)\tilde{\varepsilon}_i\tilde{\varepsilon}_i^\top W_i(\rho)\tilde{\chi}_i\biggr)\bigl(\tilde{\chi}^\top W(\rho)\tilde{\chi}\bigr)^{-1}\mathrm{e}_m,
    \]
    with $q_m:=Q(L_m)$. 
    As shown in the proof of Theorem~\ref{thm:dist-shift-2},
    \begin{equation*}
        \sup_{\rho\in\Gamma}\bigg|\int_{\cX}\mathrm{R}_x(\rho)dQ(x) - \mathcal{R}_Q(\rho)\bigg| = O\bigl(n_{\mathrm{c}}k^{-1/2}\bigr),
    \end{equation*}
    for $\mathcal{R}_Q(\rho)$ as in the proof of Theorem~\ref{thm:dist-shift-2}, which combined with the above gives
    \begin{equation*}
        \sup_{\rho\in\Gamma}\bigg|\hat{\mathrm{R}}_Q(\rho)-k\int_{\cX}\Var_P\bigl(T_{\rho}(x)\bigr)dQ(x)\bigg|
        = O_P\bigl(n_{\mathrm{c}}k^{-1/2} + s^{-\delta/(2+\delta)}\bigr),
    \end{equation*}
    and so the result follows with the given identifiability condition (in the statement of the theorem) and by~\citet[Theorem~3.2.5]{vandervaart-wellner}.
    
\end{proof}

\subsection{Proof of Theorem~\ref{thm:normality}}

In this section we study the asymptotic normality of $\hat{\mu}_I^{\textrm{MC}}$; as in the last subsection we continue to denote this by $\hat{\mu}$ 
throughout this section for notational simplicity. The asymptotic normality result of Theorem~\ref{thm:normality} relies on ideas stemming from the Efron--Stein ANOVA decomposition and H\'{a}jek projections of the random forests and decision trees~\citep{hoeffding, hajek, ANOVA}. 
For simplicity of exposition we also take identically distributed clusters (thus clusters are of equal, potentially diverging, size). As noted by~\citet[Comment 4]{ANOVA}, as the ANOVA decomposition holds more generally for non-identically distributed but independent data, the results can be extended to non-identically distributed data, albeit as a more notationally involved proof. The H\'{a}jek projections of $\hat{\mu}(x)$ and $\hat{\mu}_{\rho^*}(x)$ onto $\{Z_1,\ldots,Z_I\}$ are given by
\begin{align*}
    \accentset{\circ}{\hat{\mu}}(x) &:= \E_P[\hat{\mu}(x)]  +  \sum_{i=1}^I \big\{\E_P[\hat{\mu}(x)\given Z_i]-\E_P[\hat{\mu}(x)]\big\},
    \\
    \accentset{\circ}{\hat{\mu}}_{\rho^*}(x) &:= \E_P[\hat{\mu}_{\rho^*}(x)]  +  \sum_{i=1}^I \big\{\E_P[\hat{\mu}_{\rho^*}(x)\given Z_i]-\E_P[\hat{\mu}_{\rho^*}(x)]\big\}.
\end{align*}
Similarly, the H\'{a}jek projection of the trees $T(x;Z_{i_1},\ldots,Z_{i_s})$ and $T_{\rho^*}(x;Z_{i_1},\ldots,Z_{i_s})$ onto $\{Z_{i_1},\ldots,Z_{i_s}\}$ are given by
\begin{align*}
    \accentset{\circ}{T}(x) &:= \E_P[T(x)]  +  \sum_{i\in\{i_1,\ldots,i_s\}} \big\{\E_P[T(x)\given Z_i]-\E_P[T(x)]\big\},
    \\
    \accentset{\circ}{T}_{\rho^*}(x) &:= \E_P[T_{\rho^*}(x)]  +  \sum_{i\in\{i_1,\ldots,i_s\}} \big\{\E_P[T_{\rho^*}(x)\given Z_i]-\E_P[T_{\rho^*}(x)]\big\}.
\end{align*}

Note then
\begin{align*}
    \accentset{\circ}{\hat{\mu}}(x) &= \E_P[T(x)] + \frac{s }{I}\sum_{i=1}^I \big\{\E_P[T(x)\given Z_i]-\E_P[T(x)]\big\},
    \\
    \accentset{\circ}{\hat{\mu}}_{\rho^*}(x) &= \E_P[T_{\rho^*}(x)] + \frac{s }{I}\sum_{i=1}^I \big\{\E_P[T_{\rho^*}(x)\given Z_i]-\E_P[T_{\rho^*}(x)]\big\},
\end{align*}

\begin{proof}[Proof of Theorem~\ref{thm:normality}]
First define
\begin{multline*}
    \sigma_{P,I}^2(x;\rho^*) := \Var_P\big[\accentset{\circ}{\hat{\mu}}_{\rho^*}(x)\big]
    \\
    = 
    \Var_P\bigg[\sum_{i=1}^I \big\{\E_P[\hat{\mu}_{\rho^*}(x)\given Z_i]-\E_P[\hat{\mu}_{\rho^*}(x)]\big\}\bigg]
    =
    \sum_{i=1}^I \Var_P\Big[\E_P[\hat{\mu}_{\rho^*}(x)\given Z_i]\Big],
\end{multline*}
and consider the decomposition
\begin{multline*}
    \frac{\hat{\mu}(x)-\mu(x)}{\sigma_{P,I}(x;\rho^*)} = 
    \underbrace{\frac{\accentset{\circ}{\hat{\mu}}_{\rho^*}(x) - \E_P\big[\accentset{\circ}{\hat{\mu}}_{\rho^*}(x)\big]}{\sigma_{P,I}(x;\rho^*)}}_{=:A_P}
    +
    \underbrace{\frac{\hat{\mu}(x)-\accentset{\circ}{\hat{\mu}}(x)}{\sigma_{P,I}(x;\rho^*)}}_{=:B_P}
    +
    \underbrace{\frac{\E_P[\hat{\mu}(x)]-\mu(x)}{\sigma_{P,I}(x;\rho^*)}}_{=:C_P}
    \\
    +
    \underbrace{\frac{\accentset{\circ}{\hat{\mu}}(x)-\E_P\big[\accentset{\circ}{\hat{\mu}}(x)\big]-\accentset{\circ}{\hat{\mu}}_{\rho^*}(x)+\E_P\big[\accentset{\circ}{\hat{\mu}}_{\rho^*}(x)\big]}{\sigma_{P,I}(x;\rho^*)}}_{=:D_P},
\end{multline*}
noting that $\E_P\big[\accentset{\circ}{\hat{\mu}}(x)\big] = \E_P[\hat{\mu}(x)]$.

In brief, Term $A_P$ is a sum of independent, mean zero terms which is unit variance, and so under Lyapunov-type conditions can be shown to be asymptotically standard Gaussian. The additional terms will all be shown to be $o_\cP(1)$. Term $B_P$ relates to the projection of the random forest predictor onto the orthogonal complement of the image of the H\'{a}jek projection. Term $C_P$ is a (scaled) bias term, already shown to be sufficiently small (see Lemma~\ref{lem:unbiasedness}); for the range of $\beta$ as given in Assumption~\ref{ass:tree} the Term $C_P$ is controlled to be $o_\cP(1)$. Finally, Term $D_P$ would be precisely zero were one to set $\hat{\rho}=\rho^*$ a deterministic parameter. When estimated as in Algorithm~\ref{alg:crf}, Term $D_P$ is mean zero, with variance of order $(\hat{\rho}-\rho^*)^2$ up to logarithmic factors, and thus provided $\hat{\rho}$ converges in probability to $\rho^*$ at arbitrarily slow polynomial rates this term is $o_\cP(1)$. We now prove each of these claims in turn.

\medskip
\noindent {\bf Part 1:} Showing 
$\lim_{I\to\infty}\supP\sup_{t\in\R}|\PP_P(A_P\leq t)-\Phi(t)|=0$.

Recall that
\begin{align*}
    \accentset{\circ}{\hat{\mu}}(x) - \E_P[\hat{\mu}(x)]
    &=
    \sum_{i=1}^I \big\{\E_P[\hat{\mu}(x)\given Z_i] - \E_P[\hat{\mu}(x)]\big\}
    \\
    &=
    \frac{s }{I}\sum_{i=1}^I \big\{\E_P[T(x)\given Z_i] - \E_P[T(x)]\big\}.
\end{align*}
We proceed to show the Lyapunov-style condition
\begin{equation}\label{eq:lyapunov}
    \lim_{I\to\infty}\sup_{P\in\cP} \frac{\sum_{i=1}^I \E_P\Big[\big|\E_P[T(x)\given Z_i]-\E_P[T(x)]\big|^{2+\delta}\Big]}{\big(\sum_{i=1}^I\Var_P\big[\E_P[T(x)\given Z_i]\big]\big)^{1+\frac{\delta}{2}}} = 0.
\end{equation}
Decompose
\begin{equation*}
    \E_P[T(x)\given Z_i] - \E_P[T(x)]
    =
    \underbrace{\E_P[T(x) \given Z_i] - \E_P[T(x) \given X_i]}_{=:A^{\RN{1}}_i(x)}
    +
    \underbrace{\E_P[T(x) \given X_i] - \E_P[T(x)]}_{=:A^{\RN{2}}_i(x)}.
\end{equation*}
Then
\begin{multline}\label{eq:AiI}
    A_i^{\RN{1}}(x) 
    = \E_P[T(x)\given Z_i] - \E_P[T(x)\given X_i]
    \\
    = \E_P[\omega_i(x)\given Z_i]^\top (Y_i-\mu(X_i)) = \E_P[\omega_i(x)\given X_i]^\top (Y_i-\mu(X_i)),
\end{multline}
using the honesty property~\ref{honesty}. Also $\E_P\big[\big\{\E_P[T(x)\given Z_i]-\E_P[T(x)\given X_i]\big\}\big\{\E_P[T(x)\given X_i]\big\}\big] = 0$, and so
\begin{align*}
    \Var_P\big[\E_P[T(x)\given Z_i]\big]
    &=
    \Var_P\big[\E_P[T(x)\given Z_i]-\E_P[T(x)\given X_i]\big]
    +
    \Var_P\big[\E_P[T(x)\given X_i]\big]
    \\
    &\geq 
    \Var_P\big[\E_P[T(x)\given Z_i]-\E_P[T(x)\given X_i]\big]
    \\
    &=
    \Var_P\big[\E_P[\omega_i(x)\given X_i]^\top (Y_i-\mu(X_i))\big]
    \\
    &=
    \E_P\Big[\Var_P\big[\E_P[\omega_i(x)\given X_i]^\top (Y_i-\mu(X_i))\given X_i\big]\Big]
    \\
    &=
    \E_P\Big[\E_P[\omega_i(x)\given X_i]^\top \sigma_{P,I}(x_i)\E_P[\omega_i(x)\given X_i]\Big]
    \\
    &\geq
    \lambda^{-1} \, \E_P\Big[\big\|\E_P[\omega_i(x)\given X_i]\big\|_2^2\Big].
\end{align*}
Also
\begin{align*}
    \E_P\big[|A_i^{\RN{1}}(x)|^{2+\delta}\big]
    &=
    \E_P\Big[\big|\E_P[\omega_i(x)\given X_i]^\top (Y_i-\mu(X_i))\big|^{2+\delta}\Big]
    \\
    &\leq 
    \tau\,\E_P\Big[\big\{\E_P[\omega_i(x)\given X_i](Y_i-\mu(X_i))\big\}^2\Big]^{1+\frac{\delta}{2}}
    \\
    &\leq
    \tau\,\E_P\Big[\E_P[\omega_i(x)\given X_i]^\top\Cov(Y_i\given X_i)\E_P[\omega_i(x)\given X_i]\Big]^{1+\frac{\delta}{2}}
    \\
    &\leq
    \tau\lambda C_1^\delta \, \E_P\Big[\big\|\E_P[\omega_i(x)\given X_i]\big\|_2^2\Big].
\end{align*}
Also, defining $u := \sup_{x\in\cX}|\mu(x)| < \infty$ and $\mu := \E_P[\mu(X_{ij})]$,
\begin{align*}
    &\quad\E_P\big[|A^{\RN{2}}_i(x)|^{2+\delta}\big]
    \\
    &\leq
    (2uC_1)^{\delta}\E_P\big[|A^{\RN{2}}_i(x)|^2\big],
    \\
    &= 
    (2uC_1)^{\delta}\Var_P\big(\E_P[T(x)\given X_i]\big)
    \\
    &=
    (2uC_1)^{\delta}\Var_P\bigg(
        \sum_{j=1}^{n_i}\E_P[\omega_{ij}(x)\given X_i](\mu(X_{ij})-\mu)
        +
        \sum_{i'\neq i}\sum_{j=1}^{n_{i'}}\E_P[\omega_{i'j}(x)\given X_i](\mu(X_{i'j})-\mu)
    \bigg)
    \\
    &\leq
    2^{1+\delta}u^{\delta}C_1^{\delta} \bigg(
        \Var_P\Big(\sum_{j=1}^{n_i}\E_P[\omega_{ij}(x)\given X_i]\mu(X_{ij})\Big)
        +
        \Var_P\Big(\sum_{i'\neq i}\sum_{j=1}^{n_{i'}}\E_P[\omega_{i'j}(x)\given X_i]\mu(X_{i'j})\Big)
    \bigg)
    \\
    &\leq
    2^{1+\delta}u^{2+\delta}C_1^{\delta} \bigg(
        \Var_P\Big(\sum_{j=1}^{n_i}\E_P[\omega_{ij}(x)\given X_i]\Big)
        +
        \Var_P\Big(\sum_{i'\neq i}\sum_{j=1}^{n_{i'}}\E_P[\omega_{i'j}(x)\given X_i]\Big)
    \bigg)
    \\
    &\leq
    2^{2+\delta}u^{2+\delta}C_1^{\delta} \Var_P\Big(\sum_{j=1}^{n_i}\E_P[\omega_{ij}(x)\given X_i]\Big)
    \\
    &\leq
    2^{2+\delta}u^{2+\delta}C_1^{\delta} \E_P\Big[\big\|\E_P[\omega_i(x)\given X_i]\big\|_2^2\Big],
\end{align*}
where in the first inequality we make use of the crude bound
\begin{multline*}
    |A_i^{\RN{2}}(x)| 
    =
    \big|\E_P\big[\E_P[\omega_i(x)\given X_i]^\top \big(\mu(X_i)-\mu\big)\big]\big|
    \leq
    \sum_{j=1}^{n_i}\E_P\big[\E_P[|\omega_{ij}(x)|\given X_i]\cdot|\mu(X_i)-\mu|\big]
    \\
    \leq
    2u \E_P\bigg[\sum_{j=1}^{n_i}|\omega_{ij}(x)|\bigggiven X_i\bigg] \leq 2uC_1.
\end{multline*}
Combining the above,
\begin{align*}
    \E_P\Big[\big|\E_P[T(x)\given Z_i]-\E_P[T(x)]\big|^{2+\delta}\Big]
    &\leq
    2^{1+\delta} \E_P\Big[\big|\E_P[T(x)\given Z_i]-\E_P[T(x)\given X_i]\big|^{2+\delta}\Big] 
    \\
    &\qquad
    + 2^{1+\delta} \E_P\Big[\big|\E_P[T(x)\given X_i]\big|^{2+\delta}\Big]
    \\
    &\leq
    \big(2^{1+\delta}\tau\lambda C_1^{\delta}+2^{3+2\delta}u^{2+\delta}C_1^{\delta}\big)
    \sum_{j=1}^{n_i} \E_P\Big[\big(\E_P[\omega_{ij}(x)\given X_i]\big)^2\Big].
\end{align*}
Together, it follows that
\begin{equation*}
    \frac{\sum_{i=1}^I \E_P\Big[\big|\E_P[T(x)\given Z_i]-\E_P[T(x)]\big|^{2+\delta}\Big]}{\big(\sum_{i=1}^I\Var_P\big[\E_P[T(x)\given Z_i]\big]\big)^{1+\frac{\delta}{2}}} 
    \lesssim
    \bigg(\sum_{i=1}^I \E_P\Big[\big\|\E_P[\omega_i(x)\given X_i]\big\|_2^2\Big]\bigg)^{-\delta/2}.
\end{equation*}
Applying Lemma~\ref{lem:omega}, 
\begin{multline*}
    \sup_{P\in\cP}\frac{\sum_{i=1}^I \E_P\Big[\big|\E_P[T(x)\given Z_i]-\E_P[T(x)]\big|^{2+\delta}\Big]}{\big(\sum_{i=1}^I\Var_P\big[\E_P[T(x)\given Z_i]\big]\big)^{1+\frac{\delta}{2}}} 
    \lesssim
    \sup_{P\in\cP}\bigg(\sum_{i=1}^I \E_P\Big[\big\|\E_P[\omega_i(x)\given X_i]\big\|_2^2\Big]\bigg)^{-\delta/2}
    \\
    \lesssim \bigg(I^{1-\beta}k\log(n_{\mathrm{c}}\vee I)\bigg)^{-\delta/2}
    =o(1),
\end{multline*}
for any $\alpha<1$, and thus the Lyapunov condition~\eqref{eq:lyapunov} holds. Having verified all the conditions, applying a uniform version of the Lyapunov central limit theorem (see e.g.~\citet[Lemma 17]{young}) completes the proof of Part 1.

\medskip
\noindent {\bf Part 2:} Showing $B_P=o_\cP(1)$.

A single decision tree $T(x)$ predictor admits the Efron--Stein ANOVA decomposition~\citep{ANOVA}
\begin{equation*}
    T(x) = \E_P[T(x)] + \sum_{k=1}^{s }
    \sum_{i_1<\cdots<i_k}T^{(k)}(Z_{i_1},\ldots,Z_{i_k}),
\end{equation*}
for some symmetric functions $(T^{(k)})_{k=1,\ldots,s}$ such that $\big(T^{(k)}(Z_{i_1},\ldots,Z_{i_k})\big)_{k\in\{1,2,\ldots,s \}}$ are all mean-zero, uncorrelated functions. Note that $T^{(1)}$ takes the form as in~\eqref{eq:T^{(1)}}. Similarly the Efron--Stein ANOVA decomposition for the random forest predictor $\hat{\mu}(x)$ is
\begin{equation*}
    \hat{\mu}(x) = \E_P[T(x)] + \binom{I}{s }^{-1} \sum_{k=1}^{s }
        \binom{I-k}{s -k} \sum_{i_1<\cdots<i_k}
        T^{(k)}(Z_{i_1},\ldots,Z_{i_k}).
\end{equation*}
Note in particular such a decomposition holds even if the tree $T$ is not symmetric and $(Z_1,\ldots,Z_I)$ are independent but not identically distributed (see~\citet[Comment 4]{ANOVA}). Then
\begin{align*}
    \E_P\Big[\big(\hat{\mu}(x)-\accentset{\circ}{\hat{\mu}}(x)\big)^2\Big]
    &=
    \E_P\Bigg[\binom{I}{s }^{-2}\bigg(\sum_{k=2}^{s }\binom{I-k}{s -k}\sum_{i_1<\cdots<i_k}T^{(k)}(Z_{i_1},\ldots,Z_{i_k})\bigg)^2\Bigg]
    \\
    &=
    \binom{I}{s }^{-2}\sum_{k=2}^{s }\binom{I-k}{s -k}^2\Var_P\bigg[\sum_{i_1<\cdots<i_k}T^{(k)}(Z_{i_1},\ldots,Z_{i_k})\bigg]
    \\
    &\leq \bigg(\frac{s }{I}\bigg)^2 \sum_{k=2}^{s } \Var_P\bigg[\sum_{i_1<\cdots<i_k}T^{(k)}(Z_{i_1},\ldots,Z_{i_k})\bigg]
    \\
    &\leq \bigg(\frac{s }{I}\bigg)^2 \sum_{k=1}^{s }  \Var_P\bigg[\sum_{i_1<\cdots<i_k}T^{(k)}(Z_{i_1},\ldots,Z_{i_k})\bigg]
    \\
    &\leq \bigg(\frac{s }{I}\bigg)^2 \Var_P\bigg[\sum_{k=1}^{s } \sum_{i_1<\cdots<i_k}T^{(k)}(Z_{i_1},\ldots,Z_{i_k})\bigg]
    \\
    &= \frac{s ^2}{I^2} \Var_P\big[T(x)\big].
\end{align*}
Therefore
\begin{equation*}
    \E_P\big[B_P^2\big]
    =
    \frac{\E_P\Big[\big(\hat{\mu}(x)-\accentset{\circ}{\hat{\mu}}(x)\big)^2\Big]}{\sigma_{P,I}^2(x;\rho^*)}
    \lesssim
    {I^{-1}}{s (\log s_N)^d}\cdot\frac{\Var_P [T(x)]}{\Var_P [T_{\rho^*}(x)]}
    \sim
    I^{\beta-1}(\log s_N)^d
\end{equation*}
where we additionally use Lemmas~\ref{lem:sig2}~and~\ref{lem:VarT(x)}, and $\beta<1$. By Markov's inequality it follows that for $t>0$,
\begin{equation*}
    \sup_{P\in\cP}\PP_P\big(|B_P|>t\big)\leq \sup_{P\in\cP}t^{-2}\E_P\big[B_P^2\big] \lesssim I^{\beta-1}(\log s_N)^d = o(1),
\end{equation*}
and so $B_P=o_\cP(1)$.

\medskip
\noindent {\bf Part 3:} Showing $\sup_{P\in\cP}C_P=o(1)$.

Recall by Lemma~\ref{lem:unbiasedness} we have that for any $\epsilon>0$ and with probability at least $1-I^{-\beta}$ for sufficiently large $I$ we have  $\sup_{P\in\cP}\sup_{x\in\cX}\big|\E_P[\hat{\mu}(x)]-\mu(x)\big| = O(I^{-\gamma(\beta)})$ for $\gamma(\beta):=(1-\epsilon)\frac{\log((1-\alpha)^{-1})}{\log(\alpha^{-1})}\,\frac{\pi}{d}\,\beta$. Also by Lemma~\ref{lem:sig2}, 
\begin{equation*}
    \inf_{P\in\cP}\sigma_{P,I}^2(x;\rho^*) \gtrsim I^{\beta-1-\epsilon'}
\end{equation*}
for any $\epsilon'>0$, and so
\begin{equation*}
    \sup_{P\in\cP}C_P = O\bigg(I^{\frac{1}{2}\big(1+\epsilon'-\beta\big(1+\frac{2(1-\epsilon)\log((1-\alpha)^{-1})}{\log(\alpha^{-1})}\frac{\pi}{d}\big)\big)}\bigg) = o(1),
\end{equation*}
provided 
\begin{equation*}
    \beta > \bigg(1+\frac{2\pi\log((1-\alpha)^{-1})}{d\log(\alpha^{-1})}\bigg)^{-1}
    =
    1-\bigg(1+\frac{d}{2\pi}\frac{\log(\alpha^{-1})}{\log((1-\alpha)^{-1})}\bigg)^{-1} =: \beta_{\min}
    ,
\end{equation*}
and by taking $\epsilon,\epsilon'>0$ sufficiently small.

\medskip
\noindent {\bf Part 4:} Showing $D = o_\cP(1)$.

Recall
\begin{small}
\begin{align*}
    D &:= 
    \frac{\accentset{\circ}{\hat{\mu}}(x) - \E_P[\accentset{\circ}{\hat{\mu}}(x)] - \accentset{\circ}{\hat{\mu}}_{\rho^*}(x) + \E_P[\accentset{\circ}{\hat{\mu}}_{\rho^*}(x)]}
    {\sigma_{P,I}(x;\rho^*)}
    \\
    &=
    \frac{1}{\sigma_{P,I}(x;\rho^*)} \sum_{i=1}^I \binom{I}{s }^{-1} \sum_{i_1<\cdots<i_s}\big\{\E_P[T(x)\given Z_i]-\E_P[T(x)]-\E_P[T_{\rho^*}(x)\given Z_i]+\E_P[T_{\rho^*}(x)]\big\}
    \\
    &=
    \underbrace{\frac{1}{\sigma_{P,I}(x;\rho^*)}\sum_{i=1}^I \binom{I}{s }^{-1} \sum_{i_1<\cdots<i_s} \big\{\E_P[T(x)\given Z_i]-\E_P[T(x)\given X_i]-\E_P[T_{\rho^*}(x)\given Z_i]+\E_P[T_{\rho^*}(x)\given X_i]\big\}}_{=:D_{\RN{1}}}
    \\
    &\qquad+
    \underbrace{\frac{1}{\sigma_{P,I}(x;\rho^*)}\sum_{i=1}^I \binom{I}{s }^{-1} \sum_{i_1<\cdots<i_s} \big\{\E_P[T(x)\given X_i]-\E_P[T(x)]-\E_P[T_{\rho^*}(x)\given X_i]+\E_P[T_{\rho^*}(x)]\big\}}_{=:D_{\RN{2}}}
    .
\end{align*}
\end{small}
Recall, as discussed at the beginning of the proof, notationally we consider without loss of generality the setting of $\hat{\rho}$ and splits estimated using an auxiliary sample. Notationally therefore the above quantities $D_{\RN{1}}$ and $D_{\RN{2}}$ as well as their expectations $\E_P[D_{\RN{1}}]$ and $\E_P[D_{\RN{2}}]$ should be read as functions also of $\hat{\rho}$ and the random forest splits (thus all the aforementioned quantities are random given the auxiliary data). 
Also define
\begin{gather*}
    \bar{\mu}(x) := \binom{I}{s }^{-1}\sum_{i_1<\cdots<i_s}\bar{T}(x;X_{i_1},\ldots,X_{i_s}),
    \\
    \bar{T}(x;X_{i_1},\ldots,X_{i_s}) := \E_P\big[T(x;Z_{i_1},\ldots,Z_{i_s})\given X_{\text{all}}\big]
\end{gather*}
\begin{equation*}
    \accentset{\circ}{\bar{\mu}}(x) = \E_P[T(x)] + \frac{s }{I}\sum_{i=1}^I \big\{\E_P[T(x)\given X_i]-\E_P[T(x)]\big\}.
\end{equation*}
Also define $\bar{\mu}_{\rho^*}(x)$ and its H\'{a}jek projection $\accentset{\circ}{\bar{\mu}}_{\rho^*}(x)$ analogously. Then
\begin{align*}
    D_{\RN{2}} 
    &=
    \frac{\accentset{\circ}{\bar{\mu}}(x)-\accentset{\circ}{\bar{\mu}}_{\rho^*}(x)}{\sigma_{P,I}(x;\rho^*)}
    \\
    &=
    \frac{\bar{\mu}(x)-\bar{\mu}_{\rho^*}(x)}{\sigma_{P,I}(x;\rho^*)}
    +
    \frac{\accentset{\circ}{\bar{\mu}}(x) - \bar{\mu}(x)}{\sigma_{P,I}(x;\rho^*)}
    -
    \frac{\accentset{\circ}{\bar{\mu}}_{\rho^*}(x) - \bar{\mu}_{\rho^*}(x)}{\sigma_{P,I}(x;\rho^*)},
    \\
    &=
    \frac{\bar{\mu}(x)-\bar{\mu}_{\rho^*}(x)}{\sigma_{P,I}(x;\rho^*)}
    + o_\cP(1),
\end{align*}
where for the final step, similar to the property for H\'{a}jek projections used in Part 2, it can be shown that
\begin{equation*}
    \frac{\E_P\Big[\big(\accentset{\circ}{\bar{\mu}}(x)-\bar{\mu}(x)\big)^2\Big]}{\sigma_{P,I}^2(x;\rho^*)}
    \leq
    \frac{s ^2}{I^2\sigma_{P,I}^2(x;\rho^*)}\Var_P\big(\E_P[T(x)\given X_i]\big)
    \leq
    \frac{(2uC_1)^2 s ^2}{I^2\sigma_{P,I}^2(x;\rho^*)} \lesssim \frac{s (\log s)^d}{I},
\end{equation*}
again using Lemma~\ref{lem:sig2}, and so
\begin{equation*}
    \supP\frac{\E_P\Big[\big(\accentset{\circ}{\bar{\mu}}(x)-\bar{\mu}(x)\big)^2\Big]}{\sigma_{P,I}^2(x;\rho^*)} = o(1).
\end{equation*}
Also for each tree, and for arbitrarily small $\epsilon>0$, by Lemma~\ref{lem:diam} for sufficiently large $I$,
\begin{align*}
    \big|\bar{T}(x)-\bar{T}_{\rho^*}(x)\big|
    &= 
    \big|
    \big( \bar{T}(x) - \mu(x) \big) 
    - 
    \big( \bar{T}_{\rho^*}(x) - \mu(x) \big)
    \big|
    \\
    &\leq
    2L_{\mu}d^{1/2}(1+C_1)\Big(\frac{s }{2k-1}\Big)^{-(1-\epsilon)\frac{\log((1-\alpha)^{-1})}{\log(\alpha^{-1})}\cdot\frac{\pi}{d}}
\end{align*}
and so, as $\bar{\mu}$ is an average of trees $\bar{T}$,
\begin{equation*}
    |\bar{\mu}(x)-\bar{\mu}_{\rho^*}(x)| = O_{\cP}\Big(I^{-(1-\epsilon)\frac{\log((1-\alpha)^{-1})}{\log(\alpha^{-1})}\cdot\frac{\pi}{d}\beta}\Big).
\end{equation*}
Also recalling (again by Lemma~\ref{lem:sig2}) $\sigma_{P,I}^2(x;\rho^*)\gtrsim I^{\beta-1-\epsilon'}$ for any $\epsilon'>0$,
\begin{equation*}
    \frac{|\bar{\mu}(x)-\bar{\mu}_{\rho^*}(x)|}{\sigma_{P,I}(x;\rho^*)} = O_{\cP}\bigg(I^{\frac{1}{2}\big({1-\beta+\epsilon'-2(1-\epsilon)\frac{\log((1-\alpha)^{-1})}{\log(\alpha^{-1})}\cdot\frac{\pi}{d}\beta}\big)}\bigg) = o_{\cP}(1),
\end{equation*}
for any $\beta>\beta_{\min}$ (i.e.~obtained by also taking $\epsilon,\epsilon'>0$ sufficiently small). Therefore it follows by Markov's inequality that $D_{\RN{2}} = o_\cP(1)$.

Similarly, for the term $D_{\RN{1}}$, note first that $\E_P[D_{\RN{1}}]=0$, and so we obtain 
\begin{align*}
    \E_P\big[D_{\RN{1}}^2\big]
    &= 
    \frac{s ^2}{I^2\sigma_{P,I}^2(x;\rho^*)}\sum_{i=1}^I \E_P\Big[\big\{\E_P[T_{\hat{\rho}}(x)\given Z_i]-\E_P[T_{\hat{\rho}}(x)\given X_i]
    \\
    &\hspace{18em}-\E_P[T_{\rho^*}(x)\given Z_i]+\E_P[T_{\rho^*}(x)\given X_i]\big\}^2\Big]
    \\
    &=
    \frac{s }{I\sigma_{P,I}^2(x;\rho^*)}\sum_{i=1}^s \E_P\Big[\big\{\partial_\rho\big(\E_P[T_{\rho}(x)\given Z_i]-\E_P[T_{\rho}(x)\given X_i]\big)\big|_{\rho=\bar{\rho}(\hat{\rho},\rho^*)}\big\}^2(\hat{\rho}-\rho^*)^2\Big],
\end{align*}
for some $\bar{\rho}(\hat{\rho},\rho^*)\in\Gamma$. 
Via algebraic calculations it can be verified that
\begin{equation*}
    \E_P[T_{\rho}(x)\given Z_i] - \E_P[T_{\rho}(x)\given X_i]
    =
    \mathrm{e}_{J(x)}^\top \big(\chi^\top W(\rho)\chi\big)^{-1}\big(\chi_i^\top W_i(\rho)(Y_i-\mu(X_i))\big),
\end{equation*}
and
\begin{multline*}
    \partial_{\rho}\Big(\E_P[T_{\rho}(x)\given Z_i] - \E_P[T_{\rho}(x)\given X_i]\Big)
    \\
    =
    - \mathrm{e}_{J(x)}^\top \big(
    \chi^\top W\chi\big)^{-1}\Big(\chi^\top \frac{\partial W}{\partial\rho}\chi\Big)\big(
    \chi^\top W\chi\big)^{-1}\big(\chi_i^\top W_i\Cov_P(\varepsilon_i\given X_i)\big)
    \\
    +
    \mathrm{e}_{J(x)}^\top \big(
    \chi^\top W\chi\big)^{-1}\Big(\chi_i^\top \frac{\partial W_i}{\partial\rho}\Cov_P(\varepsilon_i\given X_i)\Big)
    ,
\end{multline*}
where $\varepsilon_i := Y_i - \mu(X_i)$. Then
\begin{align*}
    &\quad
    \Big\{\partial_{\rho}\Big(\E_P[T_{\rho}(x)\given Z_i] - \E_P[T_{\rho}(x)\given X_i]\Big)\Big\}^2
    \\
    &\leq
    2\Big\{\mathrm{e}_{J(x)}^\top \big(
    \chi^\top W\chi\big)^{-1}\Big(\chi^\top \frac{\partial W}{\partial\rho}\chi\Big)\big(
    \chi^\top W\chi\big)^{-1}\big(\chi_i^\top W_i\varepsilon_i\big)\Big\}^2
    +
    2\Big\{\mathrm{e}_{J(x)}^\top \big(
    \chi^\top W\chi\big)^{-1}\big(\chi_i^\top \frac{\partial W_i}{\partial\rho}\varepsilon_i\big)\Big\}^2
    ,
\end{align*}
following because $(a+b)^2\leq 2(a^2+b^2)$. Also define  $\Sigma_i:=\Cov_P(\varepsilon_i\given X_i)$ and $\Sigma:=\diag(\Sigma_i\ind_{(i\in\cI_{\text{eval}})})_{i\in[I]}$ the corresponding block diagonal matrix. Then
\begin{align*}
    &\quad
    \sum_{i=1}^I\E_P\bigg[\Big\{\partial_{\rho}\Big(\E_P[T_{\rho}(x)\given Z_i] - \E_P[T_{\rho}(x)\given X_i]\Big)\Big\}^2\bigg]
    \\
    &\leq
    2 \E_P\bigg[\mathrm{e}_{J(x)}^\top (\chi^\top W\chi)^{-1}\Big(\chi^\top \frac{\partial W}{\partial\rho}\chi\Big)(\chi^\top W\chi)^{-1}(\chi^\top W\Sigma W\chi)(\chi^\top W\chi)^{-1}\Big(\chi^\top \frac{\partial W}{\partial\rho}\chi\Big)(\chi^\top W\chi)^{-1}\mathrm{e}_{J(x)} \bigg]
    \\
    &\quad
    +
    2 \E_P\bigg[\mathrm{e}_{J(x)}^\top (\chi^\top W\chi)^{-1}(\chi^\top \frac{\partial W}{\partial\rho}\Sigma\frac{\partial W}{\partial\rho}\chi)(\chi^\top W\chi)^{-1}\mathrm{e}_{J(x)}\bigg]
    \\
    &\leq
    \frac{2\Lambda_{\max}(W\Sigma W)\Lambda_{\max}\big(\big\{\frac{\partial W}{\partial\rho}\big\}^2\big)\big(\Lambda_{\max}(\chi^\top \chi)\big)^3}{\big(\Lambda_{\min}(W)\big)^4\big(\Lambda_{\min}(\chi^\top \chi)\big)^4}
    +
    \frac{2\Lambda_{\max}(\frac{\partial W}{\partial\rho}\Sigma\frac{\partial W}{\partial\rho})\Lambda_{\max}(\chi^\top \chi)}{\big(\Lambda_{\min}(W)\big)^2\big(\Lambda_{\min}(\chi^\top \chi)\big)^2}
    \\
    &\leq
    \frac{2(2k-1)^3}{k^4}\cdot\frac{\Lambda_{\max}(W\Sigma W)\big(\Lambda_{\max}\big(\big\{\frac{\partial W}{\partial\rho}\big\}^2\big)}{\big(\Lambda_{\min}(W)\big)^4}
    + \frac{2(2k-1)}{k^2}\cdot\frac{\Lambda_{\max}\big(\big\{\frac{\partial W}{\partial\rho}\big\}^2\big)\Lambda_{\max}\big(\Sigma\big)}{\big(\Lambda_{\min}(W)\big)^2}
    \\
    &\lesssim 1,
\end{align*}
almost surely and for any $\rho\in\Gamma$, and where we use that $\Lambda_{\max}(W)\leq\|W\|_1\lesssim 1$ and $\Lambda_{\max}\big(\big\{\frac{\partial W}{\partial\rho}\big\}^2\big)\lesssim 1$. Subsequently for some $r>0$,
\begin{equation*}
    \E_P\big[D_{\RN{1}}^2\big] \lesssim \frac{s }{I\sigma_{P,I}^2(x)}(\hat{\rho}-\rho^*)^2
    \lesssim
    (\log s)^d (\hat{\rho}-\rho^*)^2 = O_{\cP}\big((\log I)^dI^{-r}\big)=o_{\cP}(1),
\end{equation*}
applying Lemma~\ref{lem:sig2} (and recalling notationally the quantity $\E_P[D_{\RN{1}}^2]$ is a random quantity in terms of the auxiliary data used to construct $\hat{\rho}$). Thus $D_{\RN{1}}=o_{\cP}(1)$ and consequently $D=o_{\cP}(1)$.

Parts 1, 2, 3 and 4, together with a uniform version of Slutsky's Lemma, completes the proof.

\end{proof}

\begin{lemma}\label{lem:sig2}
    $$\sigma_{P,I}^2(x;\rho^*) \gtrsim \frac{s_I}{I(\log(n_{\mathrm{c}}\vee I))^d}.$$
\end{lemma}
\begin{proof}
    First, note that as
    \begin{equation*}
    \sigma_{P,I}^2(x;\rho^*) = \Var_P\big[\accentset{\circ}{\hat{\mu}}_{\rho^*}(x)\big]
    =
    \frac{s }{I} \Var_P\big[\accentset{\circ}{T}_{\rho^*}(x)\big],
\end{equation*}
it suffices to show that
\begin{equation*}
    \Var_P\big[\accentset{\circ}{T}_{\rho^*}(x)\big]
    \gtrsim
    (\log I\vee n_{\mathrm{c}})^{-d}
    .
\end{equation*}
To prove this claim, first note that by Lemma~\ref{lem:omega},
\begin{equation*}
    \Var_P\big[\E_P[T_{\rho^*}(x)\given Z_i]\big]
    \gtrsim
    \E_P\Big[\big\|\E_P[\omega_i(x)\given X_i]\big\|_2^2\Big]
    \gtrsim 
    s_I^{-1}k(\log(n_{\mathrm{c}}\vee I))^{-d}.
\end{equation*}
Then it follows that
\begin{equation*}
    \Var_P\big[T_{\rho^*}(x;Z_{i_1},\ldots,Z_{i_s})\big] 
    \gtrsim 
    \sum_{k=1}^{s_I} \Var_P\big[\E_P[T_{\rho^*}(x)\given Z_{i_k}]\big] \gtrsim k(\log(n_{\mathrm{c}}\vee I))^{-d},
\end{equation*}
completing the proof.
\end{proof}

\begin{lemma}\label{lem:VarT(x)}
    Each clustered decision tree $T(x)$ satisfies
    \begin{equation*}
        \Var_P \, T(x) \sim \frac{1}{k}.
    \end{equation*}
\end{lemma}
\begin{proof}
    We begin by showing that $\Var_P\,T(x)\gtrsim 1$. 
    Recall the notation introduced in the proof of Lemma~\ref{lem:unbiasedness}, and define the block diagonal matrix $\Sigma := \Var_P(\cY\given X)$.  
    Then
    \begin{align*}
        \Var_P\,T(x) &\geq \E_P\big[\Var_P(T(x) \given X)\big]
        \\
        &=\E_P\big[\Var_P\big(\mathrm{e}_{J(x)}^\top (\chi^\top \hat{W}\chi)^{-1}(\chi^\top \hat{W}\cY)\big)\big]
        \\
        &= \E_P\big[ \mathrm{e}_{J(x)}^\top (\chi^\top \hat{W}\chi)^{-1} (\chi^\top \hat{W}\Sigma \hat{W}\chi) (\chi^\top \hat{W}\chi)^{-1} \mathrm{e}_{J(x)} \big]
        \\
        &\geq \E_P\bigg[\frac{\Lambda_{\min}(\hat{W}^{\frac{1}{2}}\Sigma \hat{W}^{\frac{1}{2}})}{\Lambda_{\max}(\chi^\top \hat{W}\chi)}\|\mathrm{e}_{J(x)}\|_2^2\bigg]
        \\
        &\geq \E_P\bigg[\frac{\Lambda_{\min}(\hat{W})\Lambda_{\min}(\Sigma)}{\Lambda_{\max}(\hat{W})\sigma_{\max}^2(\chi^\top )}\bigg]
        \\
        &\geq \E_P\bigg[\frac{\Lambda_{\min}(\hat{W})\Lambda_{\min}(\Sigma)}{\Lambda_{\max}(\hat{W})(2k-1)}\bigg]
        \\
        &\geq \frac{c_W}{C_W\lambda k} \sim \frac{1}{k}.
    \end{align*}

    We now show that $\Var_P\,T(x)\lesssim 1/k$.

    \begin{align*}
    \Var_P[T(x)]
    &=
    \Var_P\big[\mathrm{e}_{J(x)}^\top (\chi^\top W\chi)^{-1}(\chi^\top W\mathcal{Y})\big]
    \\
    &=
    \Var_P\big[\E_P[\mathrm{e}_{J(x)}^\top (\chi^\top W\chi)^{-1}(\chi^\top W\mathcal{Y})\given X_{\text{eval}},Z_{\text{split}}]\big]
    \\
    &\qquad+
    \E_P\big[\Var_P[\mathrm{e}_{J(x)}^\top (\chi^\top W\chi)^{-1}(\chi^\top W\mathcal{Y})\given X_{\text{eval}},Z_{\text{split}}]\big].
\end{align*}
Then
\begin{align*}
    &\quad\E_P\big[\Var_P[\mathrm{e}_{J(x)}^\top (\chi^\top \hat{W}\chi)^{-1}(\chi^\top \hat{W}\mathcal{Y})\given X_{\text{eval}},Z_{\text{split}}]\big]
    \\
    &=
    \E_P\big[\mathrm{e}_{J(x)}^\top (\chi^\top \hat{W}\chi)^{-1}(\chi^\top \hat{W}\Sigma \hat{W}\chi)(\chi^\top \hat{W}\chi)^{-1}\mathrm{e}_{J(x)}\big]
    \\
    &\leq
    \frac{\Lambda_{\max}(\hat{W}^{\frac{1}{2}}\Sigma \hat{W}^{\frac{1}{2}})}{k \Lambda_{\min}(\hat{W})} 
    \leq
    \frac{C_W\lambda}{c_W k}
    \sim \frac{1}{k},
\end{align*}
and
\begin{align*}
    &\quad\Var_P\big[\E_P[\mathrm{e}_{J(x)}^\top (\chi^\top \hat{W}\chi)^{-1}(\chi^\top \hat{W}\mathcal{Y})\given X_{\text{eval}},Z_{\text{split}}]\big]
    \\
    &=\Var_P\big[\E_P\big[\mathrm{e}_{J(x)}^\top (\chi^\top \hat{W}\chi)^{-1}(\chi^\top \hat{W}\mu(X_{\mathrm{eval}}))\given Z_{\mathrm{split}}]\big]
    \\
    &=\Var_P\big[\E_P\big[\mathrm{e}_{J(x)}^\top (\chi^\top \hat{W}\chi)^{-1}(\chi^\top \hat{W}\{\mu(X_{\mathrm{eval}})-\E_P[Y^*|X^*\in L(x),Z_{\mathrm{split}}]\})\given Z_{\mathrm{split}}]\big]
    \\
    &\leq
    \E_P\big[\big(\E_P\big[\mathrm{e}_{J(x)}^\top (\chi^\top \hat{W}\chi)^{-1}(\chi^\top \hat{W}\{\mu(X_{\mathrm{eval}})-\E_P[Y^*|X^*\in L(x),Z_{\mathrm{split}}]\})\given Z_{\mathrm{split}}]\big)^2\big]
    \\
    &\leq L_\mu d^{1/2} C_1 \Big(\frac{n_{\mathrm{c}}I}{2k-1}\Big)^{-(1-\epsilon)\frac{\log((1-\alpha)^{-1})}{\log(\alpha^{-1})}\cdot\frac{\pi}{d}}  \lesssim \frac{1}{k},
\end{align*}
for sufficiently large $I$, with the final inequality following as in Lemma~\ref{lem:unbiasedness}. Thus
\begin{equation*}
    \Var_P\,T(x) \sim \frac{1}{k}.
\end{equation*}
\end{proof}

\begin{lemma}[\citet{wager}, Lemma 4]\label{lem:Pij}
    Suppose we observe $s\in\mathbb{N}$ i.i.d.~datapoints indexed as $X_{ij}\in\cX:=[0,1]^d$, 
    with density $f$ satisfying $\nu^{-1}\leq f(x)\leq \nu$ for all $x\in\cX$ and for some finite $\nu\geq1$. Define $P_{ij}(x) := \ind(\{X_{ij} \text{ is a $k$-PNN of $x$}\})$. Then there exists a constant $C_{\nu,d}$ depending only on $\nu$ and $d$ such that
    \begin{equation*}
        \PP_P\big(\E_P[P_{ij}(x)\given Y_{ij},X_{ij}]\geq s^{-2}\big) \lesssim C_{\nu,d} \cdot \frac{k (\log s)^d}{s}.
    \end{equation*}
    Further, if $X_{ij} \iid \mathrm{Unif}\big([0,1]^d\big)$ then the above holds with $C_{\nu,d} = \frac{2^{d+1}}{(d-1)!}$.
\end{lemma}

\begin{lemma}\label{lem:lowerbound-on-EA2}
    Suppose $A$ is a non-negative random variable. Then for any $c>0$
    \begin{equation*}
        \E\big[A^2\big] \geq \frac{\frac{1}{2}\E\big[A\big]^2 - \E\big[A\ind_{(A<c)}\big]^2}{\PP(A\geq c)}.
    \end{equation*}
\end{lemma}

\begin{proof}
Note that
    \begin{multline*}
        \E\big[A^2\big]
        \geq
        \E\big[A^2\ind_{(A\geq c)}\big]
        =
        \PP(A\geq c)\E\big[A^2\given A\geq c\big]
        \\
        \geq 
        \PP(A\geq c)\E[A\given A>c]^2
        = 
        \PP(A\geq c)\bigg(\frac{\E[A\ind_{(A\geq c)}]}{\PP(A\geq c)}\bigg)^2
        =
        \frac{\E[A\ind_{(A\geq c)}]^2}{\PP(A\geq c)}
        \\
        =
        \frac{\big(\E[A] - \E[A\ind_{(A< c)}]\big)^2}{\PP(A\geq c)}
        \geq
        \frac{\frac{1}{2}\E[A]^2 - \E[A\ind_{(A<c)}]^2}{\PP(A\geq c)}
        ,
    \end{multline*}
    where the second inequality follows by Jensen's inequality, and the final inequality follows as $(a-b)^2\geq \frac{1}{2}a^2 - b^2$ for any $a,b\in\R$.
\end{proof}

\begin{lemma}\label{lem:omega}
    There exists some universal constant $C>0$ such that
    \begin{equation*}
        \E_P\big[\|\E_P[\omega_i(x)\given X_i]\|_2^2\big]
        \geq 
        C \frac{k}{s_I(\log(n_{\mathrm{c}}\vee I))^d}, 
    \end{equation*}
    and thus
    \begin{equation*}
        \sum_{i=1}^I\E_P\big[\|\E_P[\omega_i(x)\given X_i]\|_2^2\big]
        \geq 
        C I^{1-\beta}k(\log(n_{\mathrm{c}}\vee I))^{-d}.
    \end{equation*}
\end{lemma}
\begin{proof}
    Define
    \[P_{ij}(x):=\ind(\{X_{ij}\text{ is a $k$ potential nearest neighbour of $x$}\}).\]
    Note therefore that $X_{ij}\in L(x)$ implies that $P_{ij}(x)=1$, and therefore
    \[\PP_P(X_{ij}\in L(x)\given X_{ij}) \leq \PP_P(P_{ij}(x)=1\given X_{ij}),\]
    almost surely. Also
\begin{equation*}
    \ind_{L(x)}(X_{ij})=\mathrm{e}_{J(x)}^\top \mathrm{e}_{J(X_{ij})}=\mathrm{e}_{J(x)}^\top (\chi^\top W\chi)^{-1}(\chi^\top W)\chi \mathrm{e}_{J({X_{ij}})}=\omega(x)^\top \chi \mathrm{e}_{J(X_{ij})}=\omega(x)^\top \chi\chi^\top\mathrm{e}_{ij}.
\end{equation*}
In case (i) of Assumption~\ref{bounded-density}, it follows that
\begin{align*}
    &\quad\PP_P(X_{ij}\in L(x)\given X_{ij})^2
    \\
    &=
    \E_P[\omega(x)^\top \chi\chi^\top e_{ij}\given X_{ij}]^2
    \\
    &=
    \E_P\bigg[\sum_{i'=1}^{I}\sum_{j'=1}^{n_{i'}}\ind_{L(X_{ij})}(X_{i'j'})\omega_{i'j'}(x) \Biggiven X_{ij}\bigg]^2
    \\
    &=
    \bigg(
        \E_P[\omega_{ij}(x)\given X_{ij}] 
        +
        \E_P\bigg[\sum_{(i',j')\neq (i,j)}\ind_{L(X_{ij})}(X_{i'j'})\omega_{i'j'}(x)\Biggiven X_{ij}\bigg]
    \bigg)^2
    \\
    &\leq
        2\,\E_P[\omega_{ij}(x)\given X_{ij}]^2 
        \\
        &\qquad
        +
        2\,\bigg(\sum_{(i',j')\neq (i,j)}\E_P[\omega_{i'j'}(x)\given X_{ij}, X_{i'j'}\in L(X_{ij})]\PP_P(X_{i'j'}\in L(X_{ij})\given X_{ij})\bigg)^2,
\end{align*}
and thus
\begin{multline}\label{eq:EP}
    \E_P\big[\PP_P(X_{ij}\in L(x)\given X_{ij})^2\big]
    \leq
    2 \E_P\big[\E_P[\omega_{ij}(x)\given X_{ij}]^2\big]
    \\
    +
    2 \E_P\bigg[\Big(\sum_{(i',j')\neq(i,j)}\E_P[\omega_{i'j'}(x)\given X_{ij},X_{i'j'}\in L(X_{ij})]\PP_P(X_{i'j'}\in L(X_{ij})\given X_{ij})\Big)^2\bigg].
\end{multline}
By exchangeability and independence of $X_{ij}\independent X_{i'j'}$, for $(i'j')\neq(i,j)$ we have that $$\PP_P(X_{i'j'}\in L(X_{ij})\given X_{ij}) = \E_P\bigg[\frac{1}{s_N-1}\sum_{(i',j')\neq(i,j)}\ind_{L(X_{ij})}(X_{i'j'})\bigggiven X_{ij}\bigg] \leq \frac{2k-1}{s_N-1},$$ almost surely, and thus
\begin{align*}
    &\quad\;
    \Big(\sum_{(i',j')\neq(i,j)}\E_P[\omega_{i'j'}(x)\given X_{ij},X_{i'j'}\in L(X_{ij})]\PP_P(X_{i'j'}\in L(X_{ij})\given X_{ij})\Big)^2
    \\
    &\leq
    \Big(\sum_{(i',j')\neq(i,j)}\E_P[|\omega_{i'j'}(x)|\given X_{ij},X_{i'j'}\in L(X_{ij})]\PP_P(X_{i'j'}\in L(X_{ij})\given X_{ij})\Big)^2
    \\
    &\leq 
    \frac{(2k-1)^2}{(s_N-1)^2}\bigg(\sum_{(i',j')\neq(i,j)}\E_P\big[|\omega_{i'j'}(x)|\given X_{ij},X_{i',j'}\in L(X_{ij})\big]\bigg)^2
    \\
    &=
    \frac{(2k-1)^2}{(s_N-1)^2}\bigg(\E_P\bigg[\underbrace{\sum_{(i',j')\neq(i,j)}|\omega_{i'j'}(x)|}_{\leq C_1 \text{ (Lemma~\ref{lem:sum-abs})}}\bigggiven X_{ij},X_{i',j'}\in L(X_{ij})\bigg]\bigg)^2
    \\
    &\leq
    \frac{(2k-1)^2 C_1^2}{(s_N-1)^2}.
\end{align*}
Along with~\eqref{eq:EP} we obtain
\begin{multline*}
    \E_P\big[\E_P[\omega_{ij}(x)\given X_{ij}]^2\big]
    \geq
    \frac{1}{2}\E_P\big[\PP_P(X_{ij}\in L(x)\given X_{ij})^2\big]
    -
    \frac{(2k-1)^2C_1^2}{(s_N-1)^2}
    \\
    \gtrsim
    \E_P\big[\PP_P(X_{ij}\in L(x)\given X_{ij})^2\big] + s_N^{-2}k^2,
\end{multline*}
and therefore by Lemma~\ref{lem:lowerbound-on-EA2} (with $c=s_N^{-2}$) 
and Lemma~\ref{lem:Pij},
\begin{align*}
    &\hspace{-0.4cm}\sum_{i=1}^I \E_P\big[\|\E_P[\omega_i(x)\given X_i]\|_2^2\big]
    =
    \sum_{i=1}^I\sum_{j=1}^{n_i} \E_P\big[\E_P[\omega_{ij}(x)\given X_i]^2\big]
    \geq
    \sum_{i=1}^I\sum_{j=1}^{n_i} \Var_P\big[\E_P[\omega_{ij}(x)\given X_i]\big]
    \\
    &\geq
    \sum_{i=1}^I\sum_{j=1}^{n_i} \Var_P\big[\E_P[\omega_{ij}(x)\given X_{ij}]\big]
    =
    \sum_{i=1}^I\sum_{j=1}^{n_i}\Big( \E_P\big[\E_P[\omega_{ij}(x)\given X_{ij}]^2\big] - \underbrace{\E_P[\omega_{ij}(x)]^2}_{\sim s_N^{-2}}\Big)
    \\
    &\geq
    \frac{1}{2} \sum_{i=1}^I \sum_{j=1}^{n_i} \Big(\E_P\big[\PP_P(X_{ij}\in L(x)\given X_{ij})^2\big] + \Omega(s_N^{-2})\Big)
    =
    \frac{1}{2} \sum_{i=1}^I \sum_{j=1}^{n_i}\frac{\PP_P(X_{ij}\in L(x))^2 + \Omega(s_N^{-4})}{\PP_P\big(\PP_P(X_{ij}\in L(x)\given X_{ij})>s_N^{-2}\big)}
    \\
    &\sim
    \frac{1}{2} \sum_{i=1}^I \sum_{j=1}^{n_i}\frac{s_N^{-2}}{\PP_P\big(\PP_P(X_{ij}\in L(x)\given X_{ij})>s_N^{-2}\big)} 
    \geq
    \frac{1}{2} \sum_{i=1}^I \sum_{j=1}^{n_i}\frac{s_N^{-2}}{\PP_P\big(\PP_P(P_{ij}(x)=1\given X_{ij})>s_N^{-2}\big)}
    \\
    &\gtrsim
    \sum_{i=1}^I \sum_{j=1}^{n_i}\frac{s_N^{-2}}{s_N^{-1}k(\log s_N)^d} 
    \geq
    \tfrac{1}{2} N s_N^{-1}k^{-1}(\log s_N)^{-d}
    \sim I^{1-\beta}k^{-1}(\log(n_{\mathrm{c}}\vee I))^{-d}
    .
\end{align*}

In case~(ii) of Assumption~\ref{bounded-density}, because $\ind_{L(X_{ij})}(X_{ij'})=1$ for all $j,j'\in[n_i]$, and following similar arguments to case (i),
\begin{align*}
    \PP_P(X_{ij}\in L(x)\given X_{ij})^2
    &= \bigg(
    n_{\mathrm{c}}\E_P[\omega_{ij}(x)\given X_{ij}]+ n_{\mathrm{c}}\E_P\bigg[\sum_{i'\neq i}\ind_{L(X_{i1})}(X_{i'1})\omega_{i'1}(x)\bigggiven X_{i1}\bigg]
    \bigg)^2
    \\
    &\leq
    2n_{\mathrm{c}}^2 \E_P[\omega_{ij}(x)\given X_{ij}]^2 + \frac{4n_{\mathrm{c}}^2(k/n_{\mathrm{c}})^2C_1^2}{(s_I-n_{\mathrm{c}})^2}
    \\
    &\lesssim
    2n_{\mathrm{c}}^2 \E_P[\omega_{ij}(x)\given X_{ij}]^2 + \frac{4k^2C_1^2}{s_I^2},
\end{align*}
and so
\[
\sum_{j=1}^{n_{\mathrm{c}}}\E_P\big[\E_P[\omega_{ij}(x)\given X_{i}]^2\big]
\gtrsim
\frac{1}{2n_{\mathrm{c}}}\E_P\big[\PP_P(X_{i1}\in L(x)\given X_{i1})^2\big] - \frac{2k^2C_1^2}{s_N^2},
\]
and so for finite $n_{\mathrm{c}}$,
\begin{equation*}
    \sum_{i=1}^I\E_P\big[\|\E_P[\omega_i(x)\given X_i]\|_2^2\big]
    \gtrsim n_{\mathrm{c}}^{-1}k^{-1}Is_I^{-1}(\log s_I)^{-d} \sim n_{\mathrm{c}}^{-1}k^{-1}I^{1-\beta}(\log I)^{-d},
\end{equation*}
and the result follows in case (ii) by the same arguments.

\end{proof}

\begin{lemma}\label{lem:unif-rho}
    Suppose $(A_{P,I}(x):P\in\cP,x\in\cX,I\in\mathbb{N})$ and $(B_{P,I}(x):P\in\cP,x\in\cX,I\in\mathbb{N})$ are sequences of non-negative random variables satisfying $\sup_{x\in\cX}|\Var_P(A_{P,I}(x)\given B_{P,I}(x))-\omega(x)|=o_{\cP}(1)$ for some bounded $\omega(\cdot)$, and $\sup_{x\in\cX}\Var_P(A_{P,I}(x)\given B_{P,I}(x))\leq C$ almost surely for some constant $C$. Then $\sup_{P\in\cP}\sup_{x\in\cX}|\E_P[\Var_P(A_{P,I}(x)\given B_{P,I}(x))]-\omega(x)|=o(1)$. Moreover, if also $\sup_{P\in\cP}\sup_{x\in\cX}\Var_P(\E_P[A_{P,I}(x)\given B_{P,I}(x)])=o(1)$ then
    $$\sup_{P\in\cP}\sup_{x\in\cX}|\Var_P(A_{P,I}(x))-\omega(x)|=o(1).$$
\end{lemma}
\begin{proof}
    It suffices to show that for any sequence of random variables $(X_{P,I}(x):P\in\cP,x\in\cX,I\in\mathbb{N})$ if $\sup_{x\in\cX}|X_{P,I}(x)|=o_\cP(1)$ and $\sup_{x\in\cX}|X_{P,I}(x)|\leq C'$ almost surely for some constant $C'$, then $\sup_{P\in\cP}\sup_{x\in\cX}\E_P[|X_{P,I}|]=o(1)$. For arbitrary $P\in\cP$, $x\in\cX$, $I\in\mathbb{N}$, $\epsilon>0$,
    \begin{align*}
        \E_P[|X_{P,I}(x)|]
        &=
        \E_P\bigl[|X_{P,I}(x)|\ind_{(|X_{P,I}(x)|>\epsilon)}\bigr]+\epsilon
        \\
        &\leq
        \E_P\bigl[X_{P,I}^2(x)\bigr]^{1/2}\PP_P(|X_{P,I}(x)|>\epsilon)^{1/2} + \epsilon
        \\
        &\leq
        C'\,\PP_P(|X_{P,I}(x)|>\epsilon)^{1/2} + \epsilon.
    \end{align*}
    Taking suprema over $P\in\cP$, followed by limits as $I\to\infty$ and then $\epsilon\downarrow0$ gives the result.
\end{proof}

\section{Further details on numerical experiments of Section~\ref{sec:numericals}}

\subsection{Further results relating to Section~\ref{sec:sim2}}\label{appsec:sim2}

We give additional results for the simulation of Section~\ref{sec:sim2}. Theorem~\ref{thm:normality} shows asymptotic normality of the testing integrated mean function $\int\hat{\mu}^{\mathrm{MC}}_I(x)dQ(x)$ for a covariate distribution $Q$. In the simulation of Section~\ref{sec:sim2} we consider mean squared prediction error under the covariate shifted distribution $Q=\text{Unif}\hspace{0.08em}[1,2]$. In this section we also show results for estimation of the integrated mean function quantity $\int\mu(x)dQ(x)$. Figure~\ref{fig:sim2-qqplots} presents QQ plots for the estimator $\int\hat{\mu}^{\mathrm{MC}}_I(x)dQ(x)$ for each of the RF and CRF.

\begin{table}[ht]
	\begin{center}
		\begin{tabular}{c|cc}
		\toprule
            Method & $\big(\int\hat{\mu}(x)dQ(x) - \int\mu(x)dQ(x)\big)^2$ ($\times10^{-5}$) & $\int (\hat{\mu}(x)-\mu(x))^2dQ(x)$ ($\times10^{-4}$)
            \\
		\midrule
            RF & 3.36 & 3.16
            \\
            CRF & {\bf 0.94}
            & {\bf 1.10}
            \\
            TRAIN & 3.36 & 7.07
            \\
            REEM & 1.72 & 7.06
            \\
		\bottomrule
		\end{tabular}
    \caption{Further results of Simulation~\ref{sec:sim2} (1000 simulations). Here $Q=\text{Unif}\hspace{0.08em}[1,2]$.}\label{tab:sim2-extra}
	\end{center}
\end{table}

\subsection{Further results relating to Section~\ref{sec:sim3}}\label{appsec:sim3}

Here we provide further results for the simulation of Section~\ref{sec:sim3}. Figure~\ref{fig:sim2-qqplots} presents Gaussian QQ plots for the predictions $\hat{\mu}({\bf1})$ studied in Section~\ref{sec:sim3}. In particular, we see the asymptotic Gaussianity of predictions persists for covariate dimensions up to $d=50$ for both standard (RF) and clustered (CRF) random forests. For all covariate dimensions considered we see evidence of normality of $\hat{\mu}({\bf1})$. Figure~\ref{fig:sim2-cov-plots} shows nominal $(1-\tilde{\alpha})$-level confidence interval coverage for $\tilde{\alpha}\in(0,1)$, and indeed for all $\tilde{\alpha}\in(0,1)$ we see valid coverage of both RFs and CRFs.

\subsection{Finite sample approximation to Monte-Carlo random forests}\label{appsec:sim2-B}

Our theoretical results of Section~\ref{sec:theory} hold for Monte-Carlo random forests~\eqref{eq:MC}, which may be approximated by the random forests of Algorithm~\ref{alg:crf} (or the analogous Algorithm~\ref{alg:dcrf}) for a sufficiently large number of trees as given by the number of trees $B$ in the forest. Figure~\ref{fig:B} shows the MSPE results of Section~\ref{sec:sim2}; in particular a moderate choice of $B=500$ as we present in Section~\ref{sec:sim2} is sufficiently large to approximate the corresponding Monte-Carlo limit.

\begin{figure}[hb]
    \centering
    \includegraphics[width=0.95\linewidth]{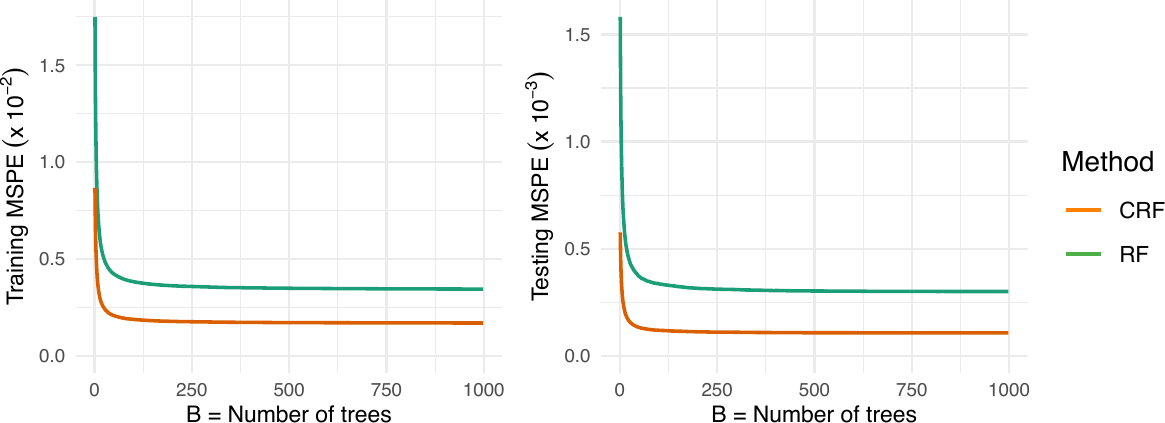}
    \caption{MSPE of standard and clustered random forests (Algorithm~\ref{alg:crf}) for a varying number of trees $B\in[1,1000]$ making up the random forests.}
    \label{fig:B}
\end{figure}

\begin{figure}[p]
    \centering
    \includegraphics[width=1.00\linewidth]{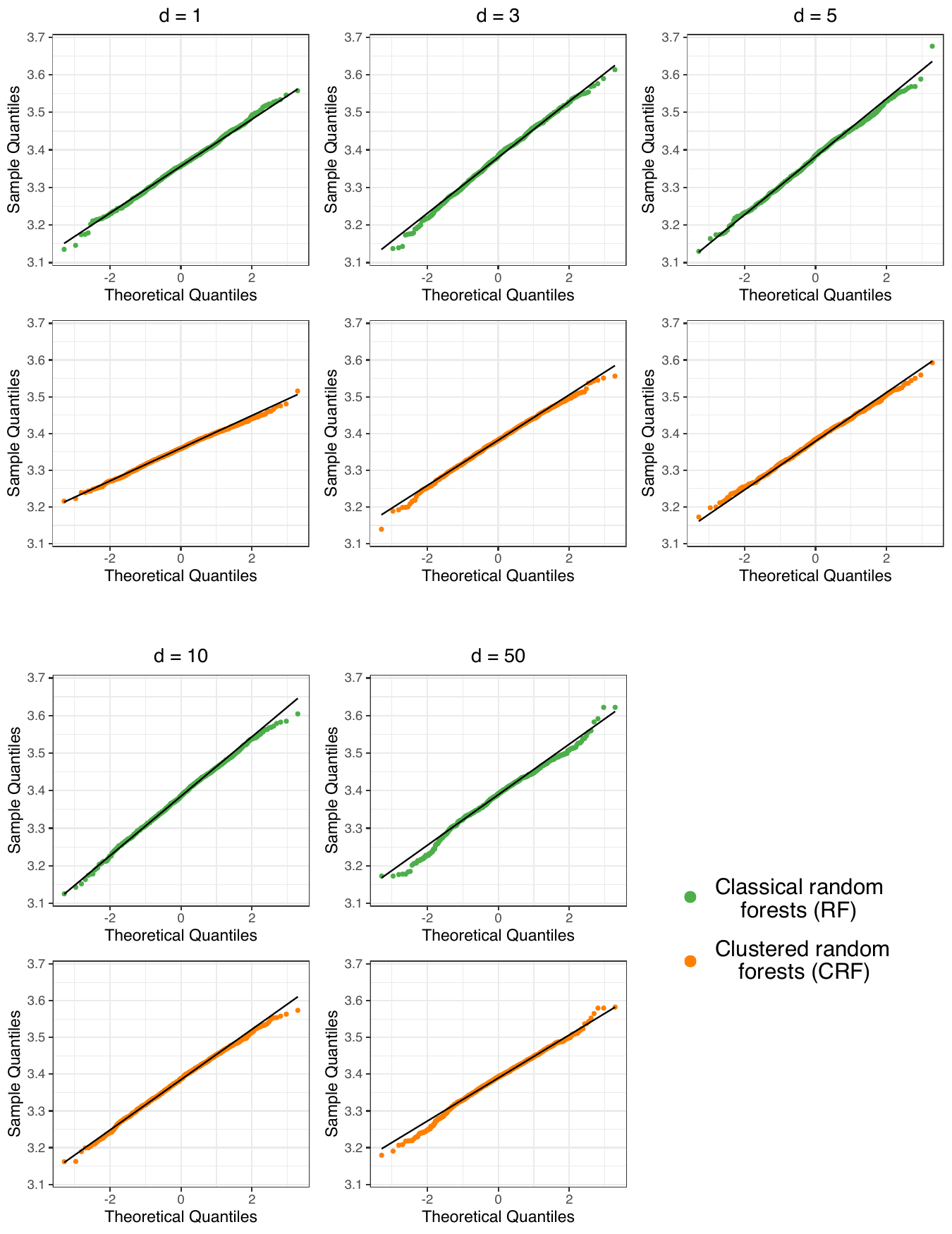}
    \caption{Gaussian QQ plots for predictions $\hat{\mu}({\bf1})$ in the simulations of Section~\ref{sec:sim3}, for each covariate dimension $d\in\{1,3,5,10,50\}$.}
    \label{fig:sim2-qqplots}
\end{figure}

\begin{figure}[p]
    \centering
    \includegraphics[width=\linewidth]{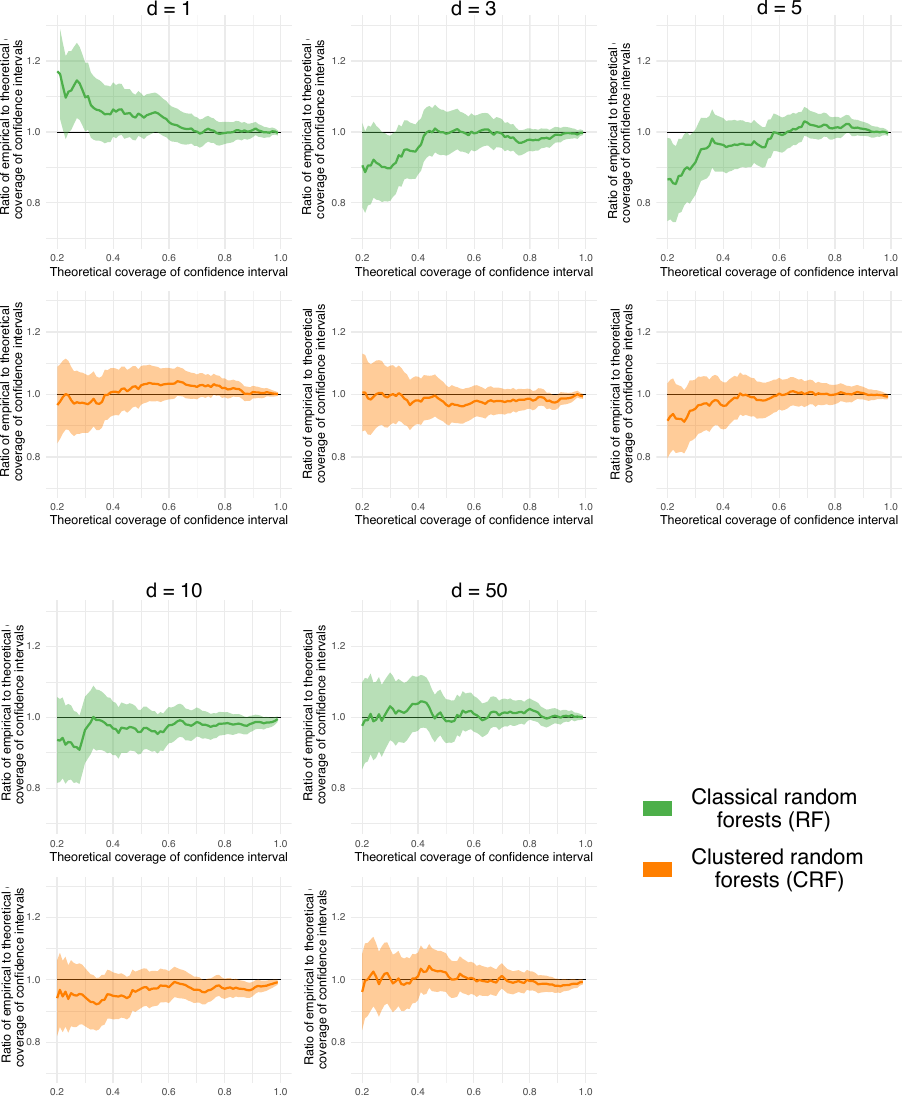}
    \caption{Coverage of nominal confidence intervals over levels $\tilde{\alpha}\in[0.2,1)$ for the simulations of Section~\ref{sec:sim3}, for each covariate dimension $d\in\{1,3,5,10,50\}$. Plots show the ratio of the empirical to theoretical coverage of the RF and CRF confidence intervals for $\mu({\bf1})$, with points above the black horizontal line indicating overcoverage, and below undercoverage.}
    \label{fig:sim2-cov-plots}
\end{figure}

\section{Proof of Proposition~\ref{prop:linear-time}}

The algorithm proposed for Proposition~\ref{prop:linear-time} simply performs conjugate gradient descent to evaluate the quantity $(\chi^\top W(\rho)\chi)^{-1}(\chi^\top W(\rho)Y)$ in case~\ref{item:comp-1} and in case~\ref{item:comp-2} calculates $(\chi^\top W(\rho)\chi)^{-1}\mathrm{e}_m$ for all $m\in[M]$ satisfying $L_m\cap\, \text{supp}\,Q\neq\emptyset$. The conjugate gradient descent method is known to be fast when calculating $A^{-1}b$ where $A$ is a sparse matrix. In our case $\chi^\top W(\rho)\chi$ is in general a dense matrix, although for certain weight classes we can decompose $\chi^\top W(\rho)\chi$ in such a way the conjugate gradient descent iterations can be perform in a fast manner. 

\begin{proof}[Proof of Proposition~\ref{prop:linear-time}]

    The number of conjugate gradient descent iterations required for an $\epsilon>0$ accurate solution in the Euclidean norm is $O\Bigl(\sqrt{\kappa(\chi^\top W(\rho)\chi)}\log(\sqrt{\kappa(\chi^\top W(\rho)\chi)}/\epsilon)\Bigr)$, where $\kappa$ denotes the condition number. Because $\kappa\bigl(\chi^\top W(\rho)\chi\bigr) 
    \leq
    \frac{C_W(2k-1)}{c_Wk}\leq 2c_W^{-1}C_W$ the number of conjugate gradient iterations required is $O\bigl(\log(1/\epsilon)\bigr)$. 

    \medskip\noindent{\bf Equicorrelated working correlation weights: }

For the equicorrelated weight structure
\begin{equation*}
    W_i(\rho)=\frac{1}{1-\rho}\biggl(\mathrm{I}_{n_i}-\frac{\rho}{1-\rho}\boldsymbol{1}_{n_i}\boldsymbol{1}_{n_i}^\top\biggr),
\end{equation*}
where
\begin{equation*}
    \chi^\top W(\rho)\chi 
    = 
    \sum_{i=1}^I\chi_i^\top W_i(\rho)\chi_i
    =
    \frac{1}{1-\rho}\biggl(\chi^\top\chi-\frac{\rho}{1-\rho}\sum_{i=1}^I\chi_i^\top\boldsymbol{1}_{n_i}\boldsymbol{1}_{n_i}^\top \chi_i\biggr),
\end{equation*}
Define $v\in\R^M$ via $v_m:=|\{(i,j):i\in[I],j\in[n_i], X_{ij}\in L_m\}|$ for $m\in[M]$, and for each $i\in[I]$ define $u_i\in\R^M$ via $(u_i)_m:=|\{j\in[n_i]:X_{ij}\in L_m\}|$ for $m\in[M]$. Note therefore $\chi^\top\chi=\diag(v)$ and $k_I\leq v_m\leq 2k_I-1$ for each $m\in[M]$, and also that $\|u_i\|_0 \leq n_i\leq 2k_I-1$ is sparse. 
Then for any $b\in\R^M$ the number of iterations it takes to calculate
    \begin{align*}
        (\chi^\top W(\rho)\chi)b
        &=
        \frac{1}{1-\rho}\biggl(\chi^\top\chi-\frac{\rho}{1-\rho}\sum_{i\in\cI_{\mathrm{eval}}}u_iu_i^\top\biggr)b
        \\
        &=
        \frac{1}{1-\rho}\biggl(\diag(v_mb_m:m\in[M])-\frac{\rho}{1-\rho}\sum_{i\in\cI_{\mathrm{eval}}}(u_i^\top b)u_i\biggr),
    \end{align*}
    is~$O(n_{\mathrm{c}}s_I)$ because calculating $\diag(v_mb_m:m\in[M])$ takes $O(M)$ operations, calculating $u_i^\top b$ takes $O(n_i)$ operations as $\|u_i\|_0\leq n_i$, and so calculating $\sum_{i\in\cI_{\mathrm{eval}}}(u_i^\top b)u_i$ takes $O(n_{\mathrm{c}}s_I)$ operations. Together therefore calculating $(\chi^\top W(\rho)\chi)b$ takes $O(n_{\mathrm{c}}s_I)$ operations. Thus calculating $(\chi^\top W(\rho)\chi)^{-1}z$ for some vector $z\in\R^M$ up to order $\epsilon>0$ accuracy in Euclidean norm is
    $O(n_{\mathrm{c}}s_I\log(\epsilon^{-1}))$

    For fitted values, calculating 
    \begin{align*}
    \chi^\top W(\rho)Y
    &=
    \frac{1}{1-\rho}\biggl(\chi^\top Y - \frac{\rho}{1-\rho}\sum_{i\in\cI_{\mathrm{eval}}}u_i(u_i^\top Y)
    \biggr)
    \\
    &=
    \frac{1}{1-\rho}\biggl[
    \biggl(\sum_{(i,j):X_{ij}\in L_m}Y_{ij}:m\in[M]\biggr)
    -
    \frac{\rho}{1-\rho}\sum_{i=1}^I\biggl\{\sum_{(i',j'):u_{i'j'}\neq0}u_{i'j'}Y_{i'j'}\biggr\}u_i
    \biggr],
    \end{align*}
    which is an order $O(n_{\mathrm{c}}s_I+\sum_{i\in\cI_{\mathrm{eval}}}\|u_i\|_0)=O(n_{\mathrm{c}}s_I)$ time calculation. 
    Therefore calculating $(\chi^\top W(\rho)\chi)^{-1}(\chi^\top W(\rho)Y)$ via conjugate gradient descent, to $\epsilon$ accuracy is $O(n_{\mathrm{c}}s_I\log(\epsilon^{-1}))$. Case (ii) follows by analogous arguments.

    \medskip\noindent{\bf AR$\boldsymbol{(1)}$ working correlation weights: }

    As each row of $\chi_i$ has precisely one non-zero value, calculating $\chi_ib$ requires $n_i$ operations, with $\|\chi_ib\|_0=n_i$. Because the inverse of an $\text{AR}(1)$ covariance matrix is tri-diagonal,
    each non-zero entry of $\chi_ib$ contributes to $W_i(\rho)\chi_ib$ via at most 3 multiplications with entries of $W_i(\rho)$, thus $W_i(\rho)\chi_ib$ can be computed in at most $3n_i$ operations, and moreover $\|W_i(\rho)\chi_ib\|_0\leq 3n_i$. 
    As precisely one entry of each row of $\chi_i$ is non-zero, $\chi_i^\top W_i(\rho)\chi_ib$ can be calculated in $O(n_i)$ operations, and so $\sum_{i=1}^I\chi_i^\top W_i(\rho)\chi_ib$ can be calculated in $O(n_{\mathrm{c}}s_I)$ operations. The results then follow by identical arguments to those used in the equicorrelated case. 
    
\end{proof}

\end{document}